\documentclass[twocolumn,10pt,superscriptaddress,prl,aps]{revtex4-2}
\pdfoutput=1

\usepackage{times}
\usepackage[T1]{fontenc}
\usepackage{microtype}
\usepackage[utf8]{inputenc}
\usepackage{amsmath,amsthm,amsfonts,amsbsy,amssymb}
\usepackage{newtxmath,mathrsfs,bbold,bbm,dsfont,pifont}
\usepackage{enumitem}
\usepackage{graphicx,float}
\usepackage[caption=false]{subfig}
\usepackage{array,longtable,booktabs,makecell}
\usepackage{tikz,xcolor}
\usetikzlibrary{arrows.meta,shadows.blur}
\usepackage{verbatim,algpseudocode,listings}
\usepackage{csquotes}
\usepackage{hyperref}
\usepackage{lipsum}

\newtheorem{theorem}{Theorem}
\newtheorem{corollary}[theorem]{Corollary}

\newtheorem{lemma}[theorem]{Lemma}
\newtheorem{proposition}[theorem]{Proposition}

\theoremstyle{definition}
\newtheorem{definition}{Definition}
\newtheorem{remark}{Remark}
\newtheorem{example}{Example}

\newcommand{\1}{\mathbb{1}}
\newcommand{\pd}{\mathscr{P}_+}
\newcommand{\psd}{\mathscr{P}}

\newcommand{\sta}{\mathscr{D}}
\newcommand{\sub}{\mathscr{D}_\leq}
\newcommand{\supp}{\mathrm{supp}}
\newcommand{\tr}{\mathrm{Tr}}

\newcommand{\abb}[1]{\textnormal{#1}} % abbreviation
\newcommand{\ch}[1]{\mathcal{#1}} % channel
\newcommand{\en}[1]{\mathscr{#1}} % ensemble
\newcommand{\h}[1]{\hat{#1}} % hamiltonian
\newcommand{\pz}[1]{\widebar{#1}} % Petz-type Rényi
\newcommand{\spa}[1]{\mathds{#1}} % space
\newcommand{\sw}[1]{\widetilde{#1}} % sandwiched Rényi

\newcommand{\bra}[1]{\langle#1\rvert}
\newcommand{\ket}[1]{\lvert#1\rangle}

\newcommand{\op}[2]{\ket{#1}\!\bra{#2}}

\newcommand{\mclose}{\mathclose{}}
\newcommand{\fleft}{\mathopen{}\left}
\newcommand{\fright}{\aftergroup\mclose\right}

\definecolor{darkblue}{rgb}{0,0,0.6}
\definecolor{darkgreen}{rgb}{0,0.45,0.1}
\definecolor{darkred}{rgb}{0.5,0,0}
\definecolor{primaryred}{RGB}{237,139,150}
\definecolor{memoryblue}{RGB}{126,178,229}
\definecolor{primarygreen}{RGB}{150,208,128}
\hypersetup{
	colorlinks = true,
	citecolor = darkgreen, % links to bibliography
	linkcolor = darkblue, % internal links
	urlcolor  = darkblue, % external links
	filecolor = darkred % file links
}

\allowdisplaybreaks
\setlist[itemize]{noitemsep,topsep=0pt}

\begin{document}

%-------------------------------
%	TITLE
%-------------------------------

\title{
Fundamental Limits for Thermodynamic Control with Quantum Feedback
}

\author{Kaiyuan Ji}
\email{kj264@cornell.edu}
\affiliation{School of Electrical and Computer Engineering, Cornell University, Ithaca, New York 14850, USA}

\author{Gilad Gour}
\email{giladgour@technion.ac.il}
\affiliation{Faculty of Mathematics, Technion - Israel Institute of Technology, Haifa 3200003, Israel}

\author{Mark M. Wilde}
\email{wilde@cornell.edu}
\affiliation{School of Electrical and Computer Engineering, Cornell University, Ithaca, New York 14850, USA}

\date{\today}

%-------------------------------
%	ABSTRACT
%-------------------------------

\begin{abstract}
The study of feedback control inspired by Maxwell's demon has yielded fundamental insights into the relationship between thermodynamics and information.  However, its fully quantum formulation, where the controller holds quantum side information that cannot be treated as a classical record, remains incomplete.  In this paper, we establish fundamental no-go limits for thermodynamic state conversion while allowing quantum side information to be fed back to the system coherently.  From two basic operational principles respected by all physically admissible feedback-control schemes, we derive the tightest possible bounds on the single-shot work of formation and extractable work for arbitrary quantum systems conditioned on arbitrary quantum side information.  In the asymptotic limit, these bounds yield a generalized second law of thermodynamics with quantum feedback, governed by a conditional free energy.  We further show that this generalized second law is consistent with the traditional second law once the acquisition, processing, and reset of side information are accounted for.  These results extend information thermodynamics to the fully quantum regime and provide precise thermodynamic meanings for the negativity of single-shot conditional entropies.
\end{abstract}

%-------------------------------
%	MAJOR CHANGES
%-------------------------------

%-------------------------------
%	MAINTEXT
%-------------------------------

\maketitle
\let\oldaddcontentsline\addcontentsline
\renewcommand{\addcontentsline}[3]{}

%\tableofcontents

\textbf{\textit{Introduction}}---The second law of thermodynamics was once challenged by Maxwell's demon~\cite{maxwell1911QuoteUndatedLetter}, an intelligent being that possesses side information about a system and can thereby extract more work from it than the second law apparently allows.  It was later clarified that the demon does not actually violate the second law~\cite{bennett1982ThermodynamicsComputationReview}, as the excess work must be offset by work expended in other stages of a closed cycle, such as resetting the demon's memory to its initial state~\cite{landauer1961IrreversibilityHeatGeneration,lloyd1989UseMutualInformation,bennett2003NotesLandauersPrinciple,maruyama2009ColloquiumPhysicsMaxwells,deoliveirajunior2025FriendlyGuideExorcising}.

Maxwell's demon inspired the formal study of \emph{feedback control} in thermodynamics, and the ensuing theoretical~\cite{sagawa2008SecondLawThermodynamics,sagawa2009MinimalEnergyCost,cao2009ThermodynamicsFeedbackControlled,jacobs2009SecondLawThermodynamics,sagawa2012FluctuationTheoremInformation,parrondo2015ThermodynamicsInformation,prech2024QuantumFluctuationTheorem,minagawa2025UniversalValiditySecond} and experimental~\cite{camati2016ExperimentalRectificationEntropy,cottet2017ObservingQuantumMaxwell,masuyama2018InformationtoworkConversionMaxwells,naghiloo2018InformationGainLoss,yada2025ExperimentallyProbingEntropy} advances have yielded significant insights into the interplay between thermodynamics and information.  In these studies, a controller playing the role of Maxwell's demon carries out a cycle of three stages: (i) measuring the system and storing the outcome as side information in a memory, (ii) applying a control operation to the system conditioned on the measurement outcome, and (iii) resetting the memory to its initial state.  The apparent violation of the second law occurs during the feedback-control stage, (ii), where, owing to the feedback of side information, the average work $\langle W\rangle$ consumed by the system is no longer lower bounded by the system's change in nonequilibrium free energy.  Instead, it is now subject to a so-called \emph{second law of thermodynamics with (classical) feedback}~\cite{parrondo2015ThermodynamicsInformation}:
\begin{align}
	\left\langle W\right\rangle&\geq\Delta F_\abb{sys}+\beta^{-1}\Delta I_\abb{c}, \label{eq:classical}
\end{align}
where $\Delta F_\abb{sys}$ is the system's change in nonequilibrium free energy, $\beta$ the constant inverse temperature, and $\Delta I_\abb{c}$ the change in classical mutual information between the system and the controller's memory.  The last term in Eq.~\eqref{eq:classical} captures the work advantage provided by the feedback relative to the traditional second law; yet this advantage must be offset in the measurement stage, (i), and the reset stage, (iii), thus ensuring the consistency between Eq.~\eqref{eq:classical} and the traditional second law over a closed cycle.  This consistency is a cornerstone of information thermodynamics~\cite{sagawa2009MinimalEnergyCost,parrondo2015ThermodynamicsInformation}.

%The significance of Eq.~\eqref{eq:classical} is both conceptual and practical.  Conceptually, it clarifies the value of information in thermodynamics by quantifying how much the traditional second law must be modified when side information is exploited through feedback control.  In practice, it separates the work expended on controlling the system from that on measurement or reset, thus providing no-go limits for experimenters who wish to focus only on the former without bookkeeping of the latter.

Thermodynamic feedback control has also been studied for quantum systems~\cite{sagawa2008SecondLawThermodynamics,jacobs2009SecondLawThermodynamics,minagawa2025UniversalValiditySecond}, in which case the controller collects side information beforehand by a weak or indirect measurement and then applies a quantum control operation to the system conditioned on the measurement outcome.  However, such an approach has a severe limitation: even if such a measurement avoids decoherence of the system, it renders the side information and the subsequent feedback \emph{classical}.  As this excludes the possibility of creating and exploiting quantum correlations between the system and the controller's memory, thermodynamic laws derived in this manner are not intrinsically universal: they do not apply to a fully quantum controller who is capable of acquiring, processing, and feeding back quantum side information coherently.

In this paper, we develop a resource-theoretic framework for thermodynamic control of quantum systems with \emph{quantum} feedback, and we address the following question: how much work is fundamentally required to transform a system from one state to another with the most general form of feedback available from a quantum memory?  We establish no-go limits, i.e., lower bounds on such work costs, assuming only the physical admissibility of the feedback-control schemes.  Our bounds characterize both microscopic (deterministic) and macroscopic (average) work.  Notably, our Theorems~\ref{thm:formation} and \ref{thm:extraction} are \emph{axiomatically tight} even in the \emph{single-shot} regime, in the sense that they cannot be improved without additional operational assumptions.  Detailed derivations of our results can be found in Supplemental Material (SM)~\footnote{See Supplemental Material for proofs and additional discussions, which includes Refs.~\cite{petz1986SufficientSubalgebrasRelative,petz1988SufficiencyChannelsNeumann,wilde2015RecoverabilityQuantumInformation,marvian2016ClocksClonersCatalytic,lemm2017InformationtheoreticLimitationsApproximate,lami2018ApproximateReversalQuantum,alhambra2018WorkReversibilityQuantum,wang2019ResourceTheoryAsymmetric,mosonyi2011QuantumRenyiRelative,holevo1973InformationtheoreticalAspectsQuantum}.} and the companion paper~\cite{ji2026QuantumConditionalRelative}.

The notion of quantum feedback was first proposed in quantum control theory by Lloyd~\cite{lloyd2000CoherentQuantumFeedback} and has been demonstrated experimentally~\cite{serafini2012FeedbackControlQuantum}.  Unlike controllers providing classical feedback~\cite{wiseman1993QuantumTheoryOptical,nielsen1998InformationtheoreticApproachQuantum}, a controller providing quantum feedback can preserve the coherence of quantum side information and feed it back to the system coherently~\cite{lloyd2000CoherentQuantumFeedback}.  Here we extend this notion to quantum thermodynamics, noting that such an ability was not envisioned in previous semiquantum generalizations of Maxwell's demon~\cite{lloyd1997QuantummechanicalMaxwellsDemon,kim2011QuantumSzilardEngine,cottet2017ObservingQuantumMaxwell,naghiloo2018InformationGainLoss} or feedback-control-like processes in thermodynamic contexts~\cite{sagawa2008SecondLawThermodynamics,sagawa2009MinimalEnergyCost,jacobs2009SecondLawThermodynamics,narasimhachar2017ResourceTheoryConditioned,narasimhachar2019QuantifyingMemoryCapacity,morris2019AssistedWorkDistillation,prech2024QuantumFluctuationTheorem,minagawa2025UniversalValiditySecond}.

Apart from being mathematically more general, incorporating quantum feedback has specific physical implications.  First, compared with schemes with classical feedback, schemes with quantum feedback are substantially more powerful.  As recognized in Ref.~\cite{lloyd2000CoherentQuantumFeedback}, some tasks in quantum control theory, such as state exchange and entanglement transfer, are straightforward to implement with quantum feedback but impossible with classical feedback.  As we will demonstrate, in general and with concrete examples, the ability to provide quantum feedback is also indispensable to achieving the ultimate thermodynamic limit of state formation and work extraction.  Because of this strict separation between quantum and classical feedback, only bounds derived assuming quantum feedback qualify as fundamental, model-independent no-go limits applicable to all physically admissible feedback-control schemes.

Second, quantum side information is fundamentally different from classical side information, both information-theoretically and operationally.  In particular, conditional entropy can be negative for certain entangled states~\cite{cerf1997NegativeEntropyInformation}, an impossibility with classical side information alone.  This negativity challenged the traditional understanding of information~\cite{horodecki2005PartialQuantumInformation} and has since acquired various operational meanings, in contexts from communication~\cite{horodecki2005PartialQuantumInformation,horodecki2006QuantumStateMerging} and the uncertainty principle~\cite{berta2010UncertaintyPrinciplePresence,tomamichel2011UncertaintyRelationSmooth,coles2012UncertaintyRelationsSimple} to thermodynamics~\cite{delrio2011ThermodynamicMeaningNegative}.  Crucially, the contrast between negative and positive conditional entropy, and thus that between quantum and classical side information, is qualitative, not merely quantitative, as it may signify the distinction between certain resources being generated rather than consumed in a protocol.  We will demonstrate similar contrasts in our findings, such as erasure accompanied by cooling instead of heating in the single-shot regime, thereby providing precise thermodynamic meanings for single-shot conditional entropies and their negativity.  Such phenomena require entanglement to be exploited without being broken, which renders the use of quantum feedback indispensable.

\textbf{\textit{Framework}}---Our operational framework is cast as a resource theory of thermal nonequilibrium in the presence of quantum side information, or briefly, \emph{quantum conditional athermality}.  A semiclassical specialization of our resource theory was considered in Ref.~\cite{narasimhachar2017ResourceTheoryConditioned}, whereas here our setup is fully quantum.

The formalism of resource theories~\cite{horodecki2013QuantumnessContextResource,chitambar2019QuantumResourceTheories,gour2025QuantumResourceTheories} has been instrumental in the study of quantum thermodynamics~\cite{janzing2000ThermodynamicCostReliability,horodecki2003ReversibleTransformationsPure,horodecki2013FundamentalLimitationsQuantum,brandao2013ResourceTheoryQuantum,gour2015ResourceTheoryInformational,lostaglio2019IntroductoryReviewResource}, especially on the microscopic scale~\cite{dahlsten2011InadequacyNeumannEntropy,aberg2013TrulyWorklikeWork,horodecki2013FundamentalLimitationsQuantum,brandao2015SecondLawsQuantum,faist2015MinimalWorkCost,faist2018FundamentalWorkCost,muller2018CorrelatingThermalMachines,shiraishi2021QuantumThermodynamicsCorrelatedcatalytic,gour2022RoleQuantumCoherence}, where the number of copies of a system is not assumed to be arbitrarily large and finite-size effects such as work fluctuations become nonnegligible~\cite{dahlsten2011InadequacyNeumannEntropy,aberg2013TrulyWorklikeWork,horodecki2013FundamentalLimitationsQuantum,alhambra2016FluctuatingWorkQuantum}.  The most general approach is to consider a single copy of a system, known in information theory as the single-shot regime~\cite{renner2005SecurityQuantumKey,tomamichel2012FrameworkNonasymptoticQuantum,khatri2024PrinciplesQuantumCommunication}.  Macroscopic laws of thermodynamics can then be deduced in the asymptotic limit of many independent copies.  Here, our resource theory provides the precise machinery necessary for the single-shot analysis of thermodynamic feedback control, which, to the best of our knowledge, has been absent from previous studies~\cite{sagawa2008SecondLawThermodynamics,sagawa2009MinimalEnergyCost,cao2009ThermodynamicsFeedbackControlled,jacobs2009SecondLawThermodynamics,sagawa2012FluctuationTheoremInformation,parrondo2015ThermodynamicsInformation,narasimhachar2019QuantifyingMemoryCapacity,morris2019AssistedWorkDistillation,prech2024QuantumFluctuationTheorem,minagawa2025UniversalValiditySecond} except in a fully classical setup~\cite{narasimhachar2017ResourceTheoryConditioned}.  The scope and objective here are also different from those of Ref.~\cite{faist2018FundamentalWorkCost}.  The latter work concerns the single-shot work cost of logical processes and provides a criterion for the applicability of thermodynamics under coarse graining, thereby helping elucidate why the traditional second law fails in Maxwell's-demon-like scenarios.  In contrast, our work aims to extend the traditional framework of quantum thermodynamics to accommodate Maxwell's-demon-like scenarios, in particular by incorporating quantum side information and quantum feedback on an operational level.

We define the free operations of our resource theory based on two operational principles.  First, every free operation represents a feedback-control scheme, and so it is naturally composed of local operations on the system and the controller's memory, combined with one-way communication from the memory to the system.  Denoting the initial (resp.\ final) system by $S$ (resp.\ $S'$) and the initial (resp.\ final) memory by $M$ (resp.\ $M'$), a free operation is thus represented by a quantum channel $\ch{N}_{SM\to S'M'}$ with the following structural decomposition:
\begin{align}
	\ch{N}_{SM\to S'M'}&=\ch{F}_{SL\to S'}\circ\ch{E}_{M\to LM'}, \label{eq:feedback}
\end{align}
where $\ch{E}_{M\to LM'}$ and $\ch{F}_{SL\to S'}$ are quantum channels.  Since we are now concerned with the feedback-control stage of a cycle, the free operation accounts only for the exploitation of side information, not its acquisition, and thus no communication from the system to the memory is needed.  Notably, since free operations are supposed to encompass all physically admissible feedback-control schemes, including those with quantum feedback, the communication through the link $L$ from the memory to the system is in general quantum, in contrast to previous studies~\cite{sagawa2008SecondLawThermodynamics,sagawa2009MinimalEnergyCost,cao2009ThermodynamicsFeedbackControlled,jacobs2009SecondLawThermodynamics,sagawa2012FluctuationTheoremInformation,parrondo2015ThermodynamicsInformation,narasimhachar2017ResourceTheoryConditioned,narasimhachar2019QuantifyingMemoryCapacity,morris2019AssistedWorkDistillation,prech2024QuantumFluctuationTheorem,minagawa2025UniversalValiditySecond}.

Second, we require that the control operation, represented by $\ch{F}_{SL\to S'}$, respect Kelvin's principle: if the feedback provided through $L$ carries no side information about the system, then the control operation should not be able to drive the system away from thermal equilibrium.  Denoting the Gibbs state of $S$ (resp.\ $S'$) by $\gamma_S\equiv e^{-\beta\h{H}_S}/\tr[e^{-\beta\h{H}_S}]$ (resp.\ $\gamma'_{S'}$) with $\h{H}_S$ its Hamiltonian, this means that
\begin{align}
	\ch{F}_{SL\to S'}\fleft[\gamma_S\otimes\left(\cdot\right)_L\fright]&=\gamma'_{S'}\otimes\tr_L\fleft[\cdot\fright]. \label{eq:kelvin}
\end{align}
That is, the system undergoes a Gibbs-preserving map~\cite{lostaglio2015DescriptionQuantumCoherence,faist2015GibbspreservingMapsOutperform} if the content of the feedback is independent of the system.  This ensures that quantum feedback, like classical feedback~\cite{sagawa2008SecondLawThermodynamics,sagawa2009MinimalEnergyCost,cao2009ThermodynamicsFeedbackControlled,jacobs2009SecondLawThermodynamics,sagawa2012FluctuationTheoremInformation,parrondo2015ThermodynamicsInformation,narasimhachar2017ResourceTheoryConditioned,narasimhachar2019QuantifyingMemoryCapacity,morris2019AssistedWorkDistillation,prech2024QuantumFluctuationTheorem,minagawa2025UniversalValiditySecond}, can only act as a source of information and not a source of work in controlling the system.

We call a channel $\ch{N}_{SM\to S'M'}$ satisfying the above two operational principles [see Eqs.~\eqref{eq:feedback} and \eqref{eq:kelvin}] a \emph{quantum-feedback-assisted Gibbs-preserving operation (QFGO)}, which we deem free in our resource theory.  Both principles are upheld by all physically admissible feedback-control schemes without external work supply, and thus QFGOs are axiomatic relaxations of such schemes.  Since QFGOs are maximally permissive, the minimum work cost under QFGOs lower bounds the actual work consumed by the system in any practical setup, thus serving as a universal, model-independent no-go limit.

We provide a simple necessary and sufficient condition for deciding whether a channel is a QFGO; see Appendix~A for a formal statement and implications.

\begin{proposition}[Informal]
\label{prop:operation}
A channel $\ch{N}_{SM\to S'M'}$ is a QFGO [i.e., satisfies Eqs.~\eqref{eq:feedback} and \eqref{eq:kelvin}] if and only if it acts covariantly with respect to thermalization of the system.
\end{proposition}

To quantify the work cost of a transformation under a QFGO, we let the system exchange work with a battery, denoted by $S_0$ (resp.\ $S_1$) before (resp.\ after) the transformation.  During the transformation, the battery remains in a pure (or nearly pure) state with a trivial Hamiltonian, but its dimension changes from $\lvert S_0\rvert$ to $\lvert S_1\rvert$~\cite{faist2015MinimalWorkCost}.  By Landauer's principle~\cite{landauer1961IrreversibilityHeatGeneration}, the battery's charge (i.e., useful free energy) changes from $\beta^{-1}\ln\lvert S_0\rvert$ to $\beta^{-1}\ln\lvert S_1\rvert$~\footnote{A battery $S_0$ in the pure state $\op{0}{0}_{S_0}$ with a trivial Hamiltonian $\h{H}_{S_0}=0$ is deterministically interconvertible with $\beta^{-1}D(\op{0}{0}_{S_0}\|\1_{S_0}/\lvert S_0\rvert)=\beta^{-1}\ln\lvert S_0\rvert$ work, just as a two-level Szilárd box~\cite{szilard1929UeberEntropieverminderungThermodynamischen} can store $\beta^{-1}\ln2$ work.}.  With the QFGO controlling both the battery and the system while receiving quantum feedback from the memory, the work consumed by the system then equals the battery's decrease in charge, $\beta^{-1}(\ln\lvert S_0\rvert-\ln\lvert S_1\rvert)$; see Fig.~\ref{fig:work}.  Alternatively, if the battery's charge increases, then $\beta^{-1}(\ln\lvert S_1\rvert-\ln\lvert S_0\rvert)$ work is extracted from the system.

\begin{figure}[t]
\begin{tikzpicture}[x=0.1612pt,y=-0.1612pt,
	wire/.style={line width=1.3pt,line cap=round,line join=round},
	arr/.style={wire,-{Straight Barb[length=3pt,width=5.2pt,line width=1.3pt]}},
	wbend/.style={rounded corners=4pt},
	lab/.style={inner sep=1pt,font=\scriptsize},
	big/.style={inner sep=0pt,font=\small},
	thk/.style={line width=1.1pt},
	spinarr/.style={line width=1.1pt,-{Triangle[length=3.6pt,width=4pt]}}]
% ================================================================ geometry (1 px = 0.1612 pt; icons are the originals scaled by \k)
% rows: in-box spacing 46.7 px = 7.5 pt, box margin 28.4 px, dashed-box margins 23.4 px = 3.8 pt (vertical) and 56.7 px = 9.1 pt (horizontal),
% L: stubs 75.1 px = 12.1 pt as in Fig. 2; 45-degree diagonal of rise 56 px (Fig. 2: 40 px), giving S and M more room; boxes 105 px = 16.9 pt wide
\def\yO{197}\def\yS{253}\def\yLF{299.7}\def\yLE{355.7}\def\yM{402.4}
\def\mg{28.4}\def\dm{23.4}\def\dh{56.7}\def\st{75.1}\def\wE{105}\def\wF{105}\def\lo{48.0}
\def\k{0.72}
\def\kc{0.64}\def\ky{0.62}                  % vertical scales of the chamber and of the cloud (flatter than the original)
\pgfmathsetmacro{\rs}{\yLE-\yLF}
% ---- left column: gear centre (\G,\yG); battery; demon
\def\G{40}
\pgfmathsetmacro{\yG}{\yS+18*\k}                 % rod runs through the upper part of the gear, as in the original
\def\kl{0.6}                                     % lightning scale (smaller, to bring the battery closer to the gear)
\pgfmathsetmacro{\yL}{\yG-61*\k+8*\k-53*\kl}     % lightning centre: its tip reaches into the gear
\pgfmathsetmacro{\yB}{\yL-53*\kl-6*\k-29*\k}     % battery centre: just above the lightning
\pgfmathsetmacro{\yA}{\yB-1*\k}                  % battery wire
\pgfmathsetmacro{\ra}{\yO-\yA}                   % rise of the battery wires
% ---- piston, chamber and cloud (cloud spans exactly from the piston head to the outer edge of the right wall)
\pgfmathsetmacro{\xH}{\G+(231+26)*\k}            % piston head
\pgfmathsetmacro{\xC}{\xH+82*\k}                 % chamber centre
\pgfmathsetmacro{\xW}{\xH+232*\k}                % right wall
\pgfmathsetmacro{\hw}{0.55/0.1612}               % half of 1.1 pt in px
\pgfmathsetmacro{\cx}{(\xH+\xW)/2}
\pgfmathsetmacro{\cs}{((\xW-\xH)/2+\hw-\hw)/117.025}
% ---- right column
\def\xR{1446.5}
\pgfmathsetmacro{\xT}{\xR-83*\k}                 % tips of the output arrows
% ---- middle part (dashed box, label columns, battery diagonals) placed as far left as the chamber allows:
%      the S_0 diagonal ends 38 px right of the chamber's outer wall, as it ends 38 px left of the label column
\pgfmathsetmacro{\xs}{\xC+95*\k}
\pgfmathsetmacro{\Wd}{2*\dh+\wE+2*\st+\rs+\wF}
\pgfmathsetmacro{\DL}{\xW+\hw+2*38+\lo}
\pgfmathsetmacro{\DR}{\DL+\Wd}
\pgfmathsetmacro{\xEl}{\DL+\dh}\pgfmathsetmacro{\xEr}{\xEl+\wE}
\pgfmathsetmacro{\xFr}{\DR-\dh}\pgfmathsetmacro{\xFl}{\xFr-\wF}
\pgfmathsetmacro{\xI}{\DL-\lo}\pgfmathsetmacro{\xO}{\DR+\lo}   % label columns
% ================================================================ dashed box and its label
\draw[line width=0.45pt,dash pattern=on 2.6pt off 1.8pt] (\DL,\yO-\mg-\dm) rectangle (\DR,\yM+\mg+\dm);
\node[inner sep=0pt,font=\small,anchor=south] (N) at ({(\DL+\DR)/2},\yO-\mg-\dm-11) {$\ch{N}_{S_0SM\to S_1S'M'}$};
% ================================================================ left column
\begin{scope}[shift={(\G,\yB)},scale=\k,shift={(-62,-99)}]
	\draw[primaryred,line width=1pt,rounded corners=0.6pt] (18,70) rectangle (133,128);
	\draw[primaryred,line width=1pt] (133,86) -- (143,86) -- (143,112) -- (133,112);
	\fill[primaryred] (31,83) rectangle (119,115);
\end{scope}
\begin{scope}[shift={(\G+14*\k,\yL)},scale=\kl,shift={(-76,-187)}]
	\fill (93,134) -- (62,186) -- (80,186) -- (46,240) -- (104,168) -- (84,168) -- (106,134) -- cycle;
\end{scope}
\begin{scope}[shift={(\G,\yG)},scale=\k]
	\foreach \a in {0,36,...,324}{\fill[rotate=\a] (-9,-61) rectangle (9,-40);}
	\fill (0,0) circle[radius=49];
	\fill[white] (0,0) circle[radius=31];
	\fill (0,0) circle[radius=19];
\end{scope}
% demon, traced from the original (local frame 140.4 x 122.4); its chin is level with the bottom of the cloud
\pgfmathsetmacro{\yD}{\yM+78*\ky+\hw-122.4*\k}
\begin{scope}[shift={(\G+10*\k,\yD)},scale=\k,shift={(-70.2,0)}]
	\fill (0,0) -- (2.2,0) .. controls (8,27) and (28,38) .. (45.8,38)
		.. controls (51,31) and (57,28.5) .. (63.6,28) -- (63.6,23.3) -- (76.8,23.3) -- (76.8,28)
		.. controls (83.4,28.5) and (89.4,31) .. (94.6,38)
		.. controls (112.4,38) and (132.4,27) .. (138.2,0) -- (140.4,0)
		.. controls (142.4,36) and (125.4,62) .. (105.4,65)
		.. controls (105.4,95) and (85.4,121) .. (70.2,122.4)
		.. controls (55,121) and (35,95) .. (35,65)
		.. controls (15,62) and (-2,36) .. (0,0) -- cycle;
	\fill[white] (43,72.5) .. controls (55,72) and (66.4,82) .. (66.4,96.3) .. controls (54,97) and (44,88) .. (43,72.5) -- cycle;
	\fill[white] (97.4,72.5) .. controls (85.4,72) and (74,82) .. (74,96.3) .. controls (86.4,97) and (96.4,88) .. (97.4,72.5) -- cycle;
\end{scope}
% thought bubbles (original sizes and gaps scaled by \k), sagging slightly in the middle
\pgfmathsetmacro{\yQ}{\yD+56*\k+7}
\pgfmathsetmacro{\bx}{\G+10*\k+50.7*\k}          % right edge of the demon at bubble height
\begin{scope}[line width=1.05pt]
	\draw ({\bx+(22.3+17.5)*\k},\yQ+1) ellipse[x radius={(17.5-4.6)*\k},y radius={(15-4.6)*\k}];
	\draw ({\bx+(22.3+35+10.9+21.4)*\k},\yQ+3) ellipse[x radius={(21.4-4.6)*\k},y radius={(17.5-4.6)*\k}];
	\draw ({\bx+(22.3+35+10.9+42.7+9.3+28)*\k},\yQ-6) ellipse[x radius={(28-4.6)*\k},y radius={(21.5-4.6)*\k}];
\end{scope}
% ================================================================ piston and chamber
\begin{scope}[shift={(\xH,\yS)},xscale=\k,yscale=\kc,shift={(-293,-275)}]
	\draw[thk,line cap=butt] (293-257+38,277) -- (293,277);
	\draw[thk] (293,202) -- (293,348);
	\draw[thk,line cap=butt,line join=miter] (225,205) -- (525,205) -- (525,345) -- (225,345);
\end{scope}
% ================================================================ thought cloud
\begin{scope}[shift={(\cx,\yM)},xscale=\cs,yscale=\ky,shift={(-415,-577)}]
\filldraw[fill=white,thk,line join=round] (330,542) .. controls (310.4,510.7) and (422.6,477.7) .. (415.0,517.0) .. controls (407.4,477.7) and (519.6,510.7) .. (500.0,542.0) .. controls (542.7,530.8) and (542.7,623.2) .. (500.0,612.0) .. controls (519.6,639.2) and (407.4,678.8) .. (415.0,642.0) .. controls (422.6,678.8) and (310.4,639.2) .. (330.0,612.0) .. controls (287.3,623.2) and (287.3,530.8) .. (330.0,542.0) -- cycle;
\end{scope}
% ================================================================ entanglement links (red -> blue)
\foreach \xo/\fl/\ya/\yb in {\xC+32*\k/1/\yS+5*\k+24*\k/\yM+10*\k-24*\k,\xR/-1/\yS+5*\k+24*\k/\yM-24*\k}{
	\pgfmathsetmacro{\xx}{\xo}\pgfmathsetmacro{\yy}{\ya}\pgfmathsetmacro{\yz}{\yb}
	\foreach \i in {0,...,35}{
		\pgfmathsetmacro{\ta}{\yy+(\yz-\yy)*\i/36}
		\pgfmathsetmacro{\tb}{\yy+(\yz-\yy)*(\i+1)/36}
		\pgfmathsetmacro{\pc}{max(0,min(100,100-(\i-14)*100/7))}
		\pgfmathsetmacro{\xa}{\xx+\fl*5.5*sin(15*\i)}
		\pgfmathsetmacro{\xb}{\xx+\fl*5.5*sin(15*(\i+1))}
		\draw[primaryred!\pc!memoryblue,line width=1.3pt,line cap=round] (\xa,\ta) -- (\xb,\tb);
	}
}
% ================================================================ spins
\begin{scope}[shift={(\xC+30*\k,\yS+5*\k)},scale=\k]
	\fill[primaryred] (0,0) circle[radius=24];
	\draw[primaryred,spinarr] (-24.2,36.7) -- (33.0,-50.1);
\end{scope}
\begin{scope}[shift={(\xC+35*\k,\yM+10*\k)},scale=\k]
	\fill[memoryblue] (0,0) circle[radius=24];
	\draw[memoryblue,spinarr] (36.9,24.0) -- (-50.3,-32.7);
\end{scope}
\begin{scope}[shift={(\xR,\yS+5*\k)},scale=\k]
	\fill[primaryred] (0,0) circle[radius=24];
	\draw[primaryred,spinarr] (35.2,26.4) -- (-48.0,-36.0);
\end{scope}
\begin{scope}[shift={(\xR,\yM)},scale=\k]
	\fill[memoryblue] (0,0) circle[radius=24];
	\draw[memoryblue,spinarr] (44,0) -- (-60,0);
\end{scope}
% ================================================================ battery, right
\begin{scope}[shift={(\xR,\yB)},scale=\k,shift={(-1408,-99)}]
	\draw[primaryred,line width=1pt,rounded corners=0.6pt] (1355,70) rectangle (1470,128);
	\draw[primaryred,line width=1pt] (1470,86) -- (1480,86) -- (1480,112) -- (1470,112);
	\fill[primaryred] (1368,83) rectangle (1393,115);
\end{scope}
% ================================================================ operations
\fill[primaryred,rounded corners=1.5pt] (\xFl,\yO-\mg) rectangle (\xFr,\yLF+\mg);
\node[big] at ({(\xFl+\xFr)/2},{(\yO+\yLF)/2}) {$\ch{F}$};
\fill[memoryblue,rounded corners=1.5pt] (\xEl,\yLE-\mg) rectangle (\xEr,\yM+\mg);
\node[big] at ({(\xEl+\xEr)/2},{(\yLE+\yM)/2}) {$\ch{E}$};
% ================================================================ wires
% battery wires: 45-degree diagonals ending (starting) 38 px before (after) the label columns
\draw[primaryred,arr,wbend] (\G+118*\k,\yA) -- (\xI-38-\ra,\yA) -- (\xI-38,\yO) -- (\xFl,\yO);
\draw[primaryred,arr,wbend] (\xFr,\yO) -- (\xO+38,\yO) -- (\xO+38+\ra,\yA) -- (\xT,\yA);
\draw[primaryred,arr] (\xs,\yS) -- (\xFl,\yS);
\draw[primaryred,arr] (\xFr,\yS) -- (\xT,\yS);
\draw[memoryblue,arr] (\xs,\yM) -- (\xEl,\yM);
\draw[memoryblue,arr] (\xEr,\yM) -- (\xT,\yM);
% L: stub, 45-degree diagonal, stub of equal length
\draw[memoryblue,arr,wbend] (\xEr,\yLE) -- (\xEr+\st,\yLE) -- (\xEr+\st+\rs,\yLF) -- (\xFl,\yLF);
% labels, in two columns symmetric about the dashed box
\node[lab,above] at (\xI,\yO) {$S_0$};
\node[lab,above] at (\xI,\yS) {$S$};
\node[lab,above] at (\xI,\yM) {$M$};
\node[lab,above] at (\xO,\yO) {$S_1$};
\node[lab,above] at (\xO,\yS) {$S'$};
\node[lab,above] at (\xO,\yM) {$M'$};
\node[lab,anchor=south east,inner sep=0.7pt,xshift=0.45pt,yshift=-0.45pt] at ({\xEr+\st+\rs/2},{(\yLE+\yLF)/2}) {$L$};
% ================================================================ tight bounding box (Bezier control points would otherwise pad it)
\pgfresetboundingbox
\coordinate (BL) at ({\G-61*\k-1},{\yM+78*\ky+\hw});\coordinate (TR) at ({\xR+72*\k+3.5},0);
\path[use as bounding box] (BL) rectangle (TR |- N.north);
\end{tikzpicture}
	\caption{Transformation under a QFGO $\ch{N}_{S_0SM\to S_1S'M'}$ (dashed box) with work supply from a battery.  The QFGO controls the battery and the system (red) while receiving quantum feedback from the controller's memory (blue).  The work consumed by the system equals the battery's decrease in charge, $\beta^{-1}(\ln\lvert S_0\rvert-\ln\lvert S_1\rvert)$.}
\label{fig:work}
\end{figure}

\textbf{\textit{State formation with quantum feedback}}---State formation is a fundamental and well-studied task in quantum thermodynamics~\cite{horodecki2013FundamentalLimitationsQuantum,brandao2013ResourceTheoryQuantum,gour2022RoleQuantumCoherence}.  We now characterize a generalized version thereof, incorporating quantum feedback.  Here the controller aims to prepare a given target state $\rho_{SM}$ of a system $S$ coupled with a memory $M$ while expending as little work as possible.  This task is particularly relevant to writing data from a memory to a thermodynamic device~\cite{narasimhachar2019QuantifyingMemoryCapacity}.  With classical feedback, as considered in Ref.~\cite{narasimhachar2019QuantifyingMemoryCapacity}, only classical data can be written into the device despite the device being quantum, and only classical correlations can be created between the device and the memory, however much work is supplied.  In contrast, with quantum feedback, the controller can now write quantum data into the device, which can thereby become entangled with the memory.  For additional clarity, we denote the target state of the task by the pair $(\rho_{SM},\gamma_S)$ with $\gamma_S$ the Gibbs state of $S$.  We further allow for imperfection in the conversion, measured by an error parameter $\varepsilon$ bounding the trace distance between the actual final state and the target state.  We define the \emph{single-shot work of formation under QFGOs}, $W_\abb{form}^\varepsilon(\rho_{SM},\gamma_S)$, as the battery's decrease in charge, $\beta^{-1}(\ln\lvert S_0\rvert-\ln\lvert S_1\rvert)$, minimized over every QFGO $\ch{N}_{S_0\to S_1SM}$ such that $\ch{N}_{S_0\to S_1SM}[\op{0}{0}_{S_0}]\approx_\varepsilon\op{0}{0}_{S_1}\otimes\rho_{SM}$.  Our first main result is a precise characterization of this work of formation and its asymptotic limit in terms of generalized mutual informations (GMIs).

\begin{theorem}
\label{thm:formation}
For an error parameter $\varepsilon\in[0,1]$ and a target state $(\rho_{SM},\gamma_S)$, the single-shot work of formation under QFGOs and its asymptotic limit are given by
\begin{align}
	W_\abb{form}^\varepsilon\fleft(\rho_{SM},\gamma_S\fright)&=\beta^{-1}I_{\max}^{\uparrow,\varepsilon}\fleft(\rho_{SM}\middle\|\gamma_S\fright), \label{eq:formation}\\
	\lim_{n\to\infty}\tfrac{1}{n}W_\abb{form}^\varepsilon\fleft(\rho_{SM}^{\otimes n},\gamma_S^{\otimes n}\fright)&=\beta^{-1}I\fleft(\rho_{SM}\middle\|\gamma_S\fright)\;\;\forall\varepsilon\in(0,1).
\end{align}
	Here we define the smoothed max-GMI as~\footnote{For $\varepsilon=0$, Eq.~\eqref{eq:max} reduces to the (unoptimized) max-GMI $I_{\max}^\uparrow(\rho_{SM}\|\gamma_S)\coloneq D_{\max}(\rho_{SM}\|\gamma_S\otimes\rho_M)$, which has been considered in Refs.~\cite{hayashi2016CorrelationDetectionOperational,burri2025DoublyMinimizedSandwiched}, with $D_{\max}(\rho\|\sigma)\coloneq\ln\lVert\sigma^{-1/2}\rho\sigma^{-1/2}\rVert_\infty$ the max-relative entropy~\cite{datta2009MinMaxrelativeEntropies}.  The general case of $\varepsilon>0$, however, uses a novel approach to smoothing, which may be of independent interest.} [see Eq.~\eqref{eq:max-alternative} for an alternative expression]
\begin{align}
	&I_{\max}^{\uparrow,\varepsilon}\fleft(\rho_{SM}\middle\|\gamma_S\fright)\coloneq\inf\left\{\ln t\colon\tau_{SM}\leq t\gamma_S\otimes\sigma_M,\right. \notag\\
	&\qquad\qquad\left.\tau_M\leq\sigma_M,\;\tau_{SM}\approx_\varepsilon\rho_{SM}\right\}, \label{eq:max}
\end{align}
where the infimization is over every real number $t$, subnormalized state $\tau_{SM}$, and state $\sigma_M$, and the proximity is measured in the generalized trace distance~\cite{tomamichel2010DualitySmoothMin}.  The Umegaki GMI~\cite{hayashi2016CorrelationDetectionOperational,morris2019AssistedWorkDistillation} is defined as $I(\rho_{SM}\|\gamma_S)\coloneq D(\rho_{SM}\|\gamma_S\otimes\rho_M)=D(\rho_S\|\gamma_S)+I(S;M)_\rho$ with $D(\rho\|\sigma)\coloneq\tr[\rho(\ln\rho-\ln\sigma)]$ the Umegaki relative entropy~\cite{umegaki1962ConditionalExpectationOperator} and $I(S;M)_\rho\coloneq D(\rho_{SM}\|\rho_S\otimes\rho_M)$ the quantum mutual information.
\end{theorem}

We explicitly construct in Appendix~C a QFGO that achieves the single-shot work of formation in Eq.~\eqref{eq:formation}.  Its optimality (see SM~\cite[Sec.~\ref{sec:formation}]{Note1}) follows from properties of the smoothed max-GMI presented in Appendix~B.

\textbf{\textit{Work extraction with quantum feedback}}---The next task we characterize is work extraction with quantum feedback, which has the opposite goal to state formation: the controller aims to convert a given initial state $(\rho_{SM},\gamma_S)$ of a system $S$ coupled with a memory $M$ into as much work as possible.  This generalizes work extraction without feedback~\cite{aberg2013TrulyWorklikeWork,horodecki2013FundamentalLimitationsQuantum,brandao2013ResourceTheoryQuantum,gour2022RoleQuantumCoherence}.  Work extraction assisted by quantum side information has been considered before in Ref.~\cite{morris2019AssistedWorkDistillation}, but with the side information fed back only classically.  Here, we allow for quantum feedback and observe that the extractable work under QFGOs can be strictly larger than that when the feedback is classical, with the difference given by quantum discord~\cite{ollivier2001QuantumDiscordMeasure,zurek2003QuantumDiscordMaxwells} in the asymptotic regime~\cite{manzano2018OptimalWorkExtraction}.  Allowing for imperfection parametrized by $\varepsilon$, we define the \emph{single-shot extractable work under QFGOs}, $W_\abb{extr}^\varepsilon(\rho_{SM},\gamma_S)$, as the battery's increase in charge, $\beta^{-1}(\ln\lvert S_1\rvert-\ln\lvert S_0\rvert)$, maximized over every QFGO $\ch{N}_{S_0SM\to S_1}$ such that $\ch{N}_{S_0SM\to S_1}[\op{0}{0}_{S_0}\otimes\rho_{SM}]\approx_\varepsilon\op{0}{0}_{S_1}$.  Our second main result is a precise characterization of this extractable work and its asymptotic limit in terms of GMIs.

\begin{theorem}
\label{thm:extraction}
For an error parameter $\varepsilon\in[0,1]$ and an initial state $(\rho_{SM},\gamma_S)$, the single-shot extractable work under QFGOs and its asymptotic limit are given by
\begin{align}
	W_\abb{extr}^\varepsilon\fleft(\rho_{SM},\gamma_S\fright)&=\beta^{-1}I_{\min}^{\downarrow,\varepsilon}\fleft(\rho_{SM}\middle\|\gamma_S\fright), \label{eq:extraction}\\
	\lim_{n\to\infty}\tfrac{1}{n}W_\abb{extr}^\varepsilon\fleft(\rho_{SM}^{\otimes n},\gamma_S^{\otimes n}\fright)&=\beta^{-1}I\fleft(\rho_{SM}\middle\|\gamma_S\fright)\;\;\forall\varepsilon\in(0,1).
\end{align}
Here we define the smoothed min-GMI as
\begin{align}
	&I_{\min}^{\downarrow,\varepsilon}\fleft(\rho_{SM}\middle\|\gamma_S\fright)\coloneq\sup\left\{-\ln\left\lVert\tr_S\fleft[\Lambda_{SM}\gamma_S\fright]\right\rVert_\infty\colon\right. \notag\\
	&\qquad\qquad\left.\tr\fleft[\Lambda_{SM}\rho_{SM}\fright]\geq1-\varepsilon,\;0\leq\Lambda_{SM}\leq\1_{SM}\right\}, \label{eq:min}
\end{align}
with the supremization over every operator $\Lambda_{SM}$.
\end{theorem}

The single-shot extractable work in Eq.~\eqref{eq:extraction} is achievable with a QFGO constructed in Appendix~D, whose optimality (see SM~\cite[Sec.~\ref{sec:extraction}]{Note1}) follows from properties of the smoothed min-GMI presented in Appendix~B.

Just as state formation is relevant to data writing, work extraction is closely related to data erasure.  Erasing a system $S$ with a trivial Hamiltonian, i.e., converting $S$ from its initial state to the pure state $\op{0}{0}_S$, costs, by definition, $\beta^{-1}\ln\lvert S\rvert$ minus the extractable work from $S$.  Hence, for an initial state $\rho_{SM}$ coupled with a memory $M$, the single-shot work cost of erasing $S$ under QFGOs is precisely expressible in terms of a smoothed conditional max-entropy, $\beta^{-1}H_{\max}^{\uparrow,\varepsilon}(S|M)_\rho\coloneq-\beta^{-1}I_{\min}^{\downarrow,\varepsilon}(\rho_{SM}\|\1_S)$~\footnote{Here $H_{\max}^{\uparrow,\varepsilon}(S|M)_\rho$ is the (operator-) smoothed (optimized) (Petz-type) conditional max-entropy, different from quantities with similar names or notations in the literature~\cite{renner2005SecurityQuantumKey,konig2009OperationalMeaningMin,tomamichel2009FullyQuantumAsymptotic,tomamichel2010DualitySmoothMin,tomamichel2011LeftoverHashingQuantum}; see Appendix~B.}.  This can be viewed as a single-shot generalization of Landauer's principle~\cite{landauer1961IrreversibilityHeatGeneration,esposito2011SecondLawLandauer,reeb2014ImprovedLandauerPrinciple} with quantum feedback.  The potential negativity of $H_{\max}^{\uparrow,\varepsilon}(S|M)_\rho$ thus acquires a thermodynamic meaning, signifying cooling rather than heating during the erasure.  The conditional von Neumann entropy $H(S|M)_\rho=-I(\rho_{SM}\|\1_S)$ is recovered in the asymptotic limit, echoing the main finding of Ref.~\cite{delrio2011ThermodynamicMeaningNegative} in a different operational setting.

\textbf{\textit{Generalized second law with quantum feedback}}---Most generally, one can consider the task of converting a given initial state $(\rho_{SM},\gamma_S)$ to a given target state $(\rho'_{S'M'},\gamma'_{S'})$ with quantum feedback while expending as little work as possible.  This reduces to state formation or work extraction with quantum feedback, respectively, when the initial or target state is trivial.  We define the associated \emph{single-shot work cost under QFGOs}, $W^\varepsilon((\rho_{SM},\gamma_S)\to(\rho'_{S'M'},\gamma'_{S'}))$, as the battery's decrease in charge, $\beta^{-1}(\ln\lvert S_0\rvert-\ln\lvert S_1\rvert)$, minimized over every QFGO $\ch{N}_{S_0SM\to S_1S'M'}$ such that $\ch{N}_{S_0SM\to S_1S'M'}[\op{0}{0}_{S_0}\otimes\rho_{SM}]\approx_\varepsilon\op{0}{0}_{S_1}\otimes\rho'_{S'M'}$.  We derive upper and lower bounds on this single-shot work cost (see Appendix~E for a formal statement); these bounds coincide in the asymptotic limit, identifying a unique conditional free energy.

\begin{theorem}
\label{thm:asymptotic}
For an error parameter $\varepsilon\in(0,1)$, an initial state $(\rho_{SM},\gamma_S)$, and a target state $(\rho'_{S'M'},\gamma'_{S'})$, the asymptotic work cost under QFGOs is given by
\begin{align}
	&\lim_{n\to\infty}\tfrac{1}{n}W^\varepsilon\fleft(\fleft(\rho_{SM}^{\otimes n},\gamma_S^{\otimes n}\fright)\to\fleft(\rho_{S'M'}^{\prime\,\otimes n},\gamma_{S'}^{\prime\,\otimes n}\fright)\fright) \notag\\
	&=\left(F\fleft(S'\middle|M'\fright)_{\rho'}-F\fleft(S'\fright)_{\gamma'}\right)-\left(F\fleft(S\middle|M\fright)_\rho-F\fleft(S\fright)_\gamma\right), \label{eq:asymptotic}
\end{align}
where the conditional nonequilibrium free energy~\cite{bera2017GeneralizedLawsThermodynamics} is defined as $F(S|M)_\rho\coloneq\tr[\h{H}_S\rho_S]-\beta^{-1}H(S|M)_\rho=F(S)_\rho+\beta^{-1}I(S;M)_\rho=F(S)_\gamma+\beta^{-1}I(\rho_{SM}\|\gamma_S)$.
\end{theorem}

Since QFGOs include all physically admissible feedback-control schemes without external work supply, Theorem~\ref{thm:asymptotic} implies a fundamental lower limit on the average work $\langle W\rangle$ consumed by the system in any feedback-control process, generalizing the second law with feedback in Eq.~\eqref{eq:classical} to the fully quantum regime:
\begin{align}
	\left\langle W\right\rangle&\geq\Delta F_\abb{sys}+\beta^{-1}\Delta I_\abb{q}, \label{eq:quantum}
\end{align}
where $\Delta F_\abb{sys}\equiv F(S')_{\rho'}-F(S)_\rho$ is the system's change in nonequilibrium free energy (but with the system now being quantum) and $\Delta I_\abb{q}\equiv I(S';M')_{\rho'}-I(S;M)_\rho$ is the change in \emph{quantum} mutual information between the system and the memory.  This is because $\langle W\rangle$ is lower bounded by the left-hand side of Eq.~\eqref{eq:asymptotic} plus the system's change in equilibrium free energy, $F(S')_{\gamma'}-F(S)_\gamma$, with the latter typically omitted in resource-theoretic analyses.  An inequality similar to Eq.~\eqref{eq:quantum} appears in Ref.~\cite{bera2017GeneralizedLawsThermodynamics} but in a different operational context, where a system and a thermal bath evolve with their joint von Neumann entropy preserved, while here we consider a system and a memory under feedback-control schemes respecting Kelvin's principle.  Unlike Ref.~\cite{bera2017GeneralizedLawsThermodynamics}, our framework is based on operational principles rather than phenomenological quantities, thus enabling tight results even on the microscopic scale (Theorems~\ref{thm:formation} and \ref{thm:extraction}).

\textbf{\textit{Consistency with the traditional second law over a cycle}}---We now show that our generalized second law with quantum feedback in Eq.~\eqref{eq:quantum} is consistent with the traditional second law over a closed cycle.  Here, for a fully quantum controller, who is capable of acquiring, processing, and feeding back quantum side information coherently, the most general form of such a cycle goes beyond previous studies~\cite{sagawa2008SecondLawThermodynamics,sagawa2009MinimalEnergyCost,jacobs2009SecondLawThermodynamics,parrondo2015ThermodynamicsInformation,minagawa2025UniversalValiditySecond} in the following respects: the controller's memory is quantum throughout; the acquisition of side information is no longer described by a measurement but more generally by a bipartite channel, applied to both the system and the memory; and the processing and feedback of side information are quantum; see Fig.~\ref{fig:cycle}.  To show consistency with the traditional second law, we demonstrate, from our generalized second law alone, that the overall average work cost over the cycle cannot be negative.

By Eq.~\eqref{eq:quantum}, the average work consumed by the system in the feedback-control stage (the dashed box in Fig.~\ref{fig:cycle}) is bounded as $\langle W_\abb{sys}\rangle\equiv\langle W\rangle\geq\Delta F_\abb{sys}+\beta^{-1}\Delta I_\abb{q}$.  Since the memory itself does not receive any feedback from elsewhere in the feedback-control stage, the average work it consumes in this stage is bounded as $\langle W_\abb{mem}\rangle\geq\Delta F_\abb{mem}$, without mutual-information terms, where $\Delta F_\abb{mem}\equiv F(M')_{\rho'}-F(M)_\rho$ is the memory's change in nonequilibrium free energy.  For the same reason, the average work consumed by both the system and the memory in the acquisition and reset stages is bounded as $\langle W_\abb{acqn}\rangle+\langle W_\abb{rst}\rangle\geq-\Delta F_\abb{sys\&mem}$, where $\Delta F_\abb{sys\&mem}\equiv F(S'M')_{\rho'}-F(SM)_\rho=\Delta F_\abb{sys}+\Delta F_\abb{mem}+\beta^{-1}\Delta I_\abb{q}$ is the global change in nonequilibrium free energy in the feedback-control stage (and thus $-\Delta F_\abb{sys\&mem}$ the global change in the acquisition and reset stages, as the cycle closes), with the last equality following from the definitions of nonequilibrium free energy and quantum mutual information.  Summing up yields $\langle W_\abb{sys}\rangle+\langle W_\abb{mem}\rangle+\langle W_\abb{acqn}\rangle+\langle W_\abb{rst}\rangle\geq0$, precisely as the traditional second law would predict.  This ensures that any work advantage provided by quantum feedback must be offset elsewhere in the cycle, generalizing the core insight of information thermodynamics~\cite{sagawa2009MinimalEnergyCost,minagawa2025UniversalValiditySecond} to the fully quantum regime.

\begin{figure}[t]
\begin{tikzpicture}[x=0.945cm,y=0.945cm,
	wire/.style={line width=1.3pt,line cap=round,line join=round},
	arr/.style={wire,-{Straight Barb[length=3pt,width=5.2pt,line width=1.3pt]}},
	loop/.style={wire,rounded corners=6pt},
	halo/.style={line width=4.2pt,white,rounded corners=6pt},
	lab/.style={inner sep=1pt,font=\scriptsize},
	slab/.style={blur shadow={shadow xshift=0.9pt,shadow yshift=-1.3pt,shadow blur steps=6,shadow blur extra rounding=0.6pt,shadow opacity=40}},
	txt/.style={inner sep=0pt,font=\small}]
% rows, loop depth and skew, front corners of the loops
\def\yS{1.05}\def\yM{0.25}\def\D{1.25}\def\sk{0.56}\def\xl{0.02}\def\xr{8.31}
% stage: #1 color, #2 x-left, #3 x-right, #4 y-bottom, #5 y-top, #6 text
\newcommand{\stage}[6]{\fill[#1,slab,rounded corners=1.5pt] (#2,#4) rectangle (#3,#5); \node[txt] at ({(#2+#3)/2},{(#4+#5)/2}) {#6};}
% ---------------------------------------------------------------- M loop (behind): rst -> right -> back -> left -> acqn
\draw[loop,memoryblue,-{Straight Barb[length=3pt,width=5.2pt,line width=1.3pt]}] (7.51,\yM) -- (\xr,\yM) -- (\xr+\sk,\yM+\D) -- (\xl+\sk,\yM+\D) -- (\xl,\yM) -- (0.82,\yM);
% ---------------------------------------------------------------- S loop (in front): ctrl -> right -> back -> left -> acqn
\draw[halo,overlay] (5.655,\yS) -- (\xr,\yS) -- (\xr+\sk,\yS+\D) -- (\xl+\sk,\yS+\D) -- (\xl,\yS) -- (0.72,\yS);
\draw[loop,primaryred,-{Straight Barb[length=3pt,width=5.2pt,line width=1.3pt]}] (5.655,\yS) -- (\xr,\yS) -- (\xr+\sk,\yS+\D) -- (\xl+\sk,\yS+\D) -- (\xl,\yS) -- (0.82,\yS);
% ---------------------------------------------------------------- dashed box (feedback-control stage)
\draw[line width=0.45pt,dash pattern=on 2.6pt off 1.8pt] (2.315,-0.06) rectangle (5.995,1.36);
% ---------------------------------------------------------------- front wires
\draw[arr,primaryred] (1.74,\yS) -- (4.735,\yS);
\draw[arr,memoryblue] (1.74,\yM) -- (2.655,\yM);
% L: two stubs of 0.45 and a 45-degree diagonal (rise 0.24)
\draw[arr,memoryblue,rounded corners=3pt] (3.595,\yM+0.28) -- (4.045,\yM+0.28) -- (4.285,\yS-0.28) -- (4.735,\yS-0.28);
\draw[arr,memoryblue] (3.595,\yM) -- (6.57,\yM);
% ---------------------------------------------------------------- stages (each extends 0.17 cm beyond its outermost wires, about 0.6 of the wire spacing, as in Fig. 1)
\stage{primarygreen}{0.82}{1.74}{0.08}{1.22}{acqn.}
\stage{memoryblue}{2.655}{3.595}{0.08}{0.70}{proc.}
\stage{primaryred}{4.735}{5.655}{0.60}{1.22}{ctrl.}
\stage{primarygreen}{6.57}{7.51}{-0.02}{0.52}{rst.}
% ---------------------------------------------------------------- labels
\node[lab,above] at (2.0275,\yS) {$S$};
\node[lab,above] at (2.0275,\yM) {$M$};
\node[lab,anchor=south east,inner sep=0.7pt] at (4.165,0.65) {$L$};
\node[lab,above] at (6.2825,\yS) {$S'$};
\node[lab,above] at (6.2825,\yM) {$M'$};
\end{tikzpicture}
\caption{A general cycle carried out by a fully quantum controller.  The cycle consists of three stages: (i) acquisition of side information, (ii) feedback control (dashed box, including processing of the memory), and (iii) reset of the memory.  The generalized second law with quantum feedback [Eq.~\eqref{eq:quantum}], which limits the average work consumed by the system in the feedback-control stage, does not contradict the nonnegativity of the overall average work cost over the cycle.}
\label{fig:cycle}
\end{figure}

\textbf{\textit{Example}}---We now give a concrete example of work extraction from a \emph{classical} system with quantum feedback, whose ultimate limit is achieved by a QFGO in a single shot but not by any scheme with classical feedback even asymptotically, thus demonstrating a strict advantage of quantum over classical feedback even \emph{without} entanglement.  Consider a classical system $S$ of dimension $\lvert S\rvert=4$ with a trivial Hamiltonian.  The controller's side information, stored in a qubit memory $M$, was acquired beforehand as follows: originally, the controller had perfect classical side information about $S$, stored in a classical memory $C$, with $S$ locally in thermal equilibrium; later, the classical signals in $C$ were corrupted by the noisy channel $\op{s}{s}_C\mapsto\op{\hat{\sigma}_s\psi}{\hat{\sigma}_s\psi}_M$, where $\op{\psi}{\psi}_M=(1/2)(\hat{\sigma}_{0,M}+\sum_{s\in\{1,2,3\}}\hat{\sigma}_{s,M}/\sqrt{3})$ is a pure state and $\hat{\sigma}_s$ for $s\in\{0,1,2,3\}$ are the Pauli operators.  This put $S$ and $M$ in the joint state $\hat{\rho}_{SM}=\sum_{s\in\{0,1,2,3\}}(1/4)\op{s}{s}_S\otimes\op{\hat{\sigma}_s\psi}{\hat{\sigma}_s\psi}_M$.

Now, the controller wishes to extract as much work as possible from $S$.  By Theorem~\ref{thm:extraction}, even with quantum feedback from $M$, the controller cannot extract more work than $\beta^{-1}I(\hat{\rho}_{SM}\|\1_S/4)=\beta^{-1}\chi(\en{E})=\beta^{-1}\ln2$, where $\chi(\en{E})$ is the Holevo quantity~\cite{holevo1973BoundsQuantityInformation} for the ensemble $\en{E}=((1/4,\op{\hat{\sigma}_s\psi}{\hat{\sigma}_s\psi}))_{s\in\{0,1,2,3\}}$.  This limit is achieved by the following feedback-control scheme in a single shot: first, transfer the side information in $M$ identically to a link $L$, which is fed back to the system; then, apply a Pauli unitary gate to $L$ controlled by $S$, thus restoring $L$ to the pure state $\op{\psi}{\psi}_L$; lastly, turn $L$ into a qubit battery $S_1$ with work storage $\beta^{-1}\ln2$ and discard $S$.  Notably, this feedback-control scheme is a QFGO, because its control operation $\hat{\ch{F}}_{SL\to S_1}[\cdot]=\sum_{s\in\{0,1,2,3\}}\hat{\sigma}_{s,S_1}\tr_S[\op{s}{s}_S(\cdot)_{SL\to SS_1}]\hat{\sigma}_{s,S_1}$ respects Kelvin's principle [see Eq.~\eqref{eq:kelvin}]: if $S$ were in thermal equilibrium and independent of $L$, then $S_1$ would end up being twirled and storing no work, as $\hat{\ch{F}}_{SL\to S_1}[\1_S/4\otimes(\cdot)_L]=(1/4)\sum_{s\in\{0,1,2,3\}}\hat{\sigma}_{s,S_1}(\cdot)_{L\to S_1}\hat{\sigma}_{s,S_1}=\1_{S_1}/2\otimes\tr_{L}[\cdot]$.  In contrast, the controller can never extract as much work with classical feedback even asymptotically: we show that the extractable work with classical feedback cannot exceed $\beta^{-1}I_\abb{acc}(\en{E})=\beta^{-1}\ln(4/3)$ (see SM~\cite[Sec.~\ref{sec:example}]{Note1}), where $I_\abb{acc}(\en{E})$ is the accessible information~\cite{holevo1973BoundsQuantityInformation} for $\en{E}$, strictly smaller than the $\beta^{-1}\ln2$ work extractable using the above QFGO in a single shot.

\textbf{\textit{Discussion}}---Based on two simple operational principles, we have established fundamental lower limits on the work costs of state conversion with quantum feedback, including the tightest possible bounds on the single-shot work of formation and extractable work as well as a generalized second law of thermodynamics.  Our findings extend the study of thermodynamic feedback control to the fully quantum regime and generalize many rigorous quantitative insights in single-shot quantum thermodynamics~\cite{horodecki2013FundamentalLimitationsQuantum,gour2015ResourceTheoryInformational,gour2022RoleQuantumCoherence} by incorporating quantum side information.  They also demonstrate a strict advantage of quantum over classical feedback and shed new light on uniquely quantum phenomena in both information theory and thermodynamics, such as negative conditional entropy~\cite{cerf1997NegativeEntropyInformation,horodecki2005PartialQuantumInformation} and erasure accompanied by cooling~\cite{delrio2011ThermodynamicMeaningNegative}.

Our work leads to several avenues for future research.  First, our fundamental lower limits on work costs are achievable within our maximally permissive operational framework, but their practical feasibility remains open, both theoretically and experimentally.  A natural next step is to investigate the role of coherence~\cite{lostaglio2015DescriptionQuantumCoherence,lostaglio2015QuantumCoherenceTimetranslation,cwiklinski2015LimitationsEvolutionQuantum,gour2022RoleQuantumCoherence,tajima2025GibbspreservingOperationsRequiring} in our setting and to formalize a quantum-feedback-assisted generalization of thermal operations~\cite{janzing2000ThermodynamicCostReliability,horodecki2013FundamentalLimitationsQuantum,brandao2013ResourceTheoryQuantum}.  %Second, it would be interesting to explore catalysis~\cite{brandao2015SecondLawsQuantum,muller2018CorrelatingThermalMachines,lipka-bartosik2021AllStatesAre,shiraishi2021QuantumThermodynamicsCorrelatedcatalytic,lipka-bartosik2024CatalysisQuantumInformation} in our setting, and in particular, to identify the role of Rényi GMIs in dictating single-shot catalytic convertibility~\cite{brandao2015SecondLawsQuantum,gour2021EntropyRelativeEntropy}.
Second, in addition to state conversion, our framework can be applied to studying logical processes~\cite{faist2015MinimalWorkCost,faist2018FundamentalWorkCost}, where the state of relevance is defined in reference to an external environment.  %We also expect the ideas and techniques presented here to be instructive for the study of channel entropies~\cite{gour2019ComparisonQuantumChannels,gour2021EntropyQuantumChannel,gour2025InevitableNegativityAdditivity} and channel thermodynamics~\cite{faist2019ThermodynamicCapacityQuantum,faist2021ThermodynamicImplementationsQuantum,faist2025ThermalizationPartialInformation,faist2025MaximumChannelEntropy}.  
Further beyond, our findings could have implications for characterizing the energetic costs of an adaptive agent interacting with or receiving feedback from its environment~\cite{thompson2025EnergeticAdvantagesQuantum}, with the agent, the environment, or both, being quantum~\cite{elliott2022QuantumAdaptiveAgents,lumbreras2026ReinforcementLearningQuantum}.

\begin{acknowledgements}
\textbf{\textit{Acknowledgments}}---We thank Karol Horodecki, Chung-Yun Hsieh, Seth Lloyd, Varun Narasimhachar, Nelly Ng, Bartosz Regula, and Ryuji Takagi for helpful discussions and comments.  K.J. and M.M.W. acknowledge support from the NSF under Grant No.~2329662.  G.G. acknowledges support from the Israel Science Foundation under Grant No.~1192/24.
\end{acknowledgements}

%-------------------------------
%	REFERENCES
%-------------------------------

\bibliographystyle{apsc}
\bibliography{Library}

%-------------------------------
%	APPENDICES
%-------------------------------

\onecolumngrid

\begin{center}
\textbf{\large Appendices}
\end{center}

\twocolumngrid

\textbf{\textit{Appendix~A: Properties of QFGOs}}---We provide simple necessary and sufficient conditions for deciding whether a channel is a QFGO (see the companion paper~\cite{ji2026QuantumConditionalRelative} for a proof).  Let $\ch{R}_{S\to S}^\gamma[\cdot]\equiv\gamma_S\otimes\tr_S[\cdot]$ (resp.\ $\ch{R}_{S'\to S'}^{\gamma'}$) be the thermalization channel on $S$ (resp.\ $S'$).

{
\newcounter{temp}
\setcounter{temp}{\value{theorem}}
\setcounter{theorem}{100}
\renewcommand{\thetheorem}{\ref{prop:operation}}
\begin{proposition}[\cite{ji2026QuantumConditionalRelative}]
\label{prop:operation-app}
For a channel $\ch{N}_{SM\to S'M'}$, the following statements are equivalent.
\begin{itemize}
	\item[(i)] $\ch{N}_{SM\to S'M'}$ is a QFGO.
	\item[(ii)] $\ch{N}_{SM\to S'M'}$ is thermalization covariant on the system:
	\begin{align}
		\ch{N}_{SM\to S'M'}\circ\ch{R}_{S\to S}^\gamma&=\ch{R}_{S'\to S'}^{\gamma'}\circ\ch{N}_{SM\to S'M'}.
	\end{align}
	\item[(iii)] $\ch{N}_{SM\to S'M'}$ is system-to-memory nonsignaling and conditionally Gibbs preserving: there exists a channel $\ch{E}_{M\to M'}$ such that
	\begin{align}
		\tr_{S'}\circ\ch{N}_{SM\to S'M'}&=\tr_S\otimes\ch{E}_{M\to M'}, \\
		\ch{N}_{SM\to S'M'}\fleft[\gamma_S\otimes\left(\cdot\right)_M\fright]&=\gamma'_{S'}\otimes\ch{E}_{M\to M'}\fleft[\cdot\fright].
	\end{align}
\end{itemize}
\end{proposition}
\setcounter{theorem}{\value{temp}}
}

Statement~(ii) of Proposition~\ref{prop:operation-app} identifies QFGOs precisely as the bipartite channels that act covariantly with respect to a resource-destroying map~\cite{liu2017ResourceDestroyingMaps}, represented in our setting by thermalization of the system.  Statement~(iii) indicates that QFGOs generalize the conditionally mixing operations in Ref.~\cite{gour2024InevitabilityKnowingLess} to systems with nontrivial Hamiltonians.  From an information-theoretic perspective, convertibility under QFGOs can thus be viewed as a notion of quantum conditional relative majorization (see the companion paper~\cite{ji2026QuantumConditionalRelative} for details).

Coarse graining in thermodynamics concerns the translation between different perspectives on a system.  To ensure a consistent view of thermal equilibrium across different perspectives, it was argued in Ref.~\cite{faist2018FundamentalWorkCost} that coarse graining must be described by Gibbs-preserving maps.  Extending this idea, we argue that, in the presence of a memory, coarse graining should ensure a consistent view of \emph{thermalization} of the system across different perspectives and thus must be described by QFGOs, due to statement~(ii) of Proposition~\ref{prop:operation-app}.  This point is further supported by showing that the Petz recovery map for a QFGO and thermalized initial state is also a QFGO; see SM~\cite[Sec.~\ref{sec:recovery}]{Note1} for details.

\textbf{\textit{Appendix~B: Generalized mutual informations (GMIs)}}---A detailed investigation of GMIs can be found in the companion paper~\cite{ji2026QuantumConditionalRelative}.  Here we provide a selection of results directly relevant to this paper.

For a real number $\varepsilon\in[0,1]$, recall the definitions of the \emph{smoothed max-GMI} and the \emph{smoothed min-GMI} from Eqs.~\eqref{eq:max} and \eqref{eq:min}, respectively.  The smoothed max-GMI admits the following alternative expression:
\begin{align}
	&I_{\max}^{\uparrow,\varepsilon}\fleft(\rho_{SM}\middle\|\gamma_S\fright) \notag\\
	&=\inf_{S_1,\omega_{S_1SM}}\left\{D_{\max}\fleft(\omega_{S_1SM}\middle\|\1_{S_1}/\left\lvert S_1\right\rvert\otimes\gamma_S\otimes\omega_M\fright)-\ln\left\lvert S_1\right\rvert\colon\right. \notag\\
	&\qquad\left.\omega_{S_1SM}\approx_\varepsilon\op{0}{0}_{S_1}\otimes\rho_{SM}\right\}, \label{eq:max-alternative}
\end{align}
where the infimization is over every battery $S_1$ and normalized state $\omega_{S_1SM}$, $D_{\max}(\rho\|\sigma)\coloneq\ln\lVert\sigma^{-1/2}\rho\sigma^{-1/2}\rVert_\infty$ is the max-relative entropy~\cite{datta2009MinMaxrelativeEntropies}, and the proximity is measured in the standard trace distance.  The smoothing in Eq.~\eqref{eq:max-alternative} can be interpreted as follows: we put $S$ alongside a battery $S_1$, calculate the nonsmoothed max-GMI (i.e., with $\varepsilon=0$) for the nearby state $(\omega_{S_1SM},\1_{S_1}/\lvert S_1\rvert\otimes\gamma_S)$, and subtract the battery's contribution $\ln\lvert S_1\rvert$.  The smoothed min-GMI admits the following alternative expression:
\begin{align}
	I_{\min}^{\downarrow,\varepsilon}\fleft(\rho_{SM}\middle\|\gamma_S\fright)&=\inf_{\sigma_M}D_{\min}^\varepsilon\fleft(\rho_{SM}\middle\|\gamma_S\otimes\sigma_M\fright), \label{eq:min-alternative}
\end{align}
where the infimization is over every state $\sigma_M$ and $D_{\min}^\varepsilon(\rho\|\sigma)\coloneq\sup_\Lambda\{-\ln\tr[\Lambda\sigma]\colon\tr[\Lambda\rho]\geq1-\varepsilon,\;0\leq\Lambda\leq\1\}$ is the hypothesis-testing relative entropy~\cite{buscemi2010QuantumCapacityChannels,wang2012OneshotClassicalquantumCapacity}.  Both the smoothed max- and min-GMIs are efficiently computable via semidefinite programs.  Moreover, they have the following properties, which are crucial to establishing Theorems~\ref{thm:formation} and \ref{thm:extraction}.

\begin{proposition}[\cite{ji2026QuantumConditionalRelative}]
\label{prop:properties}
For a real number $\varepsilon\in[0,1]$ and a state $(\rho_{SM},\gamma_S)$, the smoothed max- and min-GMI, denoted by $\hat{I}^\varepsilon\in\{I_{\max}^{\uparrow,\varepsilon},I_{\min}^{\downarrow,\varepsilon}\}$, satisfy the following properties.
\begin{itemize}
	\item[(i)] Monotonicity under QFGOs: for every QFGO $\ch{N}_{SM\to S'M'}$,
	\begin{align}
		\hat{I}^\varepsilon\fleft(\rho_{SM}\middle\|\gamma_S\fright)&\geq\hat{I}^\varepsilon\fleft(\ch{N}_{SM\to S'M'}\fleft[\rho_{SM}\fright]\middle\|\gamma'_{S'}\fright).
	\end{align}
	\item[(ii)] Robustness in the presence of a battery: for every battery $S_0$,
	\begin{align}
		&\hat{I}^\varepsilon\fleft(\op{0}{0}_{S_0}\otimes\rho_{SM}\middle\|\1_{S_0}/\left\lvert S_0\right\rvert\otimes\gamma_S\fright) \notag\\
		&\qquad=\ln\left\lvert S_0\right\rvert+\hat{I}^\varepsilon\fleft(\rho_{SM}\middle\|\gamma_S\fright).
	\end{align}
	\item[(iii)] Asymptotic equipartition property:
	\begin{align}
		\lim_{n\to\infty}\tfrac{1}{n}\hat{I}^\varepsilon\fleft(\rho_{SM}^{\otimes n}\middle\|\gamma_S^{\otimes n}\fright)&=I\fleft(\rho_{SM}\middle\|\gamma_S\fright)\;\;\forall\varepsilon\in(0,1).
	\end{align}
\end{itemize}
\end{proposition}

The smoothed max- and min-GMIs induce the \emph{smoothed conditional min-entropy}, $H_{\min}^{\downarrow,\varepsilon}(S|M)_\rho\coloneq-I_{\max}^{\uparrow,\varepsilon}(\rho_{SM}\|\1_S)$, and the \emph{smoothed conditional max-entropy}, $H_{\max}^{\uparrow,\varepsilon}(S|M)_\rho\coloneq-I_{\min}^{\downarrow,\varepsilon}(\rho_{SM}\|\1_S)$, respectively.  Note that the smoothed conditional min- and max-entropies here differ from quantities with similar names or notations in the literature~\cite{renner2005SecurityQuantumKey,konig2009OperationalMeaningMin,tomamichel2009FullyQuantumAsymptotic,tomamichel2010DualitySmoothMin,tomamichel2011LeftoverHashingQuantum,anshu2020PartiallySmoothedInformation} due to a different approach to smoothing; see the companion paper~\cite{ji2026QuantumConditionalRelative} for details.

\textbf{\textit{Appendix~C: Optimal scheme for state formation}}---We explicitly construct a QFGO that achieves the single-shot work of formation in Eq.~\eqref{eq:formation}.  To prepare a target state $(\rho_{SM},\gamma_S)$ up to error $\varepsilon\in[0,1]$, consider the following feedback-control scheme:
\begin{itemize}
	\item[(i)] identify a near-optimal state $\omega_{S_1SM}$ for the alternative expression of the smoothed max-GMI in Eq.~\eqref{eq:max-alternative} and prepare $\lvert M\rvert\omega_M^{1/2}\Phi_{LM}\omega_M^{1/2}$ between a link $L$ and a memory $M$ isomorphic to each other, with $\Phi_{LM}\equiv(1/\lvert M\rvert)\sum_{i,j=0}^{\lvert M\rvert-1}\op{i}{j}_L\otimes\op{i}{j}_M$ the standard maximally entangled state;
	\item[(ii)] feed $L$ back and put it alongside a battery $S_0$ of dimension $\lvert S_0\rvert\coloneq\lceil e^{D_{\max}(\omega_{S_1SM}\|\1_{S_1}/\lvert S_1\rvert\otimes\gamma_S\otimes\omega_M)}\rceil$;
	\item[(iii)] measure the battery $S_0$ in the computational basis; if the outcome is $0$, transfer the information in $L$ to a battery $S_1$ and a system $S$ with the following channel (assuming $\omega_L$ to be full-rank without loss of generality):
	\begin{align}
		\ch{F}^{(0)}_{L\to S_1S}\fleft[\cdot\fright]&\coloneq\tr_L\fleft[\omega_L^{-\frac{1}{2}}\omega_{S_1SL}\omega_L^{-\frac{1}{2}}\fleft(\cdot\fright)_L^\top\fright];
	\end{align}
	otherwise transfer the information with the following channel (which is indeed completely positive and trace preserving due to the choice of $\lvert S_0\rvert$ and by construction):
	\begin{align}
		&\ch{F}_{L\to S_1S}^\perp\fleft[\cdot\fright] \notag\\
		&\coloneq\tfrac{1}{\left\lvert S_0\right\rvert-1}\left(\left\lvert S_0\right\rvert\tfrac{\1_{S_1}}{\left\lvert S_1\right\rvert}\otimes\gamma_S\otimes\tr_L\fleft[\cdot\fright]-\ch{F}^{(0)}_{L\to S_1S}\fleft[\cdot\fright]\right).
	\end{align}
\end{itemize}

Formally, this scheme is represented by the channel $\ch{N}_{S_0\to S_1SM}\coloneq\ch{F}_{S_0L\to S_1S}\circ\ch{E}_{\emptyset\to LM}$, where
\begin{align}
	&\ch{E}_{\emptyset\to LM}\coloneq\left\lvert M\right\rvert\omega_M^{1/2}\Phi_{LM}\omega_M^{1/2}, \\
	&\ch{F}_{S_0L\to S_1S}\fleft[\cdot\fright] \notag\\
	&\coloneq\left(\tr_{S_0}\otimes\ch{F}_{L\to S_1S}^{(0)}\right)\fleft[\op{0}{0}_{S_0}\left(\cdot\right)_{S_0L}\fright] \notag\\
	&\qquad+\left(\tr_{S_0}\otimes\ch{F}_{L\to S_1S}^\perp\right)\fleft[\left(\1_{S_0}-\op{0}{0}_{S_0}\right)\left(\cdot\right)_{S_0L}\fright].
\end{align}
To see that $\ch{N}_{S_0\to S_1SM}$ is a QFGO, it suffices to show that $\ch{F}_{S_0L\to S_1S}$ respects Kelvin's principle [see Eq.~\eqref{eq:kelvin}], which indeed holds as $\ch{F}_{S_0L\to S_1S}[\1_{S_0}/\lvert S_0\rvert\otimes(\cdot)_L]=\1_{S_1}/\lvert S_1\rvert\otimes\gamma_S\otimes\tr_L\fleft[\cdot\fright]$.  The scheme prepares the target state as desired, since, by definition, $\ch{N}_{S_0\to S_1SM}[\op{0}{0}_{S_0}]=\omega_{S_1SM}\approx_\varepsilon\op{0}{0}_{S_1}\otimes\rho_{SM}$.  Moreover, it does so at the work cost of $\beta^{-1}(\ln\lvert S_0\rvert-\ln\lvert S_1\rvert)\approx\beta^{-1}(D_{\max}(\omega_{S_1SM}\|\1_{S_1}/\lvert S_1\rvert\otimes\gamma_S\otimes\omega_M)-\ln\lvert S_1\rvert)\approx\beta^{-1}I_{\max}^{\uparrow,\varepsilon}(\rho_{SM}\|\gamma_S)$, achieving the limit characterized by Eq.~\eqref{eq:formation}.  A more rigorous treatment of the approximate equalities in the calculation can be found in SM~\cite[Sec.~\ref{sec:formation}]{Note1}.

\textbf{\textit{Appendix~D: Optimal scheme for work extraction}}---We explicitly construct a QFGO that achieves the single-shot extractable work in Eq.~\eqref{eq:extraction}.  To extract work from an initial state $(\rho_{SM},\gamma_S)$ up to error $\varepsilon\in[0,1]$, consider the following feedback-control scheme:
\begin{itemize}
	\item[(i)] identify an optimal operator $\Lambda_{SM}$ for the expression of the smoothed min-GMI in Eq.~\eqref{eq:min};
	\item[(ii)] transfer the information in the memory $M$ identically to a link $L$ isomorphic to $M$ and feed it back to the system $S$;
	\item[(iii)] measure a battery $S_0$ alongside $SL$ according to the POVM $(\Omega_{S_0SL},\1_{S_0SL}-\Omega_{S_0SL})$ with the following measurement operator:
	\begin{align}
		\Omega_{S_0SL}&\coloneq\op{0}{0}_{S_0}\otimes\Lambda_{SL} \notag\\
		&\qquad+\op{1}{1}_{S_0}\otimes\1_S\otimes\left(\tfrac{\left\lvert S_0\right\rvert}{\left\lvert S_1\right\rvert}\1_L-\tr_S\fleft[\Lambda_{SL}\gamma_S\fright]\right),
	\end{align}
	where $\lvert S_1\rvert\coloneq\lfloor\lvert S_0\rvert/\lVert\tr_S[\Lambda_{SM}\gamma_S]\rVert_\infty\rfloor$ ensures that $0\leq\Omega_{S_0SL}\leq\1_{S_0SL}$; if the outcome corresponds to $\Omega_{S_0SL}$, then prepare a battery $S_1$ of dimension $\lvert S_1\rvert$ in the pure state $\op{0}{0}_{S_1}$; otherwise, prepare $S_1$ in the state $(\1_{S_1}-\op{0}{0}_{S_1})/(\lvert S_1\rvert-1)$.
\end{itemize}

Formally, this scheme is represented by the channel $\ch{N}_{S_0SM\to S_1}\coloneq\ch{F}_{S_0SL\to S_1}\circ\ch{E}_{M\to L}$, where
\begin{align}
	\ch{E}_{M\to L}\fleft[\cdot\fright]&\coloneq\left(\cdot\right)_{M\to L}, \\
	\ch{F}_{S_0SL\to S_1}\fleft[\cdot\fright]&\coloneq\tr_{S_0SL}\fleft[\Omega_{S_0SL}\left(\cdot\right)_{S_0SL}\fright]\otimes\op{0}{0}_{S_1} \notag\\
	&\qquad+\tr_{S_0SL}\fleft[\left(\1_{S_0SL}-\Omega_{S_0SL}\right)\left(\cdot\right)_{S_0SL}\fright] \notag\\
	&\qquad\otimes\tfrac{\1_{S_1}-\op{0}{0}_{S_1}}{\left\lvert S_1\right\rvert-1}.
\end{align}
To see that $\ch{N}_{S_0SM\to S_1}$ is a QFGO, it suffices to show that $\ch{F}_{S_0SL\to S_1}$ respects Kelvin's principle [see Eq.~\eqref{eq:kelvin}], which indeed holds as $\ch{F}_{S_0SL\to S_1}[\1_{S_0}/\lvert S_0\rvert\otimes\gamma_S\otimes(\cdot)_L]=\1_{S_1}/\lvert S_1\rvert\otimes\tr_L\fleft[\cdot\fright]$.  The scheme prepares the battery in the desired pure state, since, by definition, $\ch{N}_{S_0SM\to S_1}[\op{0}{0}_{S_0}\otimes\rho_{SM}]=\tr[\Lambda_{SM}\rho_{SM}]\op{0}{0}_{S_1}+(1-\tr[\Lambda_{SM}\rho_{SM}])(\1_{S_1}-\op{0}{0}_{S_1})/(\lvert S_1\rvert-1)\approx_\varepsilon\op{0}{0}_{S_1}$.  Moreover, it extracts work $\beta^{-1}(\ln\lvert S_1\rvert-\ln\lvert S_0\rvert)\approx-\beta^{-1}\ln\lVert\tr_S[\Lambda_{SM}\gamma_S]\rVert_\infty=\beta^{-1}I_{\min}^{\downarrow,\varepsilon}(\rho_{SM}\|\gamma_S)$, achieving the limit characterized by Eq.~\eqref{eq:extraction}.  A more rigorous treatment of the approximate equalities in the calculation can be found in SM~\cite[Sec.~\ref{sec:extraction}]{Note1}.

\textbf{\textit{Appendix~E: Bounds for general conversion}}---We derive upper and lower bounds on the single-shot work cost under QFGOs for general conversion (see SM~\cite[Sec.~\ref{sec:conversion}]{Note1} for a proof).

\begin{proposition}
\label{prop:bounds}
For an error parameter $\varepsilon\in[0,1)$, an initial state $(\rho_{SM},\gamma_S)$, and a target state $(\rho'_{S'M'},\gamma'_{S'})$, the single-shot work cost under QFGOs is bounded as
\begin{align}
	&\beta^{-1}\max\left\{A,B\right\}\leq W^\varepsilon\fleft(\fleft(\rho_{SM},\gamma_S\vphantom{\rho'_{S'M'},\gamma'_{S'}}\fright)\to\fleft(\rho'_{S'M'},\gamma'_{S'}\fright)\fright) \\
	&\leq\beta^{-1}\inf_{\eta\in[0,\varepsilon]}\left(I_{\max}^{\uparrow,\varepsilon-\eta}\fleft(\rho'_{S'M'}\middle\|\gamma'_{S'}\fright)-I_{\min}^{\downarrow,\eta}\fleft(\rho_{SM}\middle\|\gamma_S\fright)\right),
\end{align}
where
\begin{align}
	A&\equiv\sup_{\alpha\in(\frac{1}{2},1)}\left(\sw{I}_\alpha^\downarrow\fleft(\rho'_{S'M'}\middle\|\gamma'_{S'}\fright)-\sw{I}_\frac{\alpha}{2\alpha-1}^\downarrow\fleft(\rho_{SM}\middle\|\gamma_S\fright)-\tfrac{2\alpha}{1-\alpha}\ln\tfrac{1}{1-\varepsilon}\right), \\
	B&\equiv\sup_{\alpha\in(0,1)}\left(\pz{I}_\alpha^\downarrow\fleft(\rho'_{S'M'}\middle\|\gamma'_{S'}\fright)-\pz{I}_{2-\alpha}^\downarrow\fleft(\rho_{SM}\middle\|\gamma_S\fright)-\tfrac{2}{1-\alpha}\ln\tfrac{1}{1-\varepsilon}\right),
\end{align}
with $\sw{I}_\alpha^{\downarrow}$ and $\pz{I}_\alpha^{\downarrow}$ the (optimized) Rényi GMIs~\cite{hayashi2016CorrelationDetectionOperational} based on the sandwiched~\cite{muller-lennert2013QuantumRenyiEntropies,wilde2014StrongConverseClassical} and the Petz-type~\cite{petz1985QuasientropiesStatesNeumann,petz1986QuasientropiesFiniteQuantum} Rényi relative entropies.
\end{proposition}

The upper and lower bounds in Proposition~\ref{prop:bounds} coincide in the asymptotic limit (see SM~\cite[Sec.~\ref{sec:conversion}]{Note1} for a proof), leading to Theorem~\ref{thm:asymptotic}.

%-------------------------------
%	SUPPLEMENTAL MATERIAL
%-------------------------------

\onecolumngrid
\clearpage

\setcounter{page}{1}
\setcounter{equation}{0}
\setcounter{definition}{0}
\setcounter{figure}{0}
\setcounter{remark}{0}
\setcounter{theorem}{0}
\renewcommand{\theequation}{S\arabic{equation}}
\renewcommand{\thedefinition}{S\arabic{definition}}
\renewcommand{\thefigure}{S\arabic{figure}}
\renewcommand{\theremark}{S\arabic{remark}}
\renewcommand{\thetheorem}{S\arabic{theorem}}

\let\addcontentsline\oldaddcontentsline
\setcounter{secnumdepth}{3}

\makeatletter
\renewcommand\normalsize{\@setfontsize\normalsize{12}{14}}
\renewcommand\small{\@setfontsize\small{11}{13}}
\renewcommand\footnotesize{\@setfontsize\footnotesize{10}{12}}
\renewcommand\large{\@setfontsize\large{14}{17}}
\renewcommand\Large{\@setfontsize\Large{17}{20}}
\makeatother
\normalsize

\begin{center}
\textbf{\large Supplemental Material: Fundamental Limits for Thermodynamic Control with Quantum Feedback}
\end{center}

\tableofcontents

\newcounter{temporary}
\begin{comment}
{
\setcounter{temporary}{\value{theorem}}
\setcounter{theorem}{100}
\renewcommand{\thetheorem}{\ref{thm:xxx}}
\begin{theorem}[Restatement]
xxx
\end{theorem}
\setcounter{theorem}{\value{temporary}}
}
\end{comment}

\section{Quantum-feedback-assisted Gibbs-preserving operations (QFGO\lowercase{s})}
\label{sec:operation}

Our setting concerns the thermodynamic control of a quantum system by a controller who has access to a quantum memory.  The system is in contact with a heat bath at a constant inverse temperature $\beta$, and the resource of interest is the isothermal work content of the system.  Consider a system $S$ and a memory $M$.  Denoting the Hamiltonian of $S$ by $\h{H}_S$, the Gibbs state of $S$ is given by $\gamma_S\equiv e^{-\beta\h{H}_S}/\tr[e^{-\beta\h{H}_S}]$.  The thermalization channel on $S$ is given by
\begin{align}
	\ch{R}_{S\to S}^\gamma[\cdot]&\equiv\gamma_S\otimes\tr_S[\cdot]. \label{eq:thermalization}
\end{align}
Since we are interested in the work content of $S$, the energetics of $M$ is irrelevant for our purpose, and thus we need not specify a Hamiltonian for $M$.  The actual state of $S$ and $M$ is represented by a joint quantum state $\rho_{SM}$.  For brevity, we henceforth refer to the pair $(\rho_{SM},\gamma_S)$ as the \emph{state} of $S$ and $M$.

\begin{definition}[QFGO]
\label{def:operation}
A quantum channel $\ch{N}_{SM\to S'M'}$ is called a \emph{quantum-feedback-assisted Gibbs-preserving operation (QFGO)} whenever there exist two channels $\ch{E}_{M\to LM'}$ and $\ch{F}_{SL\to S'}$ such that
\begin{align}
	\ch{N}_{SM\to S'M'}&=\ch{F}_{SL\to S'}\circ\ch{E}_{M\to LM'}, \label{eq:operation-1}\\
	\ch{F}_{SL\to S'}\fleft[\gamma_S\otimes\left(\cdot\right)_L\fright]&=\gamma'_{S'}\otimes\tr_L\fleft[\cdot\fright], \label{eq:operation-2}
\end{align}
with $\gamma_S$ and $\gamma'_{S'}$ the Gibbs states of $S$ and $S'$, respectively.
\end{definition}

In our resource theory of \emph{quantum conditional athermality}, we deem QFGOs as the free operations.  States that can be generated from triviality by QFGOs are precisely those of the form $(\gamma_S\otimes\rho_M,\gamma_S)$, with the system in thermal equilibrium and the memory containing no side information.  We call such states \emph{conditional Gibbs states}, and they are the free states of our resource theory.

\subsection{Equivalent conditions}
\label{sec:equivalent}

\begin{proposition}[Equivalent conditions for QFGOs; restatement of Proposition~\ref{prop:operation}]
\label{prop:operation-supp}
For a channel $\ch{N}_{SM\to S'M'}$, the following statements are equivalent.
\begin{itemize}
	\item[(i)] $\ch{N}_{SM\to S'M'}$ is a QFGO.
	\item[(ii)] $\ch{N}_{SM\to S'M'}$ is thermalization covariant on the system:
	\begin{align}
		\ch{N}_{SM\to S'M'}\circ\ch{R}_{S\to S}^\gamma&=\ch{R}_{S'\to S'}^{\gamma'}\circ\ch{N}_{SM\to S'M'}. \label{eq:operation-3}
	\end{align}
	\item[(iii)] $\ch{N}_{SM\to S'M'}$ is system-to-memory nonsignaling and conditionally Gibbs preserving: there exists a channel $\ch{E}_{M\to M'}$ such that
	\begin{align}
		\tr_{S'}\circ\ch{N}_{SM\to S'M'}&=\tr_S\otimes\ch{E}_{M\to M'}, \label{eq:operation-4}\\
		\ch{N}_{SM\to S'M'}\fleft[\gamma_S\otimes\left(\cdot\right)_M\fright]&=\gamma'_{S'}\otimes\ch{E}_{M\to M'}\fleft[\cdot\fright]. \label{eq:operation-5}
	\end{align}
\end{itemize}
\end{proposition}

\begin{proof}
See the companion paper~\cite[Proposition~2]{ji2026QuantumConditionalRelative}.
\end{proof}

\subsection{Petz recovery maps and coarse graining}
\label{sec:recovery}

For a channel $\ch{E}_{S\to S'}$ and a state $\sigma_S$, the Petz recovery map is defined as~\cite{petz1986SufficientSubalgebrasRelative,petz1988SufficiencyChannelsNeumann}
\begin{align}
	\ch{P}_{S'\to S}^{\ch{E},\sigma}\fleft[\cdot\fright]&\coloneq\sigma_S^\frac{1}{2}\ch{E}_{S'\to S}^\dagger\fleft[\left(\ch{E}_{S\to S'}\fleft[\sigma_S\fright]\right)^{-\frac{1}{2}}\left(\cdot\right)_{S'}\left(\ch{E}_{S\to S'}\fleft[\sigma_S\fright]\right)^{-\frac{1}{2}}\fright]\sigma_S^\frac{1}{2}. \label{eq:recovery-1}
\end{align}
For a real number $t\in\spa{R}$, the rotated Petz recovery map is defined as~\cite{wilde2015RecoverabilityQuantumInformation}
\begin{align}
	\ch{P}_{S'\to S}^{\ch{E},\sigma,t}&\coloneq\ch{U}_{S\to S}^{\sigma,t}\circ\ch{P}_{S'\to S}^{\ch{E},\sigma}\circ\ch{U}_{S'\to S'}^{\ch{E}\fleft[\sigma\fright],-t}, \label{eq:rotation-1}
\end{align}
where
\begin{align}
	\ch{U}_{S\to S}^{\sigma,t}\fleft[\cdot\fright]&\equiv\sigma_S^{it}\left(\cdot\right)_S\sigma_S^{-it} \label{eq:rotation-2}
\end{align}
and $\ch{U}_{S'\to S'}^{\ch{E}\fleft[\sigma\fright],-t}$ are the unitary channels for time translation on $S$ and $S'$, respectively.  The rotated Petz recovery map has the following property:
\begin{align}
	\ch{P}_{S'\to S}^{\ch{E},\sigma,t}\fleft[\ch{E}_{S\to S'}\fleft[\sigma_S\fright]\fright]&=\sigma_S. \label{eq:perfect}
\end{align}

\begin{proposition}[Petz recovery map for a QFGO]
\label{prop:recovery}
For a QFGO $\ch{N}_{SM\to S'M'}$ and a conditional Gibbs state $(\gamma_S\otimes\rho_M,\gamma_S)$, the Petz recovery map $\ch{P}_{S'M'\to SM}^{\ch{N},\gamma\otimes\rho}$ is a QFGO.  Moreover, it satisfies
\begin{align}
	\tr_S\circ\ch{P}_{S'M'\to SM}^{\ch{N},\gamma\otimes\rho}&=\tr_{S'}\otimes\ch{P}_{M'\to M}^{\ch{E},\rho}, \label{eq:recovery-2}\\
	\ch{P}_{S'M'\to SM}^{\ch{N},\gamma\otimes\rho}\fleft[\gamma'_{S'}\otimes\left(\cdot\right)_{M'}\fright]&=\gamma_S\otimes\ch{P}_{M'\to M}^{\ch{E},\rho}\fleft[\cdot\fright], \label{eq:recovery-3}
\end{align}
where
\begin{align}
	\ch{E}_{M\to M'}\fleft[\cdot\fright]&\equiv\ch{N}_{SM\to M'}\fleft[\gamma_S\otimes\left(\cdot\right)_M\fright]. \label{eq:recovery-4}
\end{align}
\end{proposition}

\begin{proof}
Due to the equivalence between statements~(i) and (iii) of Proposition~\ref{prop:operation-supp}, it suffices to prove Eqs.~\eqref{eq:recovery-2} and \eqref{eq:recovery-3}.  Taking the adjoint of Eq.~\eqref{eq:operation-5}, we have that
\begin{align}
	\tr_S\fleft[\gamma_S^\frac{1}{2}\ch{N}_{S'M'\to SM}^\dagger\fleft[\cdot\fright]\gamma_S^\frac{1}{2}\fright]&=\left(\tr_{S'}\otimes\ch{E}_{M'\to M}^\dagger\right)\fleft[\gamma_{S'}^{\prime\frac{1}{2}}\left(\cdot\right)_{S'M'}\gamma_{S'}^{\prime\frac{1}{2}}\fright]. \label{pf:recovery-1}
\end{align}
It follows from the definition of the Petz recovery map [Eq.~\eqref{eq:recovery-1}] that
\begin{align}
	&\left(\tr_S\circ\ch{P}_{S'M'\to SM}^{\ch{N},\gamma\otimes\rho}\right)\fleft[\cdot\fright] \notag\\
	&=\tr_S\fleft[\left(\gamma_S^\frac{1}{2}\otimes\rho_M^\frac{1}{2}\right)\ch{N}_{S'M'\to SM}^\dagger\fleft[\left(\ch{N}_{SM\to S'M'}\fleft[\gamma_S\otimes\rho_M\fright]\right)^{-\frac{1}{2}}\left(\cdot\right)_{S'M'}\left(\ch{N}_{SM\to S'M'}\fleft[\gamma_S\otimes\rho_M\fright]\right)^{-\frac{1}{2}}\fright]\left(\gamma_S^\frac{1}{2}\otimes\rho_M^\frac{1}{2}\right)\fright] \\
	&=\tr_S\fleft[\left(\gamma_S^\frac{1}{2}\otimes\rho_M^\frac{1}{2}\right)\ch{N}_{S'M'\to SM}^\dagger\fleft[\left(\gamma_{S'}^{\prime-\frac{1}{2}}\otimes\left(\ch{E}_{M\to M'}\fleft[\rho_M\fright]\right)^{-\frac{1}{2}}\right)\left(\cdot\right)_{S'M'}\left(\gamma_{S'}^{\prime-\frac{1}{2}}\otimes\left(\ch{E}_{M\to M'}\fleft[\rho_M\fright]\right)^{-\frac{1}{2}}\right)\fright]\left(\gamma_S^\frac{1}{2}\otimes\rho_M^\frac{1}{2}\right)\fright] \label{pf:recovery-2}\\
	&=\rho_M^\frac{1}{2}\left(\tr_{S'}\otimes\ch{E}_{M'\to M}^\dagger\right)\fleft[\left(\ch{E}_{M\to M'}\fleft[\rho_M\fright]\right)^{-\frac{1}{2}}\left(\cdot\right)_{S'M'}\left(\ch{E}_{M\to M'}\fleft[\rho_M\fright]\right)^{-\frac{1}{2}}\fright]\rho_M^\frac{1}{2} \label{pf:recovery-3}\\
	&=\left(\tr_{S'}\otimes\ch{P}_{M'\to M}^{\ch{E},\rho}\right)\fleft[\cdot\fright], \label{pf:recovery-4}
\end{align}
with $\gamma'_{S'}$ the Gibbs state of $S'$.  Here Eq.~\eqref{pf:recovery-2} follows from Eqs.~\eqref{eq:operation-5} and \eqref{eq:recovery-4}; Eq.~\eqref{pf:recovery-3} follows from Eq.~\eqref{pf:recovery-1}; Eq.~\eqref{pf:recovery-4} follows from the definition of the Petz recovery map [Eq.~\eqref{eq:recovery-1}].  This recovers Eq.~\eqref{eq:recovery-2}.  Furthermore, taking the adjoint of Eq.~\eqref{eq:operation-4}, we have that
\begin{align}
	\ch{N}_{S'M'\to SM}^\dagger\fleft[\1_{S'}\otimes\left(\cdot\right)_{M'}\fright]&=\1_S\otimes\ch{E}_{M'\to M}^\dagger\fleft[\cdot\fright]. \label{pf:recovery-5}
\end{align}
It follows from the definition of the Petz recovery map [Eq.~\eqref{eq:recovery-1}] that
\begin{align}
	&\ch{P}_{S'M'\to SM}^{\ch{N},\gamma\otimes\rho}\fleft[\gamma'_{S'}\otimes\left(\cdot\right)_{M'}\fright] \notag\\
	&=\left(\gamma_S^\frac{1}{2}\otimes\rho_M^\frac{1}{2}\right)\ch{N}_{S'M'\to SM}^\dagger\fleft[\left(\ch{N}_{SM\to S'M'}\fleft[\gamma_S\otimes\rho_M\fright]\right)^{-\frac{1}{2}}\left(\gamma'_{S'}\otimes\left(\cdot\right)_{M'}\right)\left(\ch{N}_{SM\to S'M'}\fleft[\gamma_S\otimes\rho_M\fright]\right)^{-\frac{1}{2}}\fright]\left(\gamma_S^\frac{1}{2}\otimes\rho_M^\frac{1}{2}\right) \\
	&=\left(\gamma_S^\frac{1}{2}\otimes\rho_M^\frac{1}{2}\right)\ch{N}_{S'M'\to SM}^\dagger\fleft[\left(\gamma_{S'}^{\prime-\frac{1}{2}}\otimes\left(\ch{E}_{M\to M'}\fleft[\rho_M\fright]\right)^{-\frac{1}{2}}\right)\left(\gamma'_{S'}\otimes\left(\cdot\right)_{M'}\right)\left(\gamma_{S'}^{\prime-\frac{1}{2}}\otimes\left(\ch{E}_{M\to M'}\fleft[\rho_M\fright]\right)^{-\frac{1}{2}}\right)\fright]\left(\gamma_S^\frac{1}{2}\otimes\rho_M^\frac{1}{2}\right) \label{pf:recovery-6}\\
	&=\left(\gamma_S^\frac{1}{2}\otimes\rho_M^\frac{1}{2}\right)\ch{N}_{S'M'\to SM}^\dagger\fleft[\1_{S'}\otimes\left(\ch{E}_{M\to M'}\fleft[\rho_M\fright]\right)^{-\frac{1}{2}}\left(\cdot\right)_{M'}\left(\ch{E}_{M\to M'}\fleft[\rho_M\fright]\right)^{-\frac{1}{2}}\fright]\left(\gamma_S^\frac{1}{2}\otimes\rho_M^\frac{1}{2}\right) \\
	&=\left(\gamma_S^\frac{1}{2}\otimes\rho_M^\frac{1}{2}\right)\left(\1_S\otimes\ch{E}_{M'\to M}^\dagger\fleft[\left(\ch{E}_{M\to M'}\fleft[\rho_M\fright]\right)^{-\frac{1}{2}}\left(\cdot\right)_{M'}\left(\ch{E}_{M\to M'}\fleft[\rho_M\fright]\right)^{-\frac{1}{2}}\fright]\right)\left(\gamma_S^\frac{1}{2}\otimes\rho_M^\frac{1}{2}\right) \label{pf:recovery-7}\\
	&=\gamma_S\otimes\ch{P}_{M'\to M}^{\ch{E},\rho}\fleft[\cdot\fright]. \label{pf:recovery-8}
\end{align}
Here Eq.~\eqref{pf:recovery-6} follows from Eqs.~\eqref{eq:operation-5} and \eqref{eq:recovery-4}; Eq.~\eqref{pf:recovery-7} follows from Eq.~\eqref{pf:recovery-5}; Eq.~\eqref{pf:recovery-8} follows from the definition of the Petz recovery map [Eq.~\eqref{eq:recovery-1}].  This recovers Eq.~\eqref{eq:recovery-3}.
\end{proof}

\begin{corollary}[Rotated Petz recovery map for a QFGO]
\label{cor:recovery}
For a QFGO $\ch{N}_{SM\to S'M'}$, a conditional Gibbs state $(\gamma_S\otimes\rho_M,\gamma_S)$, and a real number $t\in\spa{R}$, the rotated Petz recovery map $\ch{P}_{S'M'\to SM}^{\ch{N},\gamma\otimes\rho,t}$ is a QFGO.  Moreover, it satisfies
\begin{align}
	\tr_S\circ\ch{P}_{S'M'\to SM}^{\ch{N},\gamma\otimes\rho,t}&=\tr_{S'}\otimes\ch{P}_{M'\to M}^{\ch{E},\rho,t}, \label{eq:rotation-3}\\
	\ch{P}_{S'M'\to SM}^{\ch{N},\gamma\otimes\rho,t}\fleft[\gamma'_{S'}\otimes\left(\cdot\right)_{M'}\fright]&=\gamma_S\otimes\ch{P}_{M'\to M}^{\ch{E},\rho,t}\fleft[\cdot\fright], \label{eq:rotation-4}
\end{align}
with $\ch{E}_{M\to M'}$ defined in Eq.~\eqref{eq:recovery-4}.
\end{corollary}

\begin{proof}
It follows from the definition of the rotated Petz recovery map [Eq.~\eqref{eq:rotation-1}] that
\begin{align}
	\ch{P}_{S'M'\to SM}^{\ch{N},\gamma\otimes\rho,t}&=\ch{U}_{SM\to SM}^{\gamma\otimes\rho,t}\circ\ch{P}_{S'M'\to SM}^{\ch{N},\gamma\otimes\rho}\circ\ch{U}_{S'M'\to S'M'}^{\ch{N}\fleft[\gamma\otimes\rho\fright],-t} \\
	&=\left(\ch{U}_{S\to S}^{\gamma,t}\otimes\ch{U}_{M\to M}^{\rho,t}\right)\circ\ch{P}_{S'M'\to SM}^{\ch{N},\gamma\otimes\rho}\circ\left(\ch{U}_{S'\to S'}^{\gamma',-t}\otimes\ch{U}_{M'\to M'}^{\ch{E}\fleft[\rho\fright],-t}\right), \label{pf:rotation-1}
\end{align}
with $\gamma'_{S'}$ the Gibbs state of $S'$.  Here Eq.~\eqref{pf:rotation-1} follows from Eqs.~\eqref{eq:operation-5}, \eqref{eq:recovery-4}, and \eqref{eq:rotation-2}.  This implies that
\begin{align}
	\tr_S\circ\ch{P}_{S'M'\to SM}^{\ch{N},\gamma\otimes\rho,t}&=\tr_S\circ\left(\ch{U}_{S\to S}^{\gamma,t}\otimes\ch{U}_{M\to M}^{\rho,t}\right)\circ\ch{P}_{S'M'\to SM}^{\ch{N},\gamma\otimes\rho}\circ\left(\ch{U}_{S'\to S'}^{\gamma',-t}\otimes\ch{U}_{M'\to M'}^{\ch{E}\fleft[\rho\fright],-t}\right) \\
	&=\left(\tr_S\otimes\ch{U}_{M\to M}^{\rho,t}\right)\circ\ch{P}_{S'M'\to SM}^{\ch{N},\gamma\otimes\rho}\circ\left(\ch{U}_{S'\to S'}^{\gamma',-t}\otimes\ch{U}_{M'\to M'}^{\ch{E}\fleft[\rho\fright],-t}\right) \label{pf:rotation-2}\\
	&=\ch{U}_{M\to M}^{\rho,t}\circ\left(\tr_{S'}\otimes\ch{P}_{M'\to M}^{\ch{E},\rho}\right)\circ\left(\ch{U}_{S'\to S'}^{\gamma',-t}\otimes\ch{U}_{M'\to M'}^{\ch{E}\fleft[\rho\fright],-t}\right) \label{pf:rotation-3}\\
	&=\tr_{S'}\otimes\left(\ch{U}_{M\to M}^{\rho,t}\circ\ch{P}_{M'\to M}^{\ch{E},\rho}\circ\ch{U}_{M'\to M'}^{\ch{E}\fleft[\rho\fright],-t}\right) \label{pf:rotation-4}\\
	&=\tr_{S'}\otimes\ch{P}_{M'\to M}^{\ch{E},\rho,t}. \label{pf:rotation-5}
\end{align}
Here Eqs.~\eqref{pf:rotation-2} and \eqref{pf:rotation-4} follow from the facts that $\tr_S\circ\ch{U}_{S\to S}^{\gamma,t}=\tr_S$ and $\tr_{S'}\circ\ch{U}_{S'\to S'}^{\gamma',-t}=\tr_{S'}$ due to Eq.~\eqref{eq:rotation-2}; Eq.~\eqref{pf:rotation-3} follows from Eq.~\eqref{eq:recovery-2}; Eq.~\eqref{pf:rotation-5} follows from the definition of the rotated Petz recovery map [Eq.~\eqref{eq:rotation-1}].  This recovers Eq.~\eqref{eq:rotation-3}.  Furthermore, it follows from Eq.~\eqref{pf:rotation-1} that
\begin{align}
	\ch{P}_{S'M'\to SM}^{\ch{N},\gamma\otimes\rho,t}\fleft[\gamma'_{S'}\otimes\left(\cdot\right)_{M'}\fright]&=\left(\left(\ch{U}_{S\to S}^{\gamma,t}\otimes\ch{U}_{M\to M}^{\rho,t}\right)\circ\ch{P}_{S'M'\to SM}^{\ch{N},\gamma\otimes\rho}\circ\left(\ch{U}_{S'\to S'}^{\gamma',-t}\otimes\ch{U}_{M'\to M'}^{\ch{E}\fleft[\rho\fright],-t}\right)\right)\fleft[\gamma'_{S'}\otimes\left(\cdot\right)_{M'}\fright] \\
	&=\left(\left(\ch{U}_{S\to S}^{\gamma,t}\otimes\ch{U}_{M\to M}^{\rho,t}\right)\circ\ch{P}_{S'M'\to SM}^{\ch{N},\gamma\otimes\rho}\right)\fleft[\gamma'_{S'}\otimes\ch{U}_{M'\to M'}^{\ch{E}\fleft[\rho\fright],-t}\fleft[\cdot\fright]\fright] \label{pf:rotation-6}\\
	&=\left(\ch{U}_{S\to S}^{\gamma,t}\otimes\ch{U}_{M\to M}^{\rho,t}\right)\fleft[\gamma_S\otimes\left(\ch{P}_{M'\to M}^{\ch{E},\rho}\circ\ch{U}_{M'\to M'}^{\ch{E}\fleft[\rho\fright],-t}\right)\fleft[\cdot\fright]\fright] \label{pf:rotation-7}\\
	&=\gamma_S\otimes\left(\ch{U}_{M\to M}^{\rho,t}\circ\ch{P}_{M'\to M}^{\ch{E},\rho}\circ\ch{U}_{M'\to M'}^{\ch{E}\fleft[\rho\fright],-t}\right)\fleft[\cdot\fright] \label{pf:rotation-8}\\
	&=\gamma_S\otimes\ch{P}_{M'\to M}^{\ch{E},\rho,t}\fleft[\cdot\fright]. \label{pf:rotation-9}
\end{align}
Here Eqs.~\eqref{pf:rotation-6} and \eqref{pf:rotation-8} follow from the facts that $\ch{U}_{S'\to S'}^{\gamma',-t}[\gamma'_{S'}]=\gamma'_{S'}$ and $\ch{U}_{S\to S}^{\gamma,t}[\gamma_S]=\gamma_S$ due to Eq.~\eqref{eq:rotation-2}; Eq.~\eqref{pf:rotation-7} follows from Eq.~\eqref{eq:recovery-3}; Eq.~\eqref{pf:rotation-9} follows from the definition of the rotated Petz recovery map [Eq.~\eqref{eq:rotation-1}].  This recovers Eq.~\eqref{eq:rotation-4}.
\end{proof}

\begin{remark}[QFGOs as coarse graining]
\label{rem:recovery}
Consider a system and a memory observed by two observers, Alice and Bob, with Bob's perspective being a coarse-grained version of Alice's.  The system and the memory are represented by $S$ and $M$, respectively, from Alice's perspective and by $S'$ and $M'$ from Bob's perspective.  The coarse graining, which translates Alice's perspective to Bob's, is represented by a channel $\ch{N}_{SM\to S'M'}$.  Assuming that the Gibbs states $\gamma_S$ and $\gamma'_{S'}$ are known to both Alice and Bob, a minimal requirement for coarse graining is that it should maintain a consistent view of thermalization between the two perspectives.  Specifically, it should make no difference whether one (i) observes the system thermalizing from Alice's perspective and then shifts to Bob's perspective, or (ii) shifts to Bob's perspective first and then observes the system thermalizing there.  This implies that $\ch{N}_{SM\to S'M'}$ is thermalization covariant on the system in the sense of Eq.~\eqref{eq:operation-3} and is thus a QFGO due to Proposition~\ref{prop:operation-supp}.  Another interesting question is how Bob should reconstruct what Alice observes from what he observes.  Following a similar idea to Ref.~\cite[Sec.~III~E]{faist2018FundamentalWorkCost}, one can argue that a suitable choice is for Bob to apply the Petz recovery map $\ch{P}_{S'M'\to SM}^{\ch{N},\gamma\otimes\rho}$, with $\rho_M$ a prior of the memory's marginal state from Alice's perspective.  The fact that $\ch{P}_{S'M'\to SM}^{\ch{N},\gamma\otimes\rho}$ is also a QFGO due to Proposition~\ref{prop:recovery} ensures that Bob's reconstruction channel likewise maintains a consistent view of thermalization between Alice's and Bob's perspectives.  On a related note, similar observations have been made in several other resource theories~\cite{marvian2016ClocksClonersCatalytic,lemm2017InformationtheoreticLimitationsApproximate,lami2018ApproximateReversalQuantum,alhambra2018WorkReversibilityQuantum}, where the Petz recovery map of a free operation remains a free operation. 
\end{remark}

\section{Generalized mutual informations (GMI\lowercase{s})}
\label{sec:information}

Let $\psd$ denote the set of positive semidefinite operators, $\pd$ the set of positive definite operators, $\sta\equiv\{\omega\in\psd\colon\tr[\omega]=1\}$ the set of normalized states, and $\sub\equiv\{\tau\in\psd\colon\tr[\tau]\leq1\}$ the set of subnormalized states.

For $\tau,\rho\in\sub$, the \emph{generalized trace distance} is defined as~\cite{tomamichel2010DualitySmoothMin}
\begin{align}
	\delta\fleft(\tau,\rho\fright)&\coloneq\frac{1}{2}\left(\left\lVert\tau-\rho\right\rVert_1+\left\lvert\tr\fleft[\tau\fright]-\tr\fleft[\rho\fright]\right\rvert\right) \label{eq:distance}\\
	&=\sup_{\substack{\Lambda\in\psd, \\ t\in[0,1]}}\left\{\tr\fleft[\Lambda\left(\tau-\rho\right)\fright]+t\left(\tr\fleft[\rho\fright]-\tr\fleft[\tau\fright]\right)\colon0\leq\Lambda\leq\1\right\}. \label{eq:dual}
\end{align}

For $\rho\in\sta$ and $\sigma\in\psd$ such that $\supp(\rho)\subseteq\supp(\sigma)$, the \emph{max-}~\cite{datta2009MinMaxrelativeEntropies}, \emph{(Petz-type) min-}~\cite{datta2009MinMaxrelativeEntropies}, and \emph{Umegaki}~\cite{umegaki1962ConditionalExpectationOperator} \emph{relative entropies} are defined, respectively, as
\begin{align}
	D_{\max}\fleft(\rho\middle\|\sigma\fright)&\coloneq\ln\left\lVert\sigma^{-\frac{1}{2}}\rho\sigma^{-\frac{1}{2}}\right\rVert_\infty, \label{eq:max-relative}\\
	D_{\min}\fleft(\rho\middle\|\sigma\fright)&\coloneq-\ln\tr\fleft[\rho^0\sigma\fright], \label{eq:min-relative}\\
	D\fleft(\rho\middle\|\sigma\fright)&\coloneq\tr\fleft[\rho\ln\rho\fright]-\tr\fleft[\rho\ln\sigma\fright]. \label{eq:umegaki-relative}
\end{align}
The \emph{sandwiched Rényi relative entropy} for $\alpha\in[1/2,1)\cup(1,\infty)$ is defined as~\cite{muller-lennert2013QuantumRenyiEntropies,wilde2014StrongConverseClassical}
\begin{align}
	\sw{D}_\alpha\fleft(\rho\middle\|\sigma\fright)&\coloneq\frac{1}{\alpha-1}\ln\tr\fleft[\left(\sigma^\frac{1-\alpha}{2\alpha}\rho\sigma^\frac{1-\alpha}{2\alpha}\right)^\alpha\fright]. \label{eq:sandwiched-relative}
\end{align}
The \emph{Petz-type Rényi relative entropy} for $\alpha\in[0,1)\cup(1,2]$ is defined as~\cite{petz1985QuasientropiesStatesNeumann,petz1986QuasientropiesFiniteQuantum}
\begin{align}
	\pz{D}_\alpha\fleft(\rho\middle\|\sigma\fright)&\coloneq\frac{1}{\alpha-1}\ln\tr\fleft[\rho^\alpha\sigma^{1-\alpha}\fright]. \label{eq:petz-relative}
\end{align}
For $\varepsilon\in[0,1]$, the \emph{hypothesis-testing relative entropy} is defined as~\cite{buscemi2010QuantumCapacityChannels,wang2012OneshotClassicalquantumCapacity}
\begin{align}
	D_{\min}^\varepsilon\fleft(\rho\middle\|\sigma\fright)&\coloneq\sup_{\Lambda\in\psd}\left\{-\ln\tr\fleft[\Lambda\sigma\fright]\colon\tr\fleft[\Lambda\rho\fright]\geq1-\varepsilon,\;0\leq\Lambda\leq\1\right\}. \label{eq:hypothesis-relative}
\end{align}
For $\rho_{SM}\in\sta$ and $\gamma_S\in\pd$, the \emph{(unoptimized) max-GMI} is defined as~\cite{hayashi2016CorrelationDetectionOperational}
\begin{align}
	I_{\max}^\uparrow\fleft(\rho_{SM}\middle\|\gamma_S\fright)&\coloneq D_{\max}\fleft(\rho_{SM}\middle\|\gamma_S\otimes\rho_M\fright) \\
	&=\ln\left\lVert\left(\gamma_S^{-\frac{1}{2}}\otimes\rho_M^{-\frac{1}{2}}\right)\rho_{SM}\left(\gamma_S^{-\frac{1}{2}}\otimes\rho_M^{-\frac{1}{2}}\right)\right\rVert_\infty. \label{eq:max-information}
\end{align}
The \emph{(optimized) (Petz-type) min-GMI} is defined as~\cite{hayashi2016CorrelationDetectionOperational}
\begin{align}
	I_{\min}^\downarrow\fleft(\rho_{SM}\middle\|\gamma_S\fright)&\coloneq\inf_{\sigma_M\in\sta}D_{\min}\fleft(\rho_{SM}\middle\|\gamma_S\otimes\sigma_M\fright) \\
	&=-\ln\left\lVert\tr_S\fleft[\rho_{SM}^0\gamma_S\fright]\right\rVert_\infty. \label{eq:min-information}
\end{align}
The \emph{Umegaki GMI} is defined as~\cite{hayashi2016CorrelationDetectionOperational,morris2019AssistedWorkDistillation}
\begin{align}
	I\fleft(\rho_{SM}\middle\|\gamma_S\fright)&\coloneq D\fleft(\rho_{SM}\middle\|\gamma_S\otimes\rho_M\fright) \label{eq:unoptimized-umegaki-information}\\
	&=\inf_{\sigma_M\in\sta}D\fleft(\rho_{SM}\middle\|\gamma_S\otimes\sigma_M\fright) \label{eq:optimized-umegaki-information}\\
	&=\tr\fleft[\rho_{SM}\ln\rho_{SM}\fright]-\tr\fleft[\rho_S\ln\gamma_S\fright]-\tr\fleft[\rho_M\ln\rho_M\fright]. \label{eq:umegaki-information}
\end{align}
The \emph{(optimized) sandwiched} and \emph{Petz-type Rényi GMIs} are defined, respectively, as~\cite{hayashi2016CorrelationDetectionOperational}
\begin{align}
	\sw{I}_\alpha^\downarrow\fleft(\rho_{SM}\middle\|\gamma_S\fright)&\coloneq\inf_{\sigma_M\in\sta}\sw{D}_\alpha\fleft(\rho_{SM}\middle\|\gamma_S\otimes\sigma_M\fright), \label{eq:sandwiched-information}\\
	\pz{I}_\alpha^\downarrow\fleft(\rho_{SM}\middle\|\gamma_S\fright)&\coloneq\inf_{\sigma_M\in\sta}\pz{D}_\alpha\fleft(\rho_{SM}\middle\|\gamma_S\otimes\sigma_M\fright). \label{eq:petz-information}
\end{align}

We define the \emph{smoothed (unoptimized) max-GMI} as
\begin{align}
	I_{\max}^{\uparrow,\varepsilon}\fleft(\rho_{SM}\middle\|\gamma_S\fright)&\coloneq\inf_{\substack{t\in\spa{R}, \\ \tau_{SM}\in\sub, \\ \sigma_M\in\sta}}\left\{\ln t\colon\tau_{SM}\leq t\gamma_S\otimes\sigma_M,\;\tau_M\leq\sigma_M,\;\delta\fleft(\tau_{SM},\rho_{SM}\fright)\leq\varepsilon\right\} \label{eq:smoothed-max-information}\\
	&=\inf_{\substack{S_1, \\ \omega_{S_1SM}\in\sta}}\left\{I_{\max}^\uparrow\fleft(\omega_{S_1SM}\middle\|\tfrac{\1_{S_1}}{\left\lvert S_1\right\rvert}\otimes\gamma_S\fright)-\ln\left\lvert S_1\right\rvert\colon\delta\fleft(\omega_{S_1SM},\op{0}{0}_{S_1}\otimes\rho_{SM}\fright)\leq\varepsilon\right\} \label{eq:alternative}.
\end{align}
We define the \emph{(operator-) smoothed (optimized) (Petz-type) min-GMI} as
\begin{align}
	I_{\min}^{\downarrow,\varepsilon}\fleft(\rho_{SM}\middle\|\gamma_S\fright)	&\coloneq\sup_{\Lambda_{SM}\in\psd}\left\{-\ln\left\lVert\tr_S\fleft[\Lambda_{SM}\gamma_S\fright]\right\rVert_\infty\colon\tr\fleft[\Lambda_{SM}\rho_{SM}\fright]\geq1-\varepsilon,\;0\leq\Lambda_{SM}\leq\1_{SM}\right\} \label{eq:smoothed-min-information} \\
	&=\inf_{\sigma_M\in\sta}D_{\min}^\varepsilon\fleft(\rho_{SM}\middle\|\gamma_S\otimes\sigma_M\fright). \label{eq:smoothed-optimized-min-information}
\end{align}
The alternative expressions of the smoothed max- and min-GMIs in Eqs.~\eqref{eq:alternative} and \eqref{eq:smoothed-optimized-min-information} are established in the companion paper~\cite[Lemmas~9 and 10]{ji2026QuantumConditionalRelative}, and their representations in terms of semidefinite programs are presented in Ref.~\cite[Remark~11]{ji2026QuantumConditionalRelative}.

\begin{proposition}[Properties of $I_{\max}^{\uparrow,\varepsilon}$ and $I_{\min}^{\downarrow,\varepsilon}$; restatement of Proposition~\ref{prop:properties}]
\label{prop:properties-supp}
For a real number $\varepsilon\in[0,1]$ and a state $(\rho_{SM},\gamma_S)$, the smoothed max- and min-GMI, denoted by $\hat{I}^\varepsilon\in\{I_{\max}^{\uparrow,\varepsilon},I_{\min}^{\downarrow,\varepsilon}\}$, satisfy the following properties.
\begin{itemize}
	\item[(i)] Monotonicity under QFGOs: for every QFGO $\ch{N}_{SM\to S'M'}$,
	\begin{align}
		\hat{I}^\varepsilon\fleft(\rho_{SM}\middle\|\gamma_S\fright)&\geq\hat{I}^\varepsilon\fleft(\ch{N}_{SM\to S'M'}\fleft[\rho_{SM}\fright]\middle\|\gamma'_{S'}\fright).
	\end{align}
	\item[(ii)] Robustness in the presence of a battery: for every battery $S_0$,
	\begin{align}
		\hat{I}^\varepsilon\fleft(\op{0}{0}_{S_0}\otimes\rho_{SM}\middle\|\tfrac{\1_{S_0}}{\left\lvert S_0\right\rvert}\otimes\gamma_S\fright)&=\ln\left\lvert S_0\right\rvert+\hat{I}^\varepsilon\fleft(\rho_{SM}\middle\|\gamma_S\fright).
	\end{align}
	\item[(iii)] Asymptotic equipartition property:
	\begin{align}
		\lim_{n\to\infty}\tfrac{1}{n}\hat{I}^\varepsilon\fleft(\rho_{SM}^{\otimes n}\middle\|\gamma_S^{\otimes n}\fright)&=I\fleft(\rho_{SM}\middle\|\gamma_S\fright)\;\;\forall\varepsilon\in(0,1).
	\end{align}
\end{itemize}
\end{proposition}

\begin{proof}
See the companion paper~\cite{ji2026QuantumConditionalRelative}.
\end{proof}

\section{Work costs under QFGO\lowercase{s}}
\label{sec:work}

\begin{definition}[Single-shot work under QFGOs]
\label{def:work}
For an error parameter $\varepsilon\in[0,1]$ and two states $(\rho_{SM},\gamma_S)$ and $(\rho'_{S'M'},\gamma'_{S'})$, the \emph{single-shot work cost under QFGOs} is defined as
\begin{align}
	&W^\varepsilon\fleft(\fleft(\rho_{SM},\gamma_S\vphantom{\rho'_{S'M'},\gamma'_{S'}}\fright)\to\fleft(\rho'_{S'M'},\gamma'_{S'}\fright)\fright) \notag\\
	&\coloneq\beta^{-1}\inf_{\substack{S_0,S_1, \\ \ch{N}_{S_0SM\to S_1S'M'}\in\abb{QFGO}}}\left\{\ln\left\lvert S_0\right\rvert-\ln\left\lvert S_1\right\rvert\colon\delta\fleft(\ch{N}_{S_0SM\to S_1S'M'}\fleft[\op{0}{0}_{S_0}\otimes\rho_{SM}\fright],\op{0}{0}_{S_1}\otimes\rho'_{S'M'}\fright)\leq\varepsilon\right\}, \label{eq:work}
\end{align}
with the Gibbs states of $S_0$ and $S_1$ being $\1_{S_0}/\lvert S_0\rvert$ and $\1_{S_1}/\lvert S_1\rvert$, respectively.  The \emph{single-shot work of formation under QFGOs} is defined as
\begin{align}
	W_\abb{form}^\varepsilon\fleft(\rho_{SM},\gamma_S\fright)&\coloneq W^\varepsilon\fleft(\fleft(1,1\fright)\to\fleft(\rho_{SM},\gamma_S\fright)\fright).
\end{align}
The \emph{single-shot extractable work under QFGOs} is defined as
\begin{align}
	W_\abb{extr}^\varepsilon\fleft(\rho_{SM},\gamma_S\fright)&\coloneq-W^\varepsilon\fleft(\fleft(\rho_{SM},\gamma_S\fright)\to\fleft(1,1\fright)\fright).
\end{align}
\end{definition}

\subsection{Work of formation}
\label{sec:formation}

\begin{theorem}[Work of formation under QFGOs; restatement of Theorem~\ref{thm:formation}]
\label{thm:formation-supp}
For an error parameter $\varepsilon\in[0,1]$ and a target state $(\rho_{SM},\gamma_S)$, the single-shot work of formation under QFGOs and its asymptotic limit are given by
\begin{align}
	W_\abb{form}^\varepsilon\fleft(\rho_{SM},\gamma_S\fright)&=\beta^{-1}I_{\max}^{\uparrow,\varepsilon}\fleft(\rho_{SM}\middle\|\gamma_S\fright), \\
	\lim_{n\to\infty}\frac{1}{n}W_\abb{form}^\varepsilon\fleft(\rho_{SM}^{\otimes n},\gamma_S^{\otimes n}\fright)&=\beta^{-1}I\fleft(\rho_{SM}\middle\|\gamma_S\fright)\;\;\forall\varepsilon\in(0,1).
\end{align}
\end{theorem}

\begin{proof}[Proof of achievability]
The proof follows a similar idea to that of Ref.~\cite[Theorem~1]{gour2024InevitabilityKnowingLess}.  Let $S_1$ be a battery and $(\omega_{S_1SM},\1_{S_1}/\lvert S_1\rvert\otimes\gamma_S)$ a state such that
\begin{align}
	\delta\fleft(\omega_{S_1SM},\op{0}{0}_{S_1}\otimes\rho_{SM}\fright)&\leq\varepsilon. \label{pf:formation-1}
\end{align}
Define the following channel (which is equivalent to the feedback-control scheme constructed in Appendix~C):
\begin{align}
	&\ch{N}_{S_0\to S_1SM}\fleft[\cdot\fright] \notag\\
	&\coloneq\tr_{S_0}\fleft[\op{0}{0}_{S_0}\left(\cdot\right)_{S_0}\fright]\otimes\omega_{S_1SM}+\tr_{S_0}\fleft[\left(\1_{S_0}-\op{0}{0}_{S_0}\right)\left(\cdot\right)_{S_0}\fright]\otimes\frac{\left\lvert S_0\right\rvert\1_{S_1}\otimes\gamma_S\otimes\omega_M-\left\lvert S_1\right\rvert\omega_{S_1SM}}{\left(\left\lvert S_0\right\rvert-1\right)\left\lvert S_1\right\rvert}, \label{pf:formation-2}
\end{align}
where $S_0$ is a battery such that
\begin{align}
	\left\lvert S_0\right\rvert&=\left\lceil e^{I_{\max}^\uparrow\fleft(\omega_{S_1SM}\middle\|\tfrac{\1_{S_1}}{\left\lvert S_1\right\rvert}\otimes\gamma_S\fright)}\right\rceil \label{pf:formation-3}\\
	&\geq\left\lvert S_1\right\rvert\left\lVert\left(\gamma_S^{-\frac{1}{2}}\otimes\omega_M^{-\frac{1}{2}}\right)\omega_{S_1SM}\left(\gamma_S^{-\frac{1}{2}}\otimes\omega_M^{-\frac{1}{2}}\right)\right\rVert_\infty, \label{pf:formation-4}
\end{align}
with $\lceil\cdot\rceil$ the ceiling function.  Here Eq.~\eqref{pf:formation-4} follows from the definition of the max-GMI [Eq.~\eqref{eq:max-information}].  This implies that $\lvert S_0\rvert\1_{S_1}\otimes\gamma_S\otimes\omega_M-\lvert S_1\rvert\omega_{S_1SM}\geq0$, confirming that $\ch{N}_{S_0\to S_1SM}$ is a channel.  It follows from the definition of thermalization [Eq.~\eqref{eq:thermalization}] that
\begin{align}
	\ch{N}_{S_0\to S_1SM}\circ\ch{R}_{S_0\to S_0}^{\tfrac{\1}{\left\lvert S_0\right\rvert}}&=\ch{N}_{S_0\to S_1SM}\fleft[\frac{\1_{S_0}}{\left\lvert S_0\right\rvert}\fright]\otimes\tr_{S_0} \\
	&=\left(\frac{\omega_{S_1SM}}{\left\lvert S_0\right\rvert}+\left(1-\frac{1}{\left\lvert S_0\right\rvert}\right)\frac{\left\lvert S_0\right\rvert\1_{S_1}\otimes\gamma_S\otimes\omega_M-\left\lvert S_1\right\rvert\omega_{S_1SM}}{\left(\left\lvert S_0\right\rvert-1\right)\left\lvert S_1\right\rvert}\right)\otimes\tr_{S_0} \label{pf:formation-5}\\
	&=\frac{\1_{S_1}}{\left\lvert S_1\right\rvert}\otimes\gamma_S\otimes\omega_M\otimes\tr_{S_0} \\
	&=\frac{\1_{S_1}}{\left\lvert S_1\right\rvert}\otimes\gamma_S\otimes\left(\tr_{S_1S}\circ\ch{N}_{S_0\to S_1SM}\right) \label{pf:formation-6}\\
	&=\ch{R}_{S_1S\to S_1S}^{\tfrac{\1}{\left\lvert S_1\right\rvert}\otimes\gamma}\circ\ch{N}_{S_0\to S_1SM}. \label{pf:formation-7}
\end{align}
Here Eqs.~\eqref{pf:formation-5} and \eqref{pf:formation-6} follow from Eq.~\eqref{pf:formation-2}; Eq.~\eqref{pf:formation-7} follows from the definition of thermalization [Eq.~\eqref{eq:thermalization}].  It follows from Proposition~\ref{prop:operation-supp} and Eq.~\eqref{pf:formation-7} that $\ch{N}_{S_0\to S_1SM}$ is a QFGO.  It follows from Eq.~\eqref{pf:formation-2} that
\begin{align}
	\delta\fleft(\ch{N}_{S_0\to S_1SM}\fleft[\op{0}{0}_{S_0}\fright],\op{0}{0}_{S_1}\otimes\rho_{SM}\fright)&=\delta\fleft(\omega_{S_1SM},\op{0}{0}_{S_1}\otimes\rho_{SM}\fright) \\
	&\leq\varepsilon. \label{pf:formation-8}
\end{align}
Here Eq.~\eqref{pf:formation-8} follows from Eq.~\eqref{pf:formation-1}.  It follows from the definition of the single-shot work of formation under QFGOs (Definition~\ref{def:work}), the fact that $\ch{N}_{S_0\to S_1SM}$ is a QFGO, and Eq.~\eqref{pf:formation-8} that
\begin{align}
	W_\abb{form}^\varepsilon\fleft(\rho_{SM},\gamma_S\fright)&\leq\beta^{-1}\left(\ln\left\lvert S_0\right\rvert-\ln\left\lvert S_1\right\rvert\right) \\
	&=\beta^{-1}\ln\frac{\left\lceil e^{I_{\max}^\uparrow\fleft(\omega_{S_1SM}\middle\|\tfrac{\1_{S_1}}{\left\lvert S_1\right\rvert}\otimes\gamma_S\fright)}\right\rceil}{\left\lvert S_1\right\rvert}. \label{pf:formation-9}
\end{align}
Here Eq.~\eqref{pf:formation-9} follows from Eq.~\eqref{pf:formation-3}.  Taking the infimum on the right-hand side of Eq.~\eqref{pf:formation-9} over every battery $S_1$ and state $\omega_{S_1SM}$ satisfying Eq.~\eqref{pf:formation-1}, we have that
\begin{align}
	W_\abb{form}^\varepsilon\fleft(\rho_{SM},\gamma_S\fright)&\leq\beta^{-1}\inf_{\substack{S_1, \\ \omega_{S_1SM}\in\sta}}\left\{\ln\frac{\left\lceil e^{I_{\max}^\uparrow\fleft(\omega_{S_1SM}\middle\|\tfrac{\1_{S_1}}{\left\lvert S_1\right\rvert}\otimes\gamma_S\fright)}\right\rceil}{\left\lvert S_1\right\rvert}\colon\delta\fleft(\omega_{S_1SM},\op{0}{0}_{S_1}\otimes\rho_{SM}\fright)\leq\varepsilon\right\} \\
	&\leq\beta^{-1}\inf_{\substack{S_2,S_3, \\ \omega_{S_3SM}\in\sta}}\left\{\ln\frac{\left\lceil e^{I_{\max}^\uparrow\fleft(\op{0}{0}_{S_2}\otimes\omega_{S_3SM}\middle\|\tfrac{\1_{S_2}}{\left\lvert S_2\right\rvert}\otimes\tfrac{\1_{S_3}}{\left\lvert S_3\right\rvert}\otimes\gamma_S\fright)}\right\rceil}{\left\lvert S_2\right\rvert\left\lvert S_3\right\rvert}\colon\right. \notag\\
	&\qquad\left.\vphantom{\ln\frac{\left\lceil\left\lvert S_2\right\rvert\left\lvert S_3\right\rvert e^{I_{\max}^\uparrow\fleft(\op{0}{0}_{S_2}\otimes\omega_{S_3SM}\middle\|\1_{S_2}\otimes\1_{S_3}\otimes\gamma_S\fright)}\right\rceil}{\left\lvert S_2\right\rvert\left\lvert S_3\right\rvert}\colon}\delta\fleft(\op{0}{0}_{S_2}\otimes\omega_{S_3SM},\op{0}{0}_{S_2}\otimes\op{0}{0}_{S_3}\otimes\rho_{SM}\fright)\leq\varepsilon\right\} \\
	&=\beta^{-1}\inf_{\substack{S_2,S_3, \\ \omega_{S_3SM}\in\sta}}\left\{\ln\frac{\left\lceil\left\lvert S_2\right\rvert e^{I_{\max}^\uparrow\fleft(\omega_{S_3SM}\middle\|\tfrac{\1_{S_3}}{\left\lvert S_3\right\rvert}\otimes\gamma_S\fright)}\right\rceil}{\left\lvert S_2\right\rvert\left\lvert S_3\right\rvert}\colon\delta\fleft(\omega_{S_3SM},\op{0}{0}_{S_3}\otimes\rho_{SM}\fright)\leq\varepsilon\right\} \label{pf:formation-10}\\
	&=\beta^{-1}\inf_{\substack{S_3, \\ \omega_{S_3SM}\in\sta}}\left\{I_{\max}^\uparrow\fleft(\omega_{S_3SM}\middle\|\tfrac{\1_{S_3}}{\left\lvert S_3\right\rvert}\otimes\gamma_S\fright)-\ln\left\lvert S_3\right\rvert\colon\delta\fleft(\omega_{S_3SM},\op{0}{0}_{S_3}\otimes\rho_{SM}\fright)\leq\varepsilon\right\} \label{pf:formation-16}\\
	&=\beta^{-1}I_{\max}^{\uparrow,\varepsilon}\fleft(\rho_{SM}\middle\|\gamma_S\fright). \label{pf:formation-11}
\end{align}
Here Eq.~\eqref{pf:formation-10} follows from the definitions of the max-GMI [Eq.~\eqref{eq:max-information}] and the generalized trace distance [Eq.~\eqref{eq:distance}]; Eq.~\eqref{pf:formation-16} follows from taking $\lvert S_2\rvert\to\infty$, as $\lceil\lvert S_2\rvert x\rceil/\lvert S_2\rvert\to x$ for every $x>0$; Eq.~\eqref{pf:formation-11} follows from the alternative expression for the smoothed max-GMI [Eq.~\eqref{eq:alternative}].  This shows the single-shot achievability.  The asymptotic achievability follows from the asymptotic equipartition property for the smoothed max-GMI [Proposition~\ref{prop:properties-supp}(iii)].
\end{proof}

\begin{proof}[Proof of optimality]
Let $\ch{N}_{S_0\to S_1SM}$ be a QFGO such that
\begin{align}
	\delta\fleft(\ch{N}_{S_0\to S_1SM}\fleft[\op{0}{0}_{S_0}\fright],\op{0}{0}_{S_1}\otimes\rho_{SM}\fright)&\leq\varepsilon. \label{pf:formation-12}
\end{align}
It follows from the definition of the max-GMI [Eq.~\eqref{eq:max-information}] that 
\begin{align}
	\ln\left\lvert S_0\right\rvert-\ln\left\lvert S_1\right\rvert&=I_{\max}^\uparrow\fleft(\op{0}{0}_{S_0}\middle\|\tfrac{\1_{S_0}}{\left\lvert S_0\right\rvert}\fright)-\ln\left\lvert S_1\right\rvert \label{pf:formation-13}\\
	&\geq I_{\max}^\uparrow\fleft(\ch{N}_{S_0\to S_1SM}\fleft[\op{0}{0}_{S_0}\fright]\middle\|\tfrac{\1_{S_1}}{\left\lvert S_1\right\rvert}\otimes\gamma_S\fright)-\ln\left\lvert S_1\right\rvert \label{pf:formation-14}\\
	&\geq I_{\max}^{\uparrow,\varepsilon}\fleft(\rho_{SM}\middle\|\gamma_S\fright). \label{pf:formation-15}
\end{align}
Here Eq.~\eqref{pf:formation-14} follows from the monotonicity of the max-GMI under QFGOs [Proposition~\ref{prop:properties-supp}(i)]; Eq.~\eqref{pf:formation-15} follows from Eq.~\eqref{pf:formation-12} and the alternative expression for the smoothed max-GMI [Eq.~\eqref{eq:alternative}].  Taking the infimum on the left-hand side of Eq.~\eqref{pf:formation-13} over every battery $S_0$, battery $S_1$, and QFGO $\ch{N}_{S_0\to S_1SM}$ satisfying Eq.~\eqref{pf:formation-12}, the single-shot optimality follows from the definition of the single-shot work of formation under QFGOs (Definition~\ref{def:work}).  The asymptotic optimality follows from the asymptotic equipartition property for the smoothed max-GMI [Proposition~\ref{prop:properties-supp}(iii)].
\end{proof}

\subsection{Extractable work}
\label{sec:extraction}

\begin{theorem}[Extractable work under QFGOs; restatement of Theorem~\ref{thm:extraction}]
\label{thm:extraction-supp}
For an error parameter $\varepsilon\in[0,1]$ and an initial state $(\rho_{SM},\gamma_S)$, the single-shot extractable work under QFGOs and its asymptotic limit are given by
\begin{align}
	W_\abb{extr}^\varepsilon\fleft(\rho_{SM},\gamma_S\fright)&=\beta^{-1}I_{\min}^{\downarrow,\varepsilon}\fleft(\rho_{SM}\middle\|\gamma_S\fright), \\
	\lim_{n\to\infty}\frac{1}{n}W_\abb{extr}^\varepsilon\fleft(\rho_{SM}^{\otimes n},\gamma_S^{\otimes n}\fright)&=\beta^{-1}I\fleft(\rho_{SM}\middle\|\gamma_S\fright)\;\;\forall\varepsilon\in(0,1).
\end{align}
\end{theorem}

\begin{proof}[Proof of achievability]
Let $S_0$ be a battery and $\Lambda_{SM}$ an operator such that
\begin{align}
	\tr\fleft[\Lambda_{SM}\rho_{SM}\fright]&\geq1-\varepsilon, \label{pf:extraction-1}\\
	0\leq\Lambda_{SM}&\leq\1_{SM}. \label{pf:extraction-2}
\end{align}
Define the following channel (which is equivalent to the feedback-control scheme constructed in Appendix~D):
\begin{align}
	\ch{N}_{S_0SM\to S_1}\fleft[\cdot\fright]&\coloneq\tr_{S_0SM}\fleft[\Omega_{S_0SM}\left(\cdot\right)_{S_0SM}\fright]\otimes\op{0}{0}_{S_1}+\tr_{S_0SM}\fleft[\left(\1_{S_0SM}-\Omega_{S_0SM}\right)\left(\cdot\right)_{S_0SM}\fright]\otimes\frac{\1_{S_1}-\op{0}{0}_{S_1}}{\left\lvert S_1\right\rvert-1}, \label{pf:extraction-3}
\end{align}
where $\Omega_{S_0SM}$ is an operator and $S_1$ a battery such that
\begin{align}
	\Omega_{S_0SM}&\coloneq\op{0}{0}_{S_0}\otimes\Lambda_{SM}+\op{1}{1}_{S_0}\otimes\1_S\otimes\frac{\left\lvert S_0\right\rvert\1_M-\left\lvert S_1\right\rvert\tr_S\fleft[\Lambda_{SM}\gamma_S\fright]}{\left\lvert S_1\right\rvert}, \label{pf:extraction-4}\\
	\left\lvert S_1\right\rvert&=\left\lfloor\frac{\left\lvert S_0\right\rvert}{\left\lVert\tr_S\fleft[\Lambda_{SM}\gamma_S\fright]\right\rVert_\infty}\right\rfloor, \label{pf:extraction-5}
\end{align}
with $\lfloor\cdot\rfloor$ the floor function.  
It follows from Eqs.~\eqref{pf:extraction-2}, \eqref{pf:extraction-4}, and \eqref{pf:extraction-5} that
\begin{align}
	0&\leq\left\lvert S_0\right\rvert\1_M-\left\lvert S_1\right\rvert\tr_S\fleft[\Lambda_{SM}\gamma_S\fright] \\
	&\leq\left\lvert S_1\right\rvert\1_M.
\end{align}
This implies that $0\leq\Omega_{S_0SM}\leq\1_{S_0SM}$, confirming that $\ch{N}_{S_0SM\to S_1}$ is a channel.  It follows from the definition of thermalization [Eq.~\eqref{eq:thermalization}] that
\begin{align}
	&\left(\ch{N}_{S_0SM\to S_1}\circ\ch{R}_{S_0S\to S_0S}^{\tfrac{\1}{\left\lvert S_0\right\rvert}\otimes\gamma}\right)\fleft[\cdot\fright] \notag\\
	&=\tr_{S_0SM}\fleft[\Omega_{S_0SM}\left(\frac{\1_{S_0}}{\left\lvert S_0\right\rvert}\otimes\gamma_S\otimes\tr_{S_0S}\fleft[\cdot\fright]\right)\fright]\otimes\op{0}{0}_{S_1} \notag\\
	&\qquad+\left(\tr_{S_0SM}\fleft[\cdot\fright]-\tr_{S_0SM}\fleft[\Omega_{S_0SM}\left(\frac{\1_{S_0}}{\left\lvert S_0\right\rvert}\otimes\gamma_S\otimes\tr_{S_0S}\fleft[\cdot\fright]\right)\fright]\right)\otimes\frac{\1_{S_1}-\op{0}{0}_{S_1}}{\left\lvert S_1\right\rvert-1} \label{pf:extraction-6}\\
	&=\frac{1}{\left\lvert S_0\right\rvert}\tr_{S_0SM}\fleft[\tr_S\fleft[\Omega_{SM}\gamma_S\fright]\left(\cdot\right)_{S_0SM}\fright]\otimes\op{0}{0}_{S_1} \notag\\
	&\qquad+\left(\tr_{S_0SM}\fleft[\cdot\fright]-\frac{1}{\left\lvert S_0\right\rvert}\tr_{S_0SM}\fleft[\tr_S\fleft[\Omega_{SM}\gamma_S\fright]\left(\cdot\right)_{S_0SM}\fright]\right)\otimes\frac{\1_{S_1}-\op{0}{0}_{S_1}}{\left\lvert S_1\right\rvert-1} \\
	&=\frac{1}{\left\lvert S_0\right\rvert}\tr_{S_0SM}\fleft[\left(\tr_S\fleft[\Lambda_{SM}\gamma_S\fright]+\frac{\left\lvert S_0\right\rvert\1_M-\left\lvert S_1\right\rvert\tr_S\fleft[\Lambda_{SM}\gamma_S\fright]}{\left\lvert S_1\right\rvert}\right)\left(\cdot\right)_{S_0SM}\fright]\otimes\op{0}{0}_{S_1} \notag\\
	&\qquad+\left(\tr_{S_0SM}\fleft[\cdot\fright]-\frac{1}{\left\lvert S_0\right\rvert}\tr_{S_0SM}\fleft[\left(\tr_S\fleft[\Lambda_{SM}\gamma_S\fright]+\frac{\left\lvert S_0\right\rvert\1_M-\left\lvert S_1\right\rvert\tr_S\fleft[\Lambda_{SM}\gamma_S\fright]}{\left\lvert S_1\right\rvert}\right)\left(\cdot\right)_{S_0SM}\fright]\right)\otimes\frac{\1_{S_1}-\op{0}{0}_{S_1}}{\left\lvert S_1\right\rvert-1} \label{pf:extraction-7}\\
	&=\tr_{S_0SM}\fleft[\cdot\fright]\otimes\frac{\1_{S_1}}{\left\lvert S_1\right\rvert} \\
	&=\left(\ch{R}_{S_1\to S_1}^{\tfrac{\1}{\left\lvert S_1\right\rvert}}\circ\ch{N}_{S_0SM\to S_1}\right)\fleft[\cdot\fright]. \label{pf:extraction-8}
\end{align}
Here Eq.~\eqref{pf:extraction-6} follows from Eq.~\eqref{pf:extraction-3}; Eq.~\eqref{pf:extraction-7} follows from Eq.~\eqref{pf:extraction-4}; Eq.~\eqref{pf:extraction-8} follows from the definition of thermalization [Eq.~\eqref{eq:thermalization}].  It follows from Proposition~\ref{prop:operation-supp} and Eq.~\eqref{pf:extraction-8} that $\ch{N}_{S_0SM\to S_1}$ is a QFGO.  It follows from Eq.~\eqref{pf:extraction-3} that
\begin{align}
	&\delta\fleft(\ch{N}_{S_0SM\to S_1}\fleft[\op{0}{0}_{S_0}\otimes\rho_{SM}\fright],\op{0}{0}_{S_1}\fright) \notag\\
	&=\delta\fleft(\tr\fleft[\Omega_{S_0SM}\left(\op{0}{0}_{S_0}\otimes\rho_{SM}\right)\fright]\op{0}{0}_{S_1}+\left(1-\tr\fleft[\Omega_{S_0SM}\left(\op{0}{0}_{S_0}\otimes\rho_{SM}\right)\fright]\right)\frac{\1_{S_1}-\op{0}{0}_{S_1}}{\left\lvert S_1\right\rvert-1},\op{0}{0}_{S_1}\fright) \\
	&=1-\tr\fleft[\Omega_{S_0SM}\left(\op{0}{0}_{S_0}\otimes\rho_{SM}\right)\fright] \label{pf:extraction-9}\\
	&=1-\tr\fleft[\Lambda_{SM}\rho_{SM}\fright] \label{pf:extraction-10}\\
	&\leq\varepsilon. \label{pf:extraction-11}
\end{align}
Here Eq.~\eqref{pf:extraction-9} follows from the definition of the trace distance [Eq.~\eqref{eq:distance}]; Eq.~\eqref{pf:extraction-10} follows from Eq.~\eqref{pf:extraction-4}; Eq.~\eqref{pf:extraction-11} follows from Eq.~\eqref{pf:extraction-1}.  It follows from the definition of the single-shot extractable work under QFGOs (Definition~\ref{def:work}), the fact that $\ch{N}_{S_0SM\to S_1}$ is a QFGO, and Eq.~\eqref{pf:extraction-11} that
\begin{align}
	W_\abb{extr}^\varepsilon\fleft(\rho_{SM},\gamma_S\fright)&\geq\beta^{-1}\left(\ln\left\lvert S_1\right\rvert-\ln\left\lvert S_0\right\rvert\right) \\
	&=\beta^{-1}\ln\left(\frac{1}{\left\lvert S_0\right\rvert}\left\lfloor\frac{\left\lvert S_0\right\rvert}{\left\lVert\tr_S\fleft[\Lambda_{SM}\gamma_S\fright]\right\rVert_\infty}\right\rfloor\right). \label{pf:extraction-12}
\end{align}
Here Eq.~\eqref{pf:extraction-12} follows from Eq.~\eqref{pf:extraction-5}.  Taking the supremum on the right-hand side of Eq.~\eqref{pf:extraction-12} over every battery $S_0$ and operator $\Lambda_{SM}$ satisfying Eqs.~\eqref{pf:extraction-1} and \eqref{pf:extraction-2}, we have that
\begin{align}
	W_\abb{extr}^\varepsilon\fleft(\rho_{SM},\gamma_S\fright)&\geq\beta^{-1}\sup_{\substack{S_0, \\ \Lambda_{SM}\in\psd}}\left\{\ln\left(\frac{1}{\left\lvert S_0\right\rvert}\left\lfloor\frac{\left\lvert S_0\right\rvert}{\left\lVert\tr_S\fleft[\Lambda_{SM}\gamma_S\fright]\right\rVert_\infty}\right\rfloor\right)\colon\tr\fleft[\Lambda_{SM}\rho_{SM}\fright]\geq1-\varepsilon,\;0\leq\Lambda_{SM}\leq\1_{SM}\right\} \\
	&=\beta^{-1}\sup_{\Lambda_{SM}\in\psd}\left\{-\ln\left\lVert\tr_S\fleft[\Lambda_{SM}\gamma_S\fright]\right\rVert_\infty\colon\tr\fleft[\Lambda_{SM}\rho_{SM}\fright]\geq1-\varepsilon,\;0\leq\Lambda_{SM}\leq\1_{SM}\right\} \\
	&=\beta^{-1}I_{\min}^{\downarrow,\varepsilon}\fleft(\rho_{SM}\middle\|\gamma_S\fright). \label{pf:extraction-13}
\end{align}
Here Eq.~\eqref{pf:extraction-13} follows from the definition of the smoothed min-GMI [Eq.~\eqref{eq:smoothed-min-information}].  This shows the single-shot achievability.  The asymptotic achievability follows from the asymptotic equipartition property for the smoothed min-GMI [Proposition~\ref{prop:properties-supp}(iii)].
\end{proof}

\begin{proof}[Proof of optimality]
Let $\ch{N}_{S_0SM\to S_1}$ be a QFGO such that
\begin{align}
	\delta\fleft(\ch{N}_{S_0SM\to S_1}\fleft[\op{0}{0}_{S_0}\otimes\rho_{SM}\fright],\op{0}{0}_{S_1}\fright)&\leq\varepsilon. \label{pf:extraction-14}
\end{align}
It follows from the definition of the trace distance [Eq.~\eqref{eq:distance}] that
\begin{align}
	&\tr\fleft[\op{0}{0}_{S_1}\ch{N}_{S_0SM\to S_1}\fleft[\op{0}{0}_{S_0}\otimes\rho_{SM}\fright]\fright] \notag\\
	&=1-\tr\fleft[\left(\1_{S_1}-\op{0}{0}_{S_1}\right)\left(\ch{N}_{S_0SM\to S_1}\fleft[\op{0}{0}_{S_0}\otimes\rho_{SM}\fright]-\op{0}{0}_{S_1}\right)\fright] \\
	&\geq1-\sup_{\Lambda_{S_1}\in\psd}\left\{\tr\fleft[\Lambda_{S_1}\left(\ch{N}_{S_0SM\to S_1}\fleft[\op{0}{0}_{S_0}\otimes\rho_{SM}\fright]-\op{0}{0}_{S_1}\right)\fright]\colon0\leq\Lambda_{S_1}\leq\1_{S_1}\right\} \\
	&=1-\delta\fleft(\ch{N}_{S_0SM\to S_1}\fleft[\op{0}{0}_{S_0}\otimes\rho_{SM}\fright],\op{0}{0}_{S_1}\fright) \label{pf:extraction-15}\\
	&\geq1-\varepsilon. \label{pf:extraction-16}
\end{align}
Here Eq.~\eqref{pf:extraction-15} follows from the variational expression of the trace distance [Eq.~\eqref{eq:dual}]; Eq.~\eqref{pf:extraction-16} follows from Eq.~\eqref{pf:extraction-14}.  It follows from the definition of the smoothed min-GMI [Eq.~\eqref{eq:smoothed-min-information}] that
\begin{align}
	\ln\left\lvert S_1\right\rvert-\ln\left\lvert S_0\right\rvert&=-\ln\tr\fleft[\op{0}{0}_{S_1}\frac{\1_{S_1}}{\left\lvert S_1\right\rvert}\fright]-\ln\left\lvert S_0\right\rvert \label{pf:extraction-17}\\
	&\leq I_{\min}^{\downarrow,\varepsilon}\fleft(\ch{N}_{S_0SM\to S_1}\fleft[\op{0}{0}_{S_0}\otimes\rho_{SM}\fright]\middle\|\tfrac{\1_{S_1}}{\left\lvert S_1\right\rvert}\fright)-\ln\left\lvert S_0\right\rvert \\
	&\leq I_{\min}^{\downarrow,\varepsilon}\fleft(\op{0}{0}_{S_0}\otimes\rho_{SM}\middle\|\tfrac{\1_{S_0}}{\left\lvert S_0\right\rvert}\otimes\gamma_S\fright)-\ln\left\lvert S_0\right\rvert \label{pf:extraction-18}\\
	&=I_{\min}^{\downarrow,\varepsilon}\fleft(\rho_{SM}\middle\|\gamma_S\fright). \label{pf:extraction-19}
\end{align}
Here Eq.~\eqref{pf:extraction-18} follows from the monotonicity of the smoothed min-GMI under QFGOs [Proposition~\ref{prop:properties-supp}(i)]; Eq.~\eqref{pf:extraction-19} follows from the robustness of the smoothed min-GMI in the presence of a battery [Proposition~\ref{prop:properties-supp}(ii)].  Taking the supremum on the left-hand side of Eq.~\eqref{pf:extraction-17} over every battery $S_0$, battery $S_1$, and QFGO $\ch{N}_{S_0SM\to S_1}$ satisfying Eq.~\eqref{pf:extraction-14}, the single-shot optimality follows from the definition of the single-shot extractable work under QFGOs (Definition~\ref{def:work}).  The asymptotic optimality follows from the asymptotic equipartition property for the smoothed min-GMI [Proposition~\ref{prop:properties-supp}(iii)].
\end{proof}

\subsection{General conversion}
\label{sec:conversion}

\begin{lemma}[Triangle inequality]
\label{lem:triangle}
For two error parameters $\varepsilon\in[0,1]$ and $\eta\in[0,\varepsilon]$ and three states $(\rho_{SM},\gamma_S)$, $(\rho'_{S'M'},\gamma'_{S'})$, and $(\rho''_{S''M''},\gamma''_{S''})$,
\begin{align}
	W^\varepsilon\fleft(\fleft(\rho_{SM},\gamma_S\vphantom{\rho'_{S'M'},\gamma'_{S'}}\fright)\to\fleft(\rho'_{S'M'},\gamma'_{S'}\fright)\fright)&\leq W^\eta\fleft(\fleft(\rho_{SM},\gamma_S\vphantom{\rho''_{S''M''},\gamma''_{S''}}\fright)\to\fleft(\rho''_{S''M''},\gamma''_{S''}\fright)\fright)+W^{\varepsilon-\eta}\fleft(\fleft(\rho''_{S''M''},\gamma''_{S''}\fright)\to\fleft(\rho'_{S'M'},\gamma'_{S'}\fright)\fright).
\end{align}
\end{lemma}

\begin{proof}
Let $S_0$, $S_1$, $S'_0$, and $S'_1$ be batteries and $\ch{N}_{S_0SM\to S_1S''M''}$ and $\ch{N}'_{S'_0S''M''\to S'_1S'M'}$ QFGOs such that
\begin{align}
	\delta\fleft(\ch{N}_{S_0SM\to S_1S''M''}\fleft[\op{0}{0}_{S_0}\otimes\rho_{SM}\fright],\op{0}{0}_{S_1}\otimes\rho''_{S''M''}\fright)&\leq\eta, \label{pf:triangle-1}\\
	\delta\fleft(\ch{N}'_{S'_0S''M''\to S'_1S'M'}\fleft[\op{0}{0}_{S'_0}\otimes\rho''_{S''M''}\fright],\op{0}{0}_{S'_1}\otimes\rho'_{S'M'}\fright)&\leq\varepsilon-\eta. \label{pf:triangle-2}
\end{align}
Define the following QFGO:
\begin{align}
	\ch{N}''_{S_0S'_0SM\to S_1S'_1S'M'}&\coloneq\ch{N}'_{S'_0S''M''\to S'_1S'M'}\circ\ch{N}_{S_0SM\to S_1S''M''}. \label{pf:triangle-3}
\end{align}
It follows from the triangle inequality for the trace distance that
\begin{align}
	&\delta\fleft(\ch{N}''_{S_0S'_0SM\to S_1S'_1S'M'}\fleft[\op{0}{0}_{S_0}\otimes\op{0}{0}_{S'_0}\otimes\rho_{SM}\fright],\op{0}{0}_{S_1}\otimes\op{0}{0}_{S'_1}\otimes\rho'_{S'M'}\fright) \notag\\
	&\leq\delta\fleft(\ch{N}''_{S_0S'_0SM\to S_1S'_1S'M'}\fleft[\op{0}{0}_{S_0}\otimes\op{0}{0}_{S'_0}\otimes\rho_{SM}\fright],\op{0}{0}_{S_1}\otimes\ch{N}'_{S'_0S''M''\to S'_1S'M'}\fleft[\op{0}{0}_{S'_0}\otimes\rho''_{S''M''}\fright]\fright) \notag\\
	&\qquad+\delta\fleft(\op{0}{0}_{S_1}\otimes\ch{N}'_{S'_0S''M''\to S'_1S'M'}\fleft[\op{0}{0}_{S'_0}\otimes\rho''_{S''M''}\fright],\op{0}{0}_{S_1}\otimes\op{0}{0}_{S'_1}\otimes\rho'_{S'M'}\fright) \\
	&=\delta\fleft(\left(\ch{N}'_{S'_0S''M''\to S'_1S'M'}\circ\ch{N}_{S_0SM\to S_1S''M''}\right)\fleft[\op{0}{0}_{S_0}\otimes\op{0}{0}_{S'_0}\otimes\rho_{SM}\fright],\right. \notag\\
	&\qquad\left.\op{0}{0}_{S_1}\otimes\ch{N}'_{S'_0S''M''\to S'_1S'M'}\fleft[\op{0}{0}_{S'_0}\otimes\rho''_{S''M''}\fright]\fright) \notag\\
	&\qquad+\delta\fleft(\op{0}{0}_{S_1}\otimes\ch{N}'_{S'_0S''M''\to S'_1S'M'}\fleft[\op{0}{0}_{S'_0}\otimes\rho''_{S''M''}\fright],\op{0}{0}_{S_1}\otimes\op{0}{0}_{S'_1}\otimes\rho'_{S'M'}\fright) \label{pf:triangle-4}\\
	&\leq\delta\fleft(\ch{N}_{S_0SM\to S_1S''M''}\fleft[\op{0}{0}_{S_0}\otimes\op{0}{0}_{S'_0}\otimes\rho_{SM}\fright],\op{0}{0}_{S_1}\otimes\op{0}{0}_{S'_0}\otimes\rho''_{S''M''}\fright) \notag\\
	&\qquad+\delta\fleft(\op{0}{0}_{S_1}\otimes\ch{N}'_{S'_0S''M''\to S'_1S'M'}\fleft[\op{0}{0}_{S'_0}\otimes\rho''_{S''M''}\fright],\op{0}{0}_{S_1}\otimes\op{0}{0}_{S'_1}\otimes\rho'_{S'M'}\fright) \label{pf:triangle-5}\\
	&=\delta\fleft(\ch{N}_{S_0SM\to S_1S''M''}\fleft[\op{0}{0}_{S_0}\otimes\rho_{SM}\fright],\op{0}{0}_{S_1}\otimes\rho''_{S''M''}\fright) \notag\\
	&\qquad+\delta\fleft(\ch{N}'_{S'_0S''M''\to S'_1S'M'}\fleft[\op{0}{0}_{S'_0}\otimes\rho''_{S''M''}\fright],\op{0}{0}_{S'_1}\otimes\rho'_{S'M'}\fright) \label{pf:triangle-6}\\
	&\leq\eta+\varepsilon-\eta \label{pf:triangle-7}\\
	&=\varepsilon.
\end{align}
Here Eq.~\eqref{pf:triangle-4} follows from Eq.~\eqref{pf:triangle-3}; Eq.~\eqref{pf:triangle-5} follows from the data-processing inequality for the trace distance; Eq.~\eqref{pf:triangle-6} follows from the definition of the trace distance [Eq.~\eqref{eq:distance}]; Eq.~\eqref{pf:triangle-7} follows from Eqs.~\eqref{pf:triangle-1} and \eqref{pf:triangle-2}.  It follows from the definition of the single-shot work cost under QFGOs (Definition~\ref{def:work}) that
\begin{align}
	W^\varepsilon\fleft(\fleft(\rho_{SM},\gamma_S\vphantom{\rho'_{S'M'},\gamma'_{S'}}\fright)\to\fleft(\rho'_{S'M'},\gamma'_{S'}\fright)\fright)&\leq\beta^{-1}\left(\ln\left\lvert S_0\right\rvert+\ln\left\lvert S'_0\right\rvert-\ln\left\lvert S_1\right\rvert-\ln\left\lvert S'_1\right\rvert\right). \label{pf:triangle-8}
\end{align}
Taking the infimum on the right-hand side of Eq.~\eqref{pf:triangle-8} over every battery $S_0$, battery $S_1$, battery $S'_0$, battery $S'_1$, QFGO $\ch{N}_{S_0SM\to S_1S''M''}$, and QFGO $\ch{N}'_{S'_0S''M''\to S'_1S'M'}$ satisfying Eqs.~\eqref{pf:triangle-1} and \eqref{pf:triangle-2}, the desired statement follows from the definition of the single-shot work cost under QFGOs (Definition~\ref{def:work}).
\end{proof}

\begin{proposition}[Bounds on the single-shot work cost under QFGOs; restatement of Proposition~\ref{prop:bounds}]
\label{prop:bounds-supp}
For an error parameter $\varepsilon\in[0,1)$, an initial state $(\rho_{SM},\gamma_S)$, and a target state $(\rho'_{S'M'},\gamma'_{S'})$, the single-shot work cost under QFGOs is bounded as
\begin{align}
	\beta^{-1}\max\left\{A,B\right\}&\leq W^\varepsilon\fleft(\fleft(\rho_{SM},\gamma_S\vphantom{\rho'_{S'M'},\gamma'_{S'}}\fright)\to\fleft(\rho'_{S'M'},\gamma'_{S'}\fright)\fright)\leq\beta^{-1}\inf_{\eta\in[0,\varepsilon]}\left(I_{\max}^{\uparrow,\varepsilon-\eta}\fleft(\rho'_{S'M'}\middle\|\gamma'_{S'}\fright)-I_{\min}^{\downarrow,\eta}\fleft(\rho_{SM}\middle\|\gamma_S\fright)\right),
\end{align}
where
\begin{align}
	A&\equiv\sup_{\alpha\in(\frac{1}{2},1)}\left(\sw{I}_\alpha^\downarrow\fleft(\rho'_{S'M'}\middle\|\gamma'_{S'}\fright)-\sw{I}_\frac{\alpha}{2\alpha-1}^\downarrow\fleft(\rho_{SM}\middle\|\gamma_S\fright)-\frac{2\alpha}{1-\alpha}\ln\frac{1}{1-\varepsilon}\right), \\
	B&\equiv\sup_{\alpha\in(0,1)}\left(\pz{I}_\alpha^\downarrow\fleft(\rho'_{S'M'}\middle\|\gamma'_{S'}\fright)-\pz{I}_{2-\alpha}^\downarrow\fleft(\rho_{SM}\middle\|\gamma_S\fright)-\frac{2}{1-\alpha}\ln\frac{1}{1-\varepsilon}\right).
\end{align}
\end{proposition}

\begin{proof}[Proof of the upper bound]
It follows from the triangle inequality for single-shot work costs under QFGOs (Lemma~\ref{lem:triangle}) that
\begin{align}
	&W^\varepsilon\fleft(\fleft(\rho_{SM},\gamma_S\vphantom{\rho'_{S'M'},\gamma'_{S'}}\fright)\to\fleft(\rho'_{S'M'},\gamma'_{S'}\fright)\fright) \notag\\
	&\leq\inf_{\eta\in[0,\varepsilon]}\left(W^\eta\fleft(\fleft(\rho_{SM},\gamma_S\fright)\to\fleft(1,1\fright)\fright)+W^{\varepsilon-\eta}\fleft(\fleft(1,1\fright)\to\fleft(\rho'_{S'M'},\gamma'_{S'}\fright)\fright)\right) \\
	&=\inf_{\eta\in[0,\varepsilon]}\left(W_\abb{form}^{\varepsilon-\eta}\fleft(\rho'_{S'M'},\gamma'_{S'}\fright)-W_\abb{extr}^\eta\fleft(\rho_{SM},\gamma_S\fright)\right) \label{pf:upper-1}\\
	&=\beta^{-1}\inf_{\eta\in[0,\varepsilon]}\left(I_{\max}^{\uparrow,\varepsilon-\eta}\fleft(\rho'_{S'M'}\middle\|\gamma'_{S'}\fright)-I_{\min}^{\downarrow,\eta}\fleft(\rho_{SM}\middle\|\gamma_S\fright)\right). \label{pf:upper-2}
\end{align}
Here Eq.~\eqref{pf:upper-1} follows from the definitions of the single-shot work of formation and extractable work under QFGOs (Definition~\ref{def:work}); Eq.~\eqref{pf:upper-2} follows from Theorems~\ref{thm:formation-supp} and \ref{thm:extraction-supp}.
\end{proof}

\begin{proof}[Proof of the lower bound]
Let $\ch{N}_{S_0SM\to S_1S'M'}$ be a QFGO such that
\begin{align}
	\delta\fleft(\ch{N}_{S_0SM\to S_1S'M'}\fleft[\op{0}{0}_{S_0}\otimes\rho_{SM}\fright],\op{0}{0}_{S_1}\otimes\rho'_{S'M'}\fright)&\leq\varepsilon. \label{pf:lower-1}
\end{align}
This implies that~\cite[Theorem~6.14]{khatri2024PrinciplesQuantumCommunication}
\begin{align}
	F\fleft(\ch{N}_{S_0SM\to S_1S'M'}\fleft[\op{0}{0}_{S_0}\otimes\rho_{SM}\fright],\op{0}{0}_{S_1}\otimes\rho'_{S'M'}\fright)&\geq\left(1-\varepsilon\right)^2, \label{pf:lower-2}
\end{align}
where $F(\rho,\sigma)\coloneq\lVert\rho^\frac{1}{2}\sigma^\frac{1}{2}\rVert_1^2$ is the fidelity between $\rho$ and $\sigma$.  For $\alpha\in(1/2,1)$, it follows from the additivity of the sandwiched Rényi GMI~\cite[Lemma~7]{hayashi2016CorrelationDetectionOperational} that
\begin{align}
	\sw{I}_\frac{\alpha}{2\alpha-1}^\downarrow\fleft(\rho_{SM}\middle\|\gamma_S\fright)&=\sw{I}_\frac{\alpha}{2\alpha-1}^\downarrow\fleft(\op{0}{0}_{S_0}\otimes\rho_{SM}\middle\|\tfrac{\1_{S_0}}{\left\lvert S_0\right\rvert}\otimes\gamma_S\fright)-\ln\left\lvert S_0\right\rvert \\
	&\geq\sw{I}_\frac{\alpha}{2\alpha-1}^\downarrow\fleft(\ch{N}_{S_0SM\to S_1S'M'}\fleft[\op{0}{0}_{S_0}\otimes\rho_{SM}\fright]\middle\|\tfrac{\1_{S_1}}{\left\lvert S_1\right\rvert}\otimes\gamma'_{S'}\fright)-\ln\left\lvert S_0\right\rvert \label{pf:lower-3}\\
	&=\inf_{\sigma'_{M'}\in\sta}\sw{D}_\frac{\alpha}{2\alpha-1}\fleft(\ch{N}_{S_0SM\to S_1S'M'}\fleft[\op{0}{0}_{S_0}\otimes\rho_{SM}\fright]\middle\|\tfrac{\1_{S_1}}{\left\lvert S_1\right\rvert}\otimes\gamma'_{S'}\otimes\sigma'_{M'}\fright)-\ln\left\lvert S_0\right\rvert \label{pf:lower-4}\\
	&\geq\inf_{\sigma'_{M'}\in\sta}\sw{D}_\alpha\fleft(\op{0}{0}_{S_1}\otimes\rho'_{S'M'}\middle\|\tfrac{\1_{S_1}}{\left\lvert S_1\right\rvert}\otimes\gamma'_{S'}\otimes\sigma'_{M'}\fright)-\frac{2\alpha}{1-\alpha}\ln\frac{1}{1-\varepsilon}-\ln\left\lvert S_0\right\rvert \label{pf:lower-5}\\
	&=\sw{I}_\alpha^\downarrow\fleft(\op{0}{0}_{S_1}\otimes\rho'_{S'M'}\middle\|\tfrac{\1_{S_1}}{\left\lvert S_1\right\rvert}\otimes\gamma'_{S'}\fright)-\frac{2\alpha}{1-\alpha}\ln\frac{1}{1-\varepsilon}-\ln\left\lvert S_0\right\rvert \label{pf:lower-6}\\
	&=\ln\left\lvert S_1\right\rvert+\sw{I}_\alpha^\downarrow\fleft(\rho'_{S'M'}\middle\|\gamma'_{S'}\fright)-\frac{2\alpha}{1-\alpha}\ln\frac{1}{1-\varepsilon}-\ln\left\lvert S_0\right\rvert. \label{pf:lower-7}
\end{align}
Here Eq.~\eqref{pf:lower-3} follows from the monotonicity of the sandwiched Rényi GMI under QFGOs~\cite[Proposition~8]{ji2026QuantumConditionalRelative}; Eqs.~\eqref{pf:lower-4} and \eqref{pf:lower-6} follow from the definition of the sandwiched Rényi GMI [Eq.~\eqref{eq:sandwiched-information}]; Eq.~\eqref{pf:lower-5} follows from Ref.~\cite[Lemma~1]{wang2019ResourceTheoryAsymmetric} and Eq.~\eqref{pf:lower-2}; Eq.~\eqref{pf:lower-7} follows from the additivity of the sandwiched Rényi GMI~\cite[Lemma~7]{hayashi2016CorrelationDetectionOperational}.  Rearranging the terms of Eq.~\eqref{pf:lower-7}, we have that
\begin{align}
	\ln\left\lvert S_0\right\rvert-\ln\left\lvert S_1\right\rvert&\geq\sw{I}_\alpha^\downarrow\fleft(\rho'_{S'M'}\middle\|\gamma'_{S'}\fright)-\sw{I}_\frac{\alpha}{2\alpha-1}^\downarrow\fleft(\rho_{SM}\middle\|\gamma_S\fright)-\frac{2\alpha}{1-\alpha}\ln\frac{1}{1-\varepsilon}. \label{pf:lower-8}
\end{align}
Taking the infimum on the left-hand side of Eq.~\eqref{pf:lower-8} over every battery $S_0$, battery $S_1$, and QFGO $\ch{N}_{S_0SM\to S_1S'M'}$ satisfying Eq.~\eqref{pf:lower-1} and taking the supremum on the right-hand side of Eq.~\eqref{pf:lower-8} over every $\alpha\in(1/2,1)$, it follows from the definition of the single-shot work cost under QFGOs (Definition~\ref{def:work}) that
\begin{align}
	W^\varepsilon\fleft(\fleft(\rho_{SM},\gamma_S\vphantom{\rho'_{S'M'},\gamma'_{S'}}\fright)\to\fleft(\rho'_{S'M'},\gamma'_{S'}\fright)\fright)&\geq\beta^{-1}\sup_{\alpha\in(\frac{1}{2},1)}\left(\sw{I}_\alpha^\downarrow\fleft(\rho'_{S'M'}\middle\|\gamma'_{S'}\fright)-\sw{I}_\frac{\alpha}{2\alpha-1}^\downarrow\fleft(\rho_{SM}\middle\|\gamma_S\fright)-\frac{2\alpha}{1-\alpha}\ln\frac{1}{1-\varepsilon}\right). \label{pf:lower-9}
\end{align}
Applying the same reasoning to the Petz-type Rényi GMI [Eq.~\eqref{eq:petz-information}], for $\alpha\in(0,1)$, we have that
\begin{align}
	\pz{I}_{2-\alpha}^\downarrow\fleft(\rho_{SM}\middle\|\gamma_S\fright)&=\pz{I}_{2-\alpha}^\downarrow\fleft(\op{0}{0}_{S_0}\otimes\rho_{SM}\middle\|\tfrac{\1_{S_0}}{\left\lvert S_0\right\rvert}\otimes\gamma_S\fright)-\ln\left\lvert S_0\right\rvert \\
	&\geq\pz{I}_{2-\alpha}^\downarrow\fleft(\ch{N}_{S_0SM\to S_1S'M'}\fleft[\op{0}{0}_{S_0}\otimes\rho_{SM}\fright]\middle\|\tfrac{\1_{S_1}}{\left\lvert S_1\right\rvert}\otimes\gamma'_{S'}\fright)-\ln\left\lvert S_0\right\rvert \\
	&=\inf_{\sigma'_{M'}\in\sta}\pz{D}_{2-\alpha}\fleft(\ch{N}_{S_0SM\to S_1S'M'}\fleft[\op{0}{0}_{S_0}\otimes\rho_{SM}\fright]\middle\|\tfrac{\1_{S_1}}{\left\lvert S_1\right\rvert}\otimes\gamma'_{S'}\otimes\sigma'_{M'}\fright)-\ln\left\lvert S_0\right\rvert \\
	&\geq\inf_{\sigma'_{M'}\in\sta}\pz{D}_\alpha\fleft(\op{0}{0}_{S_1}\otimes\rho'_{S'M'}\middle\|\tfrac{\1_{S_1}}{\left\lvert S_1\right\rvert}\otimes\gamma'_{S'}\otimes\sigma'_{M'}\fright)-\frac{2}{1-\alpha}\ln\frac{1}{1-\varepsilon}-\ln\left\lvert S_0\right\rvert \label{pf:lower-10}\\
	&=\pz{I}_\alpha^\downarrow\fleft(\op{0}{0}_{S_1}\otimes\rho'_{S'M'}\middle\|\tfrac{\1_{S_1}}{\left\lvert S_1\right\rvert}\otimes\gamma'_{S'}\fright)-\frac{2}{1-\alpha}\ln\frac{1}{1-\varepsilon}-\ln\left\lvert S_0\right\rvert \\
	&=\ln\left\lvert S_1\right\rvert+\pz{I}_\alpha^\downarrow\fleft(\rho'_{S'M'}\middle\|\gamma'_{S'}\fright)-\frac{2}{1-\alpha}\ln\frac{1}{1-\varepsilon}-\ln\left\lvert S_0\right\rvert. \label{pf:lower-11}
\end{align}
Here Eq.~\eqref{pf:lower-10} follows from Ref.~\cite[Lemma~3]{wang2019ResourceTheoryAsymmetric} and Eq.~\eqref{pf:lower-1}.  Rearranging the terms of Eq.~\eqref{pf:lower-11}, we have that
\begin{align}
	\ln\left\lvert S_0\right\rvert-\ln\left\lvert S_1\right\rvert&\geq\pz{I}_\alpha^\downarrow\fleft(\rho'_{S'M'}\middle\|\gamma'_{S'}\fright)-\pz{I}_{2-\alpha}^\downarrow\fleft(\rho_{SM}\middle\|\gamma_S\fright)-\frac{2}{1-\alpha}\ln\frac{1}{1-\varepsilon}. \label{pf:lower-12}
\end{align}
Taking the infimum on the left-hand side of Eq.~\eqref{pf:lower-12} over every battery $S_0$, battery $S_1$, and QFGO $\ch{N}_{S_0SM\to S_1S'M'}$ satisfying Eq.~\eqref{pf:lower-1} and taking the supremum on the right-hand side of Eq.~\eqref{pf:lower-12} over every $\alpha\in(0,1)$, it follows from the definition of the single-shot work cost under QFGOs (Definition~\ref{def:work}) that
\begin{align}
	W^\varepsilon\fleft(\fleft(\rho_{SM},\gamma_S\vphantom{\rho'_{S'M'},\gamma'_{S'}}\fright)\to\fleft(\rho'_{S'M'},\gamma'_{S'}\fright)\fright)&\geq\beta^{-1}\sup_{\alpha\in(0,1)}\left(\pz{I}_\alpha^\downarrow\fleft(\rho'_{S'M'}\middle\|\gamma'_{S'}\fright)-\pz{I}_{2-\alpha}^\downarrow\fleft(\rho_{SM}\middle\|\gamma_S\fright)-\frac{2}{1-\alpha}\ln\frac{1}{1-\varepsilon}\right). \label{pf:lower-13}
\end{align}
Combining Eqs.~\eqref{pf:lower-9} and \eqref{pf:lower-13} leads to the desired lower bound.
\end{proof}

\begin{theorem}[Asymptotic work cost under QFGOs; restatement of Theorem~\ref{thm:asymptotic}]
\label{thm:asymptotic-supp}
For an error parameter $\varepsilon\in(0,1)$, an initial state $(\rho_{SM},\gamma_S)$, and a target state $(\rho'_{S'M'},\gamma'_{S'})$, the asymptotic work cost under QFGOs is given by
\begin{align}
	\lim_{n\to\infty}\frac{1}{n}W^\varepsilon\fleft(\fleft(\rho_{SM}^{\otimes n},\gamma_S^{\otimes n}\fright)\to\fleft(\rho_{S'M'}^{\prime\,\otimes n},\gamma_{S'}^{\prime\,\otimes n}\fright)\fright)&=\beta^{-1}\left(I\fleft(\rho'_{S'M'}\middle\|\gamma'_{S'}\fright)-I\fleft(\rho_{SM}\middle\|\gamma_S\fright)\right).
\end{align}
\end{theorem}

\begin{proof}
It follows from the upper bound in Proposition~\ref{prop:bounds-supp} that
\begin{align}
	&\limsup_{n\to\infty}\frac{1}{n}W^\varepsilon\fleft(\fleft(\rho_{SM}^{\otimes n},\gamma_S^{\otimes n}\fright)\to\fleft(\rho_{S'M'}^{\prime\,\otimes n},\gamma_{S'}^{\prime\,\otimes n}\fright)\fright) \notag\\
	&\leq\beta^{-1}\limsup_{n\to\infty}\inf_{\eta\in[0,\varepsilon]}\frac{1}{n}\left(I_{\max}^{\uparrow,\varepsilon-\eta}\fleft(\rho_{S'M'}^{\prime\,\otimes n}\middle\|\gamma_{S'}^{\prime\,\otimes n}\fright)-I_{\min}^{\downarrow,\eta}\fleft(\rho_{SM}^{\otimes n}\middle\|\gamma_S^{\otimes n}\fright)\right) \\
	&\leq\beta^{-1}\inf_{\eta\in(0,\varepsilon)}\limsup_{n\to\infty}\frac{1}{n}\left(I_{\max}^{\uparrow,\varepsilon-\eta}\fleft(\rho_{S'M'}^{\prime\,\otimes n}\middle\|\gamma_{S'}^{\prime\,\otimes n}\fright)-I_{\min}^{\downarrow,\eta}\fleft(\rho_{SM}^{\otimes n}\middle\|\gamma_S^{\otimes n}\fright)\right) \\
	&=\beta^{-1}\left(I\fleft(\rho'_{S'M'}\middle\|\gamma'_{S'}\fright)-I\fleft(\rho_{SM}\middle\|\gamma_S\fright)\right). \label{pf:asymptotic-1}
\end{align}
Here Eq.~\eqref{pf:asymptotic-1} follows from the asymptotic equipartition properties of the smoothed max- and min-GMIs [Proposition~\ref{prop:properties-supp}(iii)].  Conversely, it follows from the lower bound in Proposition~\ref{prop:bounds-supp} that
\begin{align}
	&\liminf_{n\to\infty}\frac{1}{n}W^\varepsilon\fleft(\fleft(\rho_{SM}^{\otimes n},\gamma_S^{\otimes n}\fright)\to\fleft(\rho_{S'M'}^{\prime\,\otimes n},\gamma_{S'}^{\prime\,\otimes n}\fright)\fright) \notag\\
	&\geq\beta^{-1}\liminf_{n\to\infty}\sup_{\alpha\in(\frac{1}{2},1)}\frac{1}{n}\left(\sw{I}_\alpha^\downarrow\fleft(\rho_{S'M'}^{\prime\,\otimes n}\middle\|\gamma_{S'}^{\prime\,\otimes n}\fright)-\sw{I}_\frac{\alpha}{2\alpha-1}^\downarrow\fleft(\rho_{SM}^{\otimes n}\middle\|\gamma_S^{\otimes n}\fright)-\frac{2\alpha}{1-\alpha}\ln\frac{1}{1-\varepsilon}\right) \\
	&\geq\beta^{-1}\sup_{\alpha\in(\frac{1}{2},1)}\liminf_{n\to\infty}\frac{1}{n}\left(\sw{I}_\alpha^\downarrow\fleft(\rho_{S'M'}^{\prime\,\otimes n}\middle\|\gamma_{S'}^{\prime\,\otimes n}\fright)-\sw{I}_\frac{\alpha}{2\alpha-1}^\downarrow\fleft(\rho_{SM}^{\otimes n}\middle\|\gamma_S^{\otimes n}\fright)-\frac{2\alpha}{1-\alpha}\ln\frac{1}{1-\varepsilon}\right) \\
	&=\beta^{-1}\sup_{\alpha\in(\frac{1}{2},1)}\left(\sw{I}_\alpha^\downarrow\fleft(\rho'_{S'M'}\middle\|\gamma'_{S'}\fright)-\sw{I}_\frac{\alpha}{2\alpha-1}^\downarrow\fleft(\rho_{SM}\middle\|\gamma_S\fright)\right) \label{pf:asymptotic-2}\\
	&=\beta^{-1}\sup_{\alpha\in(\frac{1}{2},1)}\left(\inf_{\sigma'_{M'}\in\sta}\sw{D}_\alpha\fleft(\rho'_{S'M'}\middle\|\gamma'_{S'}\otimes\sigma'_{M'}\fright)-\inf_{\sigma_M\in\sta}\sw{D}_\frac{\alpha}{2\alpha-1}\fleft(\rho_{SM}\middle\|\gamma_S\otimes\sigma_M\fright)\right) \label{pf:asymptotic-3}\\
	&\geq\beta^{-1}\sup_{\alpha\in(\frac{1}{2},1)}\inf_{\sigma'_{M'}\in\sta}\left(\sw{D}_\alpha\fleft(\rho'_{S'M'}\middle\|\gamma'_{S'}\otimes\sigma'_{M'}\fright)-\sw{D}_\frac{\alpha}{2\alpha-1}\fleft(\rho_{SM}\middle\|\gamma_S\otimes\rho_M\fright)\right) \\
	&=\beta^{-1}\inf_{\sigma'_{M'}\in\sta}\sup_{\alpha\in(\frac{1}{2},1)}\left(\sw{D}_\alpha\fleft(\rho'_{S'M'}\middle\|\gamma'_{S'}\otimes\sigma'_{M'}\fright)-\sw{D}_\frac{\alpha}{2\alpha-1}\fleft(\rho_{SM}\middle\|\gamma_S\otimes\rho_M\fright)\right) \label{pf:asymptotic-4}\\
	&=\beta^{-1}\inf_{\sigma'_{M'}\in\sta}\left(D\fleft(\rho'_{S'M'}\middle\|\gamma'_{S'}\otimes\sigma'_{M'}\fright)-D\fleft(\rho_{SM}\middle\|\gamma_S\otimes\rho_M\fright)\right) \label{pf:asymptotic-5}\\
	&=\beta^{-1}\left(I\fleft(\rho'_{S'M'}\middle\|\gamma'_{S'}\fright)-I\fleft(\rho_{SM}\middle\|\gamma_S\fright)\right). \label{pf:asymptotic-6}
\end{align}
Here Eq.~\eqref{pf:asymptotic-2} follows from the additivity of the sandwiched Rényi GMI~\cite[Lemma~7]{hayashi2016CorrelationDetectionOperational}; Eq.~\eqref{pf:asymptotic-3} follows from the definition of the sandwiched Rényi GMI [Eq.~\eqref{eq:sandwiched-information}]; Eq.~\eqref{pf:asymptotic-4} follows from the Mosonyi--Hiai minimax theorem~\cite[Corollary~A.2]{mosonyi2011QuantumRenyiRelative}, which applies because the sandwiched Rényi relative entropy is monotonically nondecreasing in $\alpha$~\cite{muller-lennert2013QuantumRenyiEntropies} and lower semicontinuous in its second argument; Eq.~\eqref{pf:asymptotic-5} follows from the nondecreasing monotonicity of the sandwiched Rényi relative entropy in $\alpha$ and its $\alpha\nearrow1$ limit given by the Umegaki relative entropy~\cite{muller-lennert2013QuantumRenyiEntropies}; Eq.~\eqref{pf:asymptotic-6} follows from the definition of the Umegaki GMI [Eqs.~\eqref{eq:unoptimized-umegaki-information} and \eqref{eq:optimized-umegaki-information}].  Combining Eqs.~\eqref{pf:asymptotic-1} and \eqref{pf:asymptotic-6} leads to the desired statement.
\end{proof}

\section{Special case: Trivial Hamiltonians}
\label{sec:special}

An important special case is when the Hamiltonians of the systems are trivial.  For such a system $S$, we have $\h{H}_S=0$ and thus $\gamma_S=\1_S/\lvert S\rvert$.  In this case, our resource theory of quantum conditional athermality reduces to that of \emph{quantum conditional nonuniformity}, where the free operations are QFGOs with uniform Gibbs states.  Specializations of all definitions and results in previous sections are straightforward, and we lay them out selectively as follows.

For $\rho_{SM}\in\sta$, the \emph{(unoptimized) conditional min-entropy} is defined as~\cite{renner2005SecurityQuantumKey}
\begin{align}
	H_{\min}^\downarrow\fleft(S\middle|M\fright)_\rho&\coloneq-I_{\max}^\uparrow\fleft(\rho_{SM}\middle\|\1_S\fright) \\
	&=-\ln\left\lVert\rho_M^{-\frac{1}{2}}\rho_{SM}\rho_M^{-\frac{1}{2}}\right\rVert_\infty. \label{eq:conditional-min}
\end{align}
The \emph{(optimized) (Petz-type) conditional max-entropy} is defined as~\cite{renner2005SecurityQuantumKey}
\begin{align}
	H_{\max}^\uparrow\fleft(S\middle|M\fright)_\rho&\coloneq-I_{\min}^\downarrow\fleft(\rho_{SM}\middle\|\1_S\fright) \\
	&=\ln\left\lVert\tr_S\fleft[\rho_{SM}^0\fright]\right\rVert_\infty. \label{eq:conditional-max}
\end{align}
The \emph{conditional von Neumann entropy} is defined as~\cite{cerf1997NegativeEntropyInformation}
\begin{align}
	H\fleft(S\middle|M\fright)_\rho&\coloneq-I\fleft(\rho_{SM}\middle\|\1_S\fright) \\
	&=\tr\fleft[\rho_M\ln\rho_M\fright]-\tr\fleft[\rho_{SM}\ln\rho_{SM}\fright].
\end{align}

For $\varepsilon\in[0,1]$, we define the \emph{smoothed (unoptimized) conditional min-entropy} as
\begin{align}
	H_{\min}^{\downarrow,\varepsilon}\fleft(S\middle|M\fright)_\rho&\coloneq-I_{\max}^{\uparrow,\varepsilon}\fleft(\rho_{SM}\middle\|\1_S\fright) \\
	&=\sup_{\substack{t\in\spa{R},\\ \tau_{SM}\in\sub, \\ \sigma_M\in\sta}}\left\{-\ln t\colon\tau_{SM}\leq t\1_S\otimes\sigma_M,\;\tau_M\leq\sigma_M,\;\delta\fleft(\tau_{SM},\rho_{SM}\fright)\leq\varepsilon\right\} \label{eq:smoothed-conditional-min}\\
	&=\sup_{\substack{\hat{S}, \\ \omega_{\hat{S}M}\in\sta}}\left\{H_{\min}^\downarrow\fleft(\hat{S}\middle|M\fright)_\omega\colon\left\lvert\hat{S}\right\rvert\geq\left\lvert S\right\rvert,\;\delta\fleft(\omega_{\hat{S}M},\ch{V}_{S\to\hat{S}}\fleft[\rho_{SM}\fright]\fright)\leq\varepsilon\right\},
\end{align}
where $\ch{V}_{S\to\hat{S}}$ is an arbitrary isometric channel.  We define the \emph{(operator-) smoothed (optimized) (Petz-type) conditional max-entropy} as
\begin{align}
	H_{\max}^{\uparrow,\varepsilon}\fleft(S\middle|M\fright)_\rho&\coloneq-I_{\min}^{\downarrow,\varepsilon}\fleft(\rho_{SM}\middle\|\1_S\fright) \\
	&=\inf_{\Lambda_{SM}\in\psd}\left\{\ln\left\lVert\tr_S\fleft[\Lambda_{SM}\fright]\right\rVert_\infty\colon\tr\fleft[\Lambda_{SM}\rho_{SM}\fright]\geq1-\varepsilon,\;0\leq\Lambda_{SM}\leq\1_{SM}\right\} \label{eq:smoothed-conditional-max}\\
	&=\sup_{\sigma_M\in\sta}-D_{\min}^\varepsilon\fleft(\rho_{SM}\middle\|\1_S\otimes\sigma_M\fright).
\end{align}
Note that the smoothed conditional min- and max-entropies defined above are different from quantities under similar names or notations in the literature~\cite{renner2005SecurityQuantumKey,konig2009OperationalMeaningMin,tomamichel2009FullyQuantumAsymptotic,tomamichel2010DualitySmoothMin,tomamichel2011LeftoverHashingQuantum,anshu2020PartiallySmoothedInformation} due to a different approach to smoothing.

Following Theorems~\ref{thm:formation-supp} and \ref{thm:extraction-supp}, the smoothed conditional min- and max-entropies acquire precise operational interpretations in microscopic thermodynamics.  They capture the single-shot work costs of writing and erasing data to or from a system with quantum feedback, with the system initialized in or reset to a blank pure state.

\begin{corollary}[Work costs of data writing and erasure under QFGOs]
\label{cor:work}
For an error parameter $\varepsilon\in[0,1]$ and a state $\rho_{SM}$, the single-shot work cost of data writing under QFGOs and its asymptotic limit are given by
\begin{align}
	W^\varepsilon\fleft(\fleft(\op{0}{0}_S,\tfrac{\1_S}{\left\lvert S\right\rvert}\fright)\to\fleft(\rho_{SM},\tfrac{\1_S}{\left\lvert S\right\rvert}\fright)\fright)&=-\beta^{-1}H_{\min}^{\downarrow,\varepsilon}\fleft(S\middle|M\fright)_\rho, \label{eq:writing}\\
	\lim_{n\to\infty}\frac{1}{n}W^\varepsilon\fleft(\fleft(\op{0}{0}_{S^n},\tfrac{\1_{S^n}}{\left\lvert S\right\rvert^n}\fright)\to\fleft(\rho_{SM}^{\otimes n},\tfrac{\1_{S^n}}{\left\lvert S\right\rvert^n}\fright)\fright)&=-\beta^{-1}H\fleft(S\middle|M\fright)_\rho.
\end{align}
The single-shot work cost of data erasure under QFGOs and its asymptotic limit are given by
\begin{align}
	W^\varepsilon\fleft(\fleft(\rho_{SM},\tfrac{\1_S}{\left\lvert S\right\rvert}\fright)\to\fleft(\op{0}{0}_S,\tfrac{\1_S}{\left\lvert S\right\rvert}\fright)\fright)&=\beta^{-1}H_{\max}^{\uparrow,\varepsilon}\fleft(S\middle|M\fright)_\rho, \label{eq:erasure}\\
	\lim_{n\to\infty}\frac{1}{n}W^\varepsilon\fleft(\fleft(\rho_{SM}^{\otimes n},\tfrac{\1_{S^n}}{\left\lvert S\right\rvert^n}\fright)\to\fleft(\op{0}{0}_{S^n},\tfrac{\1_{S^n}}{\left\lvert S\right\rvert^n}\fright)\fright)&=\beta^{-1}H\fleft(S\middle|M\fright)_\rho.
\end{align}
\end{corollary}

Equation~\eqref{eq:erasure} can be viewed as a generalization of Landauer's principle~\cite{landauer1961IrreversibilityHeatGeneration,esposito2011SecondLawLandauer,reeb2014ImprovedLandauerPrinciple} to the single-shot regime with quantum feedback.  The fact that the smoothed conditional max-entropy and the conditional von Neumann entropy can be negative for certain entangled states implies that data erasure can be accompanied by cooling instead of heating if the controller possesses quantum side information and feeds it back coherently.  A similar phenomenon has been discovered in Ref.~\cite{delrio2011ThermodynamicMeaningNegative} under different operational assumptions.

An interesting technical insight resulting from Corollary~\ref{cor:work} is the following entropic inequality, which may be of independent interest.

\begin{corollary}[Entropic inequality]
For two smoothing parameters $\varepsilon\in[0,1]$ and $\eta\in[0,\varepsilon]$ and a state $\rho_{SM}$,
\begin{align}
	H_{\min}^{\downarrow,\eta}\fleft(S\middle|M\fright)_\rho-H_{\max}^{\uparrow,\varepsilon-\eta}\fleft(S\middle|M\fright)_\rho&\leq\ln\frac{1}{1-\varepsilon}.
\end{align}
\end{corollary}

\begin{proof}
It follows from Corollary~\ref{cor:work} that
\begin{align}
	&H_{\min}^{\downarrow,\eta}\fleft(S\middle|M\fright)_\rho-H_{\max}^{\uparrow,\varepsilon-\eta}\fleft(S\middle|M\fright)_\rho \notag\\
	&=-\beta\left(W^\eta\fleft(\fleft(\op{0}{0}_S,\tfrac{\1_S}{\left\lvert S\right\rvert}\fright)\to\fleft(\rho_{SM},\tfrac{\1_S}{\left\lvert S\right\rvert}\fright)\fright)+W^{\varepsilon-\eta}\fleft(\fleft(\rho_{SM},\tfrac{\1_S}{\left\lvert S\right\rvert}\fright)\to\fleft(\op{0}{0}_S,\tfrac{\1_S}{\left\lvert S\right\rvert}\fright)\fright)\right) \\
	&\leq-\beta W^\varepsilon\fleft(\fleft(\op{0}{0}_S,\tfrac{\1_S}{\left\lvert S\right\rvert}\fright)\to\fleft(\op{0}{0}_S,\tfrac{\1_S}{\left\lvert S\right\rvert}\fright)\fright) \label{pf:inequality-1}\\
	&=\ln\frac{1}{1-\varepsilon}. \label{pf:inequality-2}
\end{align}
Here Eq.~\eqref{pf:inequality-1} follows from the triangle inequality for single-shot work costs under QFGOs (Lemma~\ref{lem:triangle}); Eq.~\eqref{pf:inequality-2} follows from Corollary~\ref{cor:work} and the definition of the smoothed conditional max-entropy [Eq.~\eqref{eq:smoothed-conditional-max}].
\end{proof}

\section{Examples}
\label{sec:example}

We provide two concrete examples wherein the ultimate thermodynamic limit can be achieved by a simple QFGO in a single shot but not by any scheme with classical feedback even asymptotically.

\begin{example}[Work extraction from a classical system with quantum feedback]
Consider work extraction from a classical system $S$ in the presence of side information stored in a quantum memory $M$, with $S$ locally in thermal equilibrium.  The initial state is in general given by $(\rho_{SM},\rho_S)$, where
\begin{align}
	\rho_{SM}&=\sum_{s}p_s\op{s}{s}_S\otimes\rho_{s,M},
\end{align}
with $((p_s,\rho_s))_s$ an ensemble of states.  The side information here can be imagined to have been acquired by measuring $S$ (which was in thermal equilibrium) but with the record of the measurement outcome corrupted by a noisy classical-to-quantum channel.  By Theorem~\ref{thm:extraction-supp}, for every $\varepsilon\in(0,1)$, the asymptotic extractable work under QFGOs is given by
\begin{align}
	\lim_{n\to\infty}\frac{1}{n}W_\abb{extr}^\varepsilon\fleft(\rho_{SM}^{\otimes n},\rho_S^{\otimes n}\fright)&=\beta^{-1}I\fleft(S;M\fright)_\rho \\
	&\eqcolon\beta^{-1}\chi\fleft(\fleft(\fleft(p_s,\rho_s\fright)\fright)_s\fright),
\end{align}
where $\chi(((p_s,\rho_s))_s)$ is the Holevo quantity~\cite{holevo1973BoundsQuantityInformation} for $((p_s,\rho_s))_s$.  Whereas if the controller is only capable of classical feedback, then the side information has to be measured before being fed back, and thus the asymptotic extractable work is upper bounded by
\begin{align}
	\beta^{-1}\sup_{\ch{M}_{M\to X}}I\fleft(S;X\fright)_{\ch{M}_{M\to X}\fleft[\rho_{SM}\fright]}&\eqcolon\beta^{-1}I_\abb{acc}\fleft(\fleft(\fleft(p_s,\rho_s\fright)\fright)_s\fright), \label{eq:accessible}
\end{align}
where the supremization is over every measurement channel $\ch{M}_{M\to X}$ and $I_\abb{acc}(((p_s,\rho_s))_s)$ is the accessible information~\cite{holevo1973BoundsQuantityInformation} for $((p_s,\rho_s))_s$.  This upper bound holds even if the measurement acts collectively on many copies of $M$, since the accessible information is additive for tensor products of ensembles~\cite{holevo1973InformationtheoreticalAspectsQuantum}.  It is well known that
\begin{align}
	\chi\fleft(\fleft(\fleft(p_s,\rho_s\fright)\fright)_s\fright)&\geq I_\abb{acc}\fleft(\fleft(\fleft(p_s,\rho_s\fright)\fright)_s\fright),
\end{align}
with the inequality being strict for certain ensembles.  This demonstrates a strict advantage of quantum over classical feedback in work extraction even without entanglement.  

An instance of an ensemble that makes this inequality strict is $\en{E}=((1/4,\op{\hat{\sigma}_s\psi}{\hat{\sigma}_s\psi}))_{s\in\{0,1,2,3\}}$, where
\begin{align}
	\op{\psi}{\psi}&=\frac{1}{2}\left(\hat{\sigma}_0+\sum_{t\in\{1,2,3\}}\frac{\hat{\sigma}_t}{\sqrt{3}}\right)
\end{align}
is a pure qubit state and $\hat{\sigma}_s$ for $s\in\{0,1,2,3\}$ are the Pauli operators.  As presented in the main text, for the initial state
\begin{align}
	\hat{\rho}_{SM}&=\sum_{s\in\{0,1,2,3\}}\frac{1}{4}\op{s}{s}_S\otimes\op{\hat{\sigma}_s\psi}{\hat{\sigma}_s\psi}_M, \label{pf:example-1}
\end{align}
the asymptotic extractable work under QFGOs is given by $\beta^{-1}\chi(\en{E})=\beta^{-1}\ln2$, which can furthermore be achieved in a single shot by the following scheme:
\begin{itemize}
	\item[(i)] transfer the information in $M$ identically to a link $L$ isomorphic to $M$ and feed it back to $S$;
	\item[(ii)] apply a Pauli unitary gate to $L$ controlled by $S$, thus restoring $L$ to the pure state $\op{\psi}{\psi}_L$; turn $L$ unitarily into a qubit battery $S_1$, with $\op{\psi}{\psi}_L$ mapped to $\op{0}{0}_{S_1}$; discard $S$.
\end{itemize}
Formally, this scheme is represented by the channel $\ch{N}_{SM\to S_1}\coloneq\ch{U}_{L\to S_1}\circ\hat{\ch{F}}_{SL\to L}\circ\ch{E}_{M\to L}$, where
\begin{align}
	\ch{E}_{M\to L}\fleft[\cdot\fright]&\coloneq\left(\cdot\right)_{M\to L}, \\
	\hat{\ch{F}}_{SL\to L}\fleft[\cdot\fright]&\coloneq\sum_{s\in\{0,1,2,3\}}\hat{\sigma}_{s,L}\tr_S\fleft[\op{s}{s}_S\left(\cdot\right)_{SL}\fright]\hat{\sigma}_{s,L}, \label{eq:hat}
\end{align}
and $\ch{U}_{L\to S_1}$ is a unitary channel mapping $\op{\psi}{\psi}_L$ to $\op{0}{0}_{S_1}$.  The channel $\ch{N}_{SM\to S_1}$ is a QFGO because $\ch{U}_{L\to S_1}\circ\hat{\ch{F}}_{SL\to L}$ respects Kelvin's principle [see Eq.~\eqref{eq:operation-2}]:
\begin{align}
	\left(\ch{U}_{L\to S_1}\circ\hat{\ch{F}}_{SL\to L}\right)\fleft[\frac{\1_S}{4}\otimes\left(\cdot\right)_L\fright]&=\frac{1}{4}\sum_{s\in\{0,1,2,3\}}\ch{U}_{L\to S_1}\fleft[\hat{\sigma}_{s,L}\left(\cdot\right)_L\hat{\sigma}_{s,L}\fright] \\
	&=\frac{\1_{S_1}}{2}\otimes\tr_L\fleft[\cdot\fright].
\end{align}
It prepares the battery precisely in the desired pure state:
\begin{align}
	\ch{N}_{SM\to S_1}\fleft[\hat{\rho}_{SM}\fright]&=\sum_{s\in\{0,1,2,3\}}\ch{U}_{L\to S_1}\fleft[\hat{\sigma}_{s,L}\tr_S\fleft[\op{s}{s}_S\hat{\rho}_{SL}\fright]\hat{\sigma}_{s,L}\fright] \\
	&=\ch{U}_{L\to S_1}\fleft[\op{\psi}{\psi}_L\fright] \\
	&=\op{0}{0}_{S_1},
\end{align}
thereby extracting work $\beta^{-1}\ln\lvert S_1\rvert=\beta^{-1}\ln2$ in a single shot.  In contrast, as we will prove, the controller cannot extract more work than $\beta^{-1}I_\abb{acc}(\en{E})=\beta^{-1}\ln(4/3)<\beta^{-1}\ln2$, even asymptotically, with classical feedback.
\end{example}

\begin{proof}[Proof of $I_\abb{acc}(\en{E})=\ln(4/3)$]
Let $\vec{r}_s\equiv(r_{s,t})_{t\in\{1,2,3\}}\in\spa{R}^3$ be the Bloch vector corresponding to the qubit state $\op{\hat{\sigma}_s\psi}{\hat{\sigma}_s\psi}$ in the sense that
\begin{align}
	\op{\hat{\sigma}_s\psi}{\hat{\sigma}_s\psi}&=\frac{1}{2}\left(\hat{\sigma}_0+\sum_{t\in\{1,2,3\}}r_{s,t}\hat{\sigma}_t\right). \label{pf:example-2}
\end{align}
Evidently,
\begin{align}
	\vec{r}_0&=\frac{1}{\sqrt{3}}\begin{pmatrix}
		1\\1\\1
	\end{pmatrix},\;\;
	\vec{r}_1=\frac{1}{\sqrt{3}}\begin{pmatrix}
		1\\-1\\-1
	\end{pmatrix},\;\;
	\vec{r}_2=\frac{1}{\sqrt{3}}\begin{pmatrix}
		-1\\1\\-1
	\end{pmatrix},\;\;
	\vec{r}_3=\frac{1}{\sqrt{3}}\begin{pmatrix}
		-1\\-1\\1
	\end{pmatrix}. \label{pf:example-3}
\end{align}
Let $\ch{M}_{M\to X}$ be a measurement channel:
\begin{align}
	\ch{M}_{M\to X}\fleft[\cdot\fright]&=\sum_{x}\tr\fleft[\Lambda_{x,M}\left(\cdot\right)_M\fright]\op{x}{x}_X, \label{pf:example-4}
\end{align}
where $(\Lambda_x)_x$ is a POVM.  Since $\lvert M\rvert=2$, for each $x$, there exists a coefficient $a_x\in(0,2]$ and a Bloch vector $\vec{m}_x\equiv(m_{x,t})_{t\in\{1,2,3\}}\in\spa{R}^3$ such that
\begin{align}
	\Lambda_x&=\frac{a_x}{2}\left(\hat{\sigma}_0+\sum_{t\in\{1,2,3\}}m_{x,t}\hat{\sigma}_t\right), \label{pf:example-5}\\
	\left\lvert\vec{m}_x\right\rvert^2&\leq1.
\end{align}
For $s\in\{0,1,2,3\}$ and $x$, it follows from Eqs.~\eqref{pf:example-2} and \eqref{pf:example-5} that
\begin{align}
	p_{s,x}&\equiv\frac{1}{4}\tr\fleft[\Lambda_x\op{\hat{\sigma}_s\psi}{\hat{\sigma}_s\psi}\fright] \label{pf:example-6}\\
	&=\frac{a_x}{8}\left(1+\vec{m}_x\cdot\vec{r}_s\right), \label{pf:example-7}
\end{align}
which implies that
\begin{align}
	p_x&\equiv\sum_{s\in\{0,1,2,3\}}p_{s,x} \\
	&=\sum_{s\in\{0,1,2,3\}}\frac{a_x}{8}\left(1+\vec{m}_x\cdot\vec{r}_s\right) \\
	&=\frac{a_x}{2}, \label{pf:example-8}
\end{align}
and
\begin{align}
	p_{s|x}&\equiv\frac{p_{s,x}}{p_x} \label{pf:example-9}\\
	&=\frac{1}{4}\left(1+\vec{m}_x\cdot\vec{r}_s\right). \label{pf:example-10}
\end{align}
Here Eq.~\eqref{pf:example-8} follows from the fact that $\sum_{s\in\{0,1,2,3\}}\vec{r}_s=\vec{0}$ due to Eq.~\eqref{pf:example-3}.  For each $x$, since the collision entropy lower bounds the Shannon entropy, we have that
\begin{align}
	-\sum_{s\in\{0,1,2,3\}}p_{s|x}\ln p_{s|x}&\geq-\ln\sum_{s\in\{0,1,2,3\}}p_{s|x}^2 \\
	&=-\ln\left(\frac{1}{4}+\frac{1}{12}\left\lvert\vec{m}_x\right\rvert^2\right) \label{pf:example-11}\\
	&\geq\ln3. \label{pf:example-12}
\end{align}
Here Eq.~\eqref{pf:example-11} follows from Eqs.~\eqref{pf:example-3} and \eqref{pf:example-10}.  We have that
\begin{align}
	I\fleft(S;X\fright)_{\ch{M}_{M\to X}\fleft[\hat{\rho}_{SM}\fright]}&=H\fleft(S\fright)_{\ch{M}_{M\to X}\fleft[\hat{\rho}_{SM}\fright]}-H\fleft(S\middle|X\fright)_{\ch{M}_{M\to X}\fleft[\hat{\rho}_{SM}\fright]} \\
	&=H\fleft(S\fright)_{\hat{\rho}}+\sum_{x}\sum_{s\in\{0,1,2,3\}}p_{s,x}\ln p_{s|x} \label{pf:example-13}\\
	&=\ln4+\sum_{x}p_x\sum_{s\in\{0,1,2,3\}}p_{s|x}\ln p_{s|x} \label{pf:example-14}\\
	&\leq\ln(4/3). \label{pf:example-15}
\end{align}
Here Eq.~\eqref{pf:example-13} follows from Eqs.~\eqref{pf:example-4} and \eqref{pf:example-6}; Eq.~\eqref{pf:example-14} follows from Eqs.~\eqref{pf:example-1} and \eqref{pf:example-9}; Eq.~\eqref{pf:example-15} follows from Eq.~\eqref{pf:example-12}.  This inequality is achieved by the measurement channel whose corresponding POVM is $((1/2)(\1-\op{\hat{\sigma}_s\psi}{\hat{\sigma}_s\psi}))_s$.  The desired statement follows from this fact and Eqs.~\eqref{eq:accessible} and \eqref{pf:example-15}.
\end{proof}

\begin{example}[Formation of entanglement with quantum feedback]
Consider the formation of the standard maximally entangled state $\Phi_{SM}=(1/2)\sum_{i,j\in\{0,1\}}\op{i}{j}_S\otimes\op{i}{j}_M$ between a qubit system $S$ with a trivial Hamiltonian and a qubit memory $M$.  By Theorem~\ref{thm:formation-supp}, for every $\varepsilon\in(0,1)$, the asymptotic work of formation under QFGOs is given by
\begin{align}
	\lim_{n\to\infty}\frac{1}{n}W_\abb{form}^\varepsilon\fleft(\Phi_{SM}^{\otimes n},\left(\1_S/2\right)^{\otimes n}\fright)&=\beta^{-1}I\fleft(S;M\fright)_\Phi \\
	&=\beta^{-1}\ln4,
\end{align}
which can furthermore be achieved in a single shot by the QFGO $\ch{N}_{S_0\to SM}\coloneq\hat{\ch{F}}_{S_0L\to S}\circ\ch{E}_{\emptyset\to LM}$, where
\begin{align}
	\ch{E}_{\emptyset\to LM}&\coloneq\Phi_{LM},
\end{align}
and $\hat{\ch{F}}_{S_0L\to S}$ is defined in Eq.~\eqref{eq:hat} with $\lvert S_0\rvert=4$.  It prepares the target state precisely as desired in a single shot:
\begin{align}
	\ch{N}_{S_0\to SM}\fleft[\op{0}{0}_{S_0}\fright]&=\sum_{s\in\{0,1,2,3\}}\hat{\sigma}_{s,S}\tr_{S_0}\fleft[\op{s}{s}_{S_0}\left(\op{0}{0}_{S_0}\otimes\Phi_{SM}\right)\fright]\hat{\sigma}_{s,S} \\
	&=\Phi_{SM},
\end{align}
while expending work $\beta^{-1}\ln\lvert S_0\rvert=\beta^{-1}\ln4$.  In contrast, such an entangled state cannot be prepared by any means with classical feedback even asymptotically, regardless of how much work is expended.
\end{example}

\end{document}